\documentclass[a4paper,UKenglish,cleveref, autoref, thm-restate
]{lipics-v2021}

\usepackage[utf8]{inputenc}

\usepackage{todonotes}
\usepackage{hyperref} 
\usepackage{amsmath,amsfonts,amssymb,amsthm,mathtools} 
\usepackage{cleveref}

\usepackage{complexity}
\usepackage{xspace}

\DeclareMathOperator{\wrep}{wrep}
\DeclareSymbolFont{Shuffle}{U}{shuffle}{m}{n}
\DeclareFontFamily{U}{shuffle}{}
\DeclareFontShape{U}{shuffle}{m}{n}{%
  <-8>shuffle7%
  <8->shuffle10%
}{}
\DeclareMathSymbol\shuffle{\mathbin}{Shuffle}{"001}
\DeclareMathSymbol\cshuffle{\mathbin}{Shuffle}{"002}

\newcommand{\emptyword}{\lambda}

\def\ta{\mathtt{a}}
\def\tb{\mathtt{b}}
\def\tc{\mathtt{c}}
\def\td{\mathtt{d}}

\DeclareMathOperator{\alphabet}{alph}

\newcommand{\N}{\mathbb{N}}

\newcommand{\cG}{\mathcal{G}}
\newcommand{\cO}{\mathcal{O}}
\renewcommand{\cL}{\mathcal{L}}
\renewcommand{\cP}{\mathcal{P}}

\renewcommand{\alph}{\alphabet}

\def\bico{\mathtt{\mathrm{bico}\text{\,-\,}}}

\newcommand{\longversion}[1]{#1}
\newcommand{\shortversion}[1]{}

\newcommand{\iffl}{\shortversion{iff}\longversion{if and only if}\xspace}

\longversion{\usepackage[appendix=inline]{apxproof}}
\shortversion{\usepackage{apxproof}}%
\shortversion{}

\newtheoremrep{theorem}{Theorem}[section] 
\newtheoremrep{proposition}[theorem]{Proposition}
\newtheoremrep{lemma}[theorem]{Lemma}
\newtheoremrep{claim}[theorem]{Claim}
\newtheoremrep{conjecture}[theorem]{Conjecture}
\newtheoremrep{corollary}[theorem]{Corollary}
\theoremstyle{definition}
\newtheoremrep{definition}[theorem]{Definition}
\theoremstyle{remark}
\newtheoremrep{example}[theorem]{Example}
\newtheoremrep{remark}[theorem]{Remark}

\crefalias{proposition}{proposition}

\newcommand{\todoscs}[1]{\todo[color=green!50,inline]{#1}}
\newcommand{\todoscsInPlace}[1]{\todo[color=green!50]{#1}}
\newcommand{\todohf}[1]{\todo[color=red!80,inline]{#1}}
\newcommand{\todohfInPlace}[1]{\todo[color=red!80]{#1}}
\newcommand{\todokm}[1]{\todo[color=violet!50,inline]{#1}}

\renewcommand{\iff}{\Leftrightarrow}
\renewcommand{\implies}{\Rightarrow}

\title{Generalized Word-Representable Graphs\shortversion{ I}:\\ A Formal Language Approach} 

\titlerunning{Generalized Word-Representable Graphs} 

\author{Zhidan Feng}
{Universit\"at Trier, Fachbereich IV, Informatikwissenschaften, Germany}
{}
{THTps://orcid.org/0000-0002-3364-5396}
{}
{} 
{}

\author{Henning Fernau}
{Universit\"at Trier, Fachbereich IV, Informatikwissenschaften, Germany \and \url{THTps://www.uni-trier.de/index.php?id=49861}}
{fernau@uni-trier.de}
{THTps://orcid.org/0000-0002-4444-3220}
{}

\author{Pamela Fleischmann}
{Kiel University, Germany}
{fpa@informatik.uni-kiel.de}
{THTps://orcid.org/0000-0002-1531-7970} 
{}

\author{Kevin Mann}
{Universit\"at Trier, Fachbereich IV, Informatikwissenschaften, Germany}
{mann@uni-trier.de}
{THTps://orcid.org/0000-0002-0880-2513} 
{}

\author{Silas Cato Sacher}
{Universit\"at Trier, Fachbereich IV, Informatikwissenschaften, Germany}
{sacher@informatik.uni-trier.de}
{THTps://orcid.org/0009-0004-6850-1298} 
{}

\authorrunning{Z. Feng, H. Fernau, P. Fleischmann, K. Mann, and S. C. Sacher} 

\Copyright{Zhidan Feng, Henning Fernau, Pamela Fleischmann, Kevin Mann, and Silas Cato Sacher} 



\category{} 

\relatedversion{} 

\acknowledgements{}

\longversion{\nolinenumbers} 

\EventEditors{John Q. Open and Joan R. Access}
\EventNoEds{2}
\shortversion{\EventLongTitle{42nd Conference on Very Important Topics (CVIT 2016)}}
\longversion{\EventLongTitle{ArXiVe Version 1.0 2024}}
\shortversion{\EventShortTitle{CVIT 2016}}\longversion{\EventShortTitle{V1.0, 2024}}
\EventAcronym{CVIT}
\EventYear{2016}
\EventDate{December 24--27, 2016}
\EventLocation{Little Whinging, United Kingdom}
\EventLogo{}
\SeriesVolume{42}
\ArticleNo{WRG}

\begin{document}

\maketitle

\begin{abstract}
The literature on word-representable graphs is quite rich, and a number of variations of the original definition have been proposed over the years.
In this paper, we are initiating a systematic study of such variations based on formal languages. In our framework, we can associate a graph class to each language over the binary alphabet $\{0,1\}$.
All graph classes that are language-representable in this sense are hereditary and enjoy further common properties. Besides word-representable graphs and, more generally, $1^k$- or $k$-11-representable  graphs, we can identify many more graph classes in our framework, like (co)bipartite graphs, (co)comparability graphs, to name a few.

It was already known that any graph is  $111$- or 2-11-representable.
However, when such representations are considered for storing graphs (as in graph databases),  $111$- or 2-11-representability bears the disadvantage of being significantly inferior to standard adjacency matrices or lists. We will prove that quite famous language like the palindrome language, the copy language or the language of Lyndon words, as well as their complements, can match the efficiency of standard graph representations in general.

The perspective of language theory allows us to prove very general results that hold for all graph classes that can be defined in this framework. This includes certain closure properties (e.g., all language-definable graph classes are  hereditary) as well as certain limitations (e.g., all language-representable graph classes contain graphs of arbitrarily large treewidth and graphs of arbitrarily large degeneracy, except a trivial case). As each language describes a graph class, we can also ask questions like: Given a language, say, by a context-free grammar, does the described graph class only contain graphs of treewidth at most~2? We show that this (and similar questions) are decidable.    

\longversion{We also present a systematic study of graph classes that can be represented by binary languages in which each letter occurs at most twice. Here, we find graph classes like interval graphs, permutation graphs, circle graphs, bipartite chain graphs, convex graphs, and threshold graphs, to name the most prominent ones.}
\end{abstract}



\newpage

\section{Introduction}
How to represent a graph? This question has come up and has been answered in many variations and has also led to the definition of many graph classes. For instance, geometric relations as intersection or containment are possible bases of definitions of interval, chordal or circle graphs, to give concrete examples. In the context of computer science, the study of implicit representations of graphs was initiated in~\cite{KanNaoRud92}, but for instance the whole research on graph labeling (as testified by the ongoing survey of Gallian~\cite{Gal2023} with thousands of citations) can be interpreted in this way, as well as the already mentioned geometric representations. 

In this paper, we are generalizing the notion of word-representable graphs by linking it to formal language theory in a way that any language~$L$ over the binary alphabet $\{0,1\}$ defines a graph class~$\cG_L$. \longversion{Now classical word representability defines a graph~$G=(V,E)$ via a word~$w$ over the alphabet of vertices~$V$, allowing an edge between two vertices $u,v$ if, in the scattered subword obtained from $w$ by omitting all letters but $u,v$, never two subsequent letters are identical. This property can be expressed by the binary language $\overline{(0\cup 1)^*11(0\cup 1)^*}\cap \overline{(0\cup 1)^*00(0\cup 1)^*}$ in a natural way and our approach generalizes word  representability exactly in this direction. }Word-representable graphs have been introduced in~\cite{KitPya2008} as an analogue of a DAG model of words studied in~\cite{KitSei2008} in order to solve some complexity problems in the context of semigroup theory. This is remarkable due to the known links between automata and semigroup theory (via transition monoids) and as we are here presenting a completely different connection between formal languages and graphs. 
Many scientific works have been devoted to word representability. Instead of trying to list all of them, we only point to the monograph~\cite{KitLoz2015} and the survey papers \cite{Kit2017,KitPya2018}. 

According to~\cite{CheKKKP2018}, a graph $G=(V,E)$ is $k$-11-represented by a word $w\in V^*$ if, for all distinct vertices $\ta,\tb\in V$, at most $k$ times the pattern (factor) $\ta\ta$ and at most $k$-times  the pattern $\tb\tb$ occurs in~$h_{\{\ta,\tb\}}(w)$, obtained from $w$ by deleting all vertices (letters) but $\ta,\tb$. The $0$-11-representable graphs are the classical word-representable graphs, while the 2-11-representable graphs are the class of all\longversion{ undirected} graphs. These types of representability were further generalized in~\cite{GaeJi2020} towards avoiding arbitrary patterns, and in this paper, we are generalizing these approaches even more by allowing arbitrary languages (not only certain (local) regular languages) to prescribe how edges are defined by a word, so that each binary language describes a  graph class.
\longversion{Yet another generalization of word representability was introduced in~\cite{KenMal2023}. The class of word-representable graphs also found some interest in the pure graph-theoretic literature, cf. \cite{CheKitSun2016,ChoKimKim2019,EnrKit2019,GleKitPya2018,Gle2019,KitSai2020}.}

We finally mention that there are other ways to connect words to graphs and hence languages to graph classes. For instance, the notion of letter graphs introduced by Petkovšek in~\cite{Pet2002}
gives rise to an infinite chain of graph classes where the languages are simply restricted by their alphabet size; in this context, this yields the graph parameter lettericity that obtained a certain popularity recently, see \cite{AleAALZ2023} and related papers. Quite akin to our studies are pattern avoiding words in \cite{JonKPR2015,Kit2017a}. In general, this $u$-representation of a graph depends on the labelling of its vertices and hence is different from our approach. Also, permutation representability as defined in \cite{KitSei2008} is different, as this gives additional conditions on the words that can be used to represent graphs, not only on the patterns that define edges.
Lozin~\cite{Loz2008a} discussed several representations of graphs with finite automata. Another completely different way to connect the theory of regular languages with basic notions of graph theory, e.g., with bounded treewidth, was recently proposed in~\cite{DieFerWol2022}.  

\medskip


\longversion{The paper is structured as follows. We start with a preliminary \Cref{sec:prelims} that is already relatively lengthy, as we are using standard not(at)ion from different mathematical areas. In \Cref{sec:defs_and_obs}, we introduce the main (new) definitions of this paper, in particular, we discuss the notion of $L$-representable graphs, with $L$ being any binary language. We also provide a number of examples and first observations regarding this notion. In \Cref{sec:operations}, we discuss several language and graph operations and how they relate. For instance, complementary languages can be used to describe graph complements, and taking induced subgraphs corresponds to certain projective morphisms on the language side. \Cref{sec:general-results} provides structural insights into language-representable graphs by showing that their edge sets can be decomposed by paying attention to the number of times a vertex (letter) occurs in the word describing a graph. This motivates to systematically study some finite languages in \Cref{sec:special-results}. But first, we discuss in \Cref{sec:general-results} some limitations of our approach, finding a number of graph classes that cannot be represented by any language in our setting. Also, we provide counting arguments that show that no finite language can represent all graphs, neither all word-representable graphs (in the classical sense), etc. In \Cref{sec:special-results}, we characterize all graph classes that can be described by languages that contain only binary words with at most two occurrences of~0 and at most two occurrences of~1. In fact, quite a number of well-known graph classes can be described in this (simple) way, including interval graphs, permutation graphs, circle graphs, convex graphs, and bipartite chain graphs. In \Cref{sec:palindromes}, we approach this phenomenon from another angle by looking at rather famous language classes, formed palindromes, copy-words, Lyndon-words and Dyck-words.
This way, we get descriptions of all graphs, of all bipartite graphs, as well as of all comparability graphs.
We also argue that these new representations of all graphs is more parsimonious than earlier findings in the area of graphs representable by languages and might be useful in practice.}
\longversion{The study of graphs that have a short word-representation with respect to a given language is also connected to \cite{GaeJi2020,SriHar2024} where short representations with respect to the classical notion of word representability was investigated.  We conclude in \Cref{sec:conclusions} by discussing several open questions and research directions that are opened up within our framework. Many graph classes encountered in this paper are defined in \Cref{sec:graphtheory}.}

\shortversion{For reasons of space, some proofs are not given in the main text body; this is signalled by $(*)$ prefixing the statement whose proof is in the appendix.}

\section{Preliminaries: Fixing General Notions and Notation}
\label{sec:prelims}


We fix now no(ta)tions that we are going to use in this paper.\longversion{

}
Let $\N$ denote the natural numbers including $0$ and set  $[n]=\{i\in\N \mid 1\leq i\leq n\}$ and  $\N_{\geq n}=\{k\in \N\mid k\geq n\}$ for some $n\in\N$. The cardinality of a set $X$ is denoted by $|X|$. For a set $X$ and $k\in\N$, define $\binom{X}{k}=\{Y\subseteq X\mid |Y|=k\}$. An \emph{alphabet} $\Sigma$ is a non-empty, finite set with \emph{letters} as elements. A \emph{word} over $\Sigma$ is a finite \emph{concatenation} of letters from~$\Sigma$; $\Sigma^{\ast}$ denotes the set of all words over~$\Sigma$, including the \emph{empty word}~$\emptyword$.\longversion{ In other words, $\Sigma^*$ is the free monoid generated by~$\Sigma$, whose operation is called concatenation, mostly written by juxtaposition, only sometimes made explicit as~$\cdot$.} Each \longversion{subset $L$ of $\Sigma^{\ast}$}\shortversion{$L\subseteq\Sigma^*$} is called a \emph{language} over~$\Sigma$.\longversion{ The empty word is the neutral element of the free semigroup~$(\Sigma^*,\cdot)$.} We set $\Sigma^+=\Sigma^{\ast}\setminus\{\emptyword\}$, i.e., $(\Sigma^+,\cdot)$ is the free semigroup generated by~$\Sigma$. For $w\in\Sigma^{\ast}$, $w[i]$ denotes the $i^{\mbox{\tiny th}}$ letter of~$w$. We extend concatenation to languages, so that we can define powers of a language~$L$ and the \emph{Kleene star} of $L$ as $L^*=\bigcup_{n\in\N}L^n$.
Identifying singleton sets with their elements, we can build \emph{regular expressions} from letters by using concatenation, union and Kleene star. \longversion{Sometimes, w}\shortversion{W}e also include complementation and the \emph{shuffle operation}~$\shuffle$ when building such expressions. For $u,v\in\Sigma^*$, \shortversion{$u\shuffle v\coloneqq\{x_1y_1x_2y_2\cdots x_ny_n\mid \exists x_1,\dots,x_n, y_1,\dots,y_n\in\Sigma^*: u=x_1\cdots x_n\land v=y_1\cdots y_n\}$}\longversion{$$u\shuffle v\coloneqq\{x_1y_1x_2y_2\cdots x_ny_n\mid \exists x_1,\dots,x_n, y_1,\dots,y_n\in\Sigma^*: u=x_1\cdots x_n\land v=y_1\cdots y_n\}$$}
and lift this to an operation between languages.

We will frequently define monoid (homo-)morphisms. Recall that a morphism from the free monoid $(\Sigma^*,\cdot,\emptyword)$ into another monoid is uniquely defined by giving the images of all letters. 
F\longversion{or instance, f}or $A \subseteq \Sigma$,  the \emph{projective morphism} 
$h_A: \Sigma^* \rightarrow A^*$ is the morphism with $h_A(\ta) = \ta$ for $\ta \in A$ and $h_A(\ta) = \emptyword$ for $\ta \in \Sigma \setminus A$. 
Moreover, define the morphism $|\cdot|:(\Sigma^*,\cdot)\to(\N,+)$, yielding the \emph{length} $|w|$ of a word $w\in\Sigma^{\ast}$, by $\ta\mapsto 1$ for each $\ta\in\Sigma$. For counting the number of occurrences of $\ta\in\Sigma$ in $w\in\Sigma^*$, the \emph{frequentness} of~$\ta$ in~$w$, we can use $|w|_{\ta}\coloneqq |h_{\{a\}}(w)|$. Now, define $\alphabet(w)=\{\ta\in\Sigma\mid\exists i\in[|w|]:\,w[i]=\ta\}$ as the symbols that occur in~$w$. A word $w\in\Sigma^{\ast}$ is called \emph{$k$-uniform} for some $k\in\N$ if $|w|_{\ta}=k$ for all $\ta\in\Sigma$.
Let $\widetilde{\cdot}:\{0,1\}^*\rightarrow\{0,1\}^*$ be the \emph{complement-morphism} mapping $0$ to $1$ and $1$ to $0$. Here and with other morphisms, we extend this mapping to sets of words, i.e., for a language $L\subseteq\{0,1\}^{\ast}$, define $\widetilde{L}=\{\widetilde{w}\mid w\in L\}$. Clearly, $\widetilde{\cdot}$ is an involution\longversion{, i.e., composing $\widetilde{\cdot}$ with itself yields the identity}. We define the \emph{symmetric hull} operator $\langle L\rangle\coloneqq L\cup \widetilde{L}$. A language $L\subseteq \{0,1\}^{\ast}$ is called {\em $0$-$1$-symmetric} if $L=\widetilde{L}$ (i.e., $L=\langle L \rangle$). For instance, the language $L_{\wrep}=(1\cup\emptyword)(01)^*(0\cup\emptyword)=(0\cup\emptyword)(10)^*(1\cup\emptyword)$ and its complement $\{0,1\}^{\ast}\backslash L_{\wrep}$ are $0$-$1$-symmetric.  Let $\mathrm{freq}(L)=\{n\in\mathbb{N}_{\geq 1} \mid \exists w\in L: |w|_0 = n\}$ denote the set of frequentnesses in a $0$-$1$-symmetric language~$L\subseteq \{0,1\}^{\ast}$. Define the \emph{reversal} $w^R$ of $w\in\Sigma^{\ast}$ by $w^R=w[|w|]\cdot w[|w|-1]\cdots w[1]$;  $w$ is a \emph{palindrome} if $w=w^R$\longversion{ holds}. Given $L\subseteq \Sigma^{\ast}$, set $L^R=\{w^R\mid w\in L\}$. 
The word $w\in\Sigma^{\ast}$ is called a \emph{repetition} if there exist $u\in\Sigma^{\ast}$ and $k\in\N_{\geq 2}$ such that $w=u^k$\longversion{, where $u^0=\emptyword$ and $u^k=uu^{k-1}$.}\shortversion{;} $w=u^2$ is called a \emph{copy-word}. We extend any given linear ordering~$\prec$ on $\Sigma$ to the linear {\em lexicographical order} on~$\Sigma^*$, again denoted by $\prec$ and defined by $u\prec w$ \iffl either $u$ is a \emph{prefix} of~$w$, i.e., $w=ux$ for some $x\in\Sigma^{\ast}$, or $w=x\tb y_1$, $u=x\ta y_2$ for 
letters $\ta\prec \tb$ and $x,y_1,y_2\in\Sigma^{\ast}$. Words $uv$ and $vu$ (for $u,v\in\Sigma^*$) are known as \emph{conjugates}; $\Sigma^*$ is partitioned in\longversion{to} \emph{conjugacy classes}\longversion{ as being conjugate defines an equivalence relation on~$\Sigma^*$}; a \emph{Lyndon word} is the lexicographically smallest element in such a\longversion{ conjugacy} class provided that no element in the class is a repetition. If we want to make the ordering explicit, we also write $\prec$-Lyndon word, or $\prec$-smaller, etc.

After establishing the basic notions of formal languages, we introduce the necessary definitions from graph theory.
\longversion{Throughout this paper, w}\shortversion{W}e only consider undirected finite graphs. A graph~$G$ is specified 
as\longversion{ a pair} $(V,E)$, where $V$ is the finite, non-empty set of vertices and $E\subseteq\binom{V}{2}$ 
is the set of edges\longversion{. The cardinality of $V$}\shortversion{; $|V|$} is\longversion{ also} called the \emph{order} of~$G$. If $\{u,v\}\in E$, $u$ is called a \emph{neighbor} of~$v$, and $N(v)\subseteq V$ \longversion{collects all}\shortversion{are the} neighbors of~$v$; $N[v]\coloneqq N(v)\cup\{v\}$\longversion{ is the \emph{closed neighborhood} of~$v$}. The {\em degree} of \longversion{a vertex }$v\in V$ is\longversion{ defined as} $\deg(v)\coloneqq|N(v)|$. Two vertices $v,v' \in V$ are called \emph{twins} if $N[v] \setminus \{ v, v' \} = N[v'] \setminus \{ v, v' \}$. The twins $v$ and $v'$ are called \emph{true twins} if $\{ v, v'\} \in E$, and \emph{false twins} if $\{ v, v'\} \notin E$. The \emph{complement} of the graph $G$, written $\overline{G}$, satisfies $\overline{G}=(V,\binom{V}{2}\setminus E)$. If $G_1=(V_1,E_1)$ and $G_2=(V_2,E_2)$ are graphs, then $\varphi:V_1\to V_2$ 
is a \emph{graph morphism} \iffl $\{u,v\}\in E_1$ implies $\{\varphi(u),\varphi(v)\}\in E_2$; $\varphi$ 
is a \emph{graph isomorphism} \iffl $\{u,v\}\in E_1\iff \{\varphi(u),\varphi(v)\}\in E_2$;\longversion{ if such an isomorphism exists,} then we write $G_1\simeq G_2$.  Sometimes, it is also convenient to write $V(G)$ for the vertex set of graph~$G$ and $E(G)$ for its edge set.   $G_1=(V_1,E_1)$ is a \emph{subgraph} of $G_2=(V_2,E_2)$ if $V_1\subseteq V_2$ and $E_1\subseteq E_2$. $G_1$ is an  \emph{induced subgraph} of $G_2$ if $E_1=E_2\cap \binom{V_1}{2}$.
Then, we usually write $G_1=G_2[V_1]$.

As usual, we consider abstract graphs when talking about graph classes, i.e., isomorphic graphs 
are considered to be equal. In other words, abstract graphs are isomorphism classes of concrete graphs.
If we consider concrete graphs, we use the symbol~$\simeq$ to express that two concrete graphs are isomorphic.
In this sense, a graph with $n\in\N$ vertices with an empty edge set is called a {\em null graph}
and is denoted by $N_n$. In contrast, a graph with $n\in\N$ vertices and all possible edges is called a {\em complete graph}, denoted by $K_n$. Transferring operations to abstract graphs, we find $\overline{K_n}=N_n$.
A graph with $n+m$ vertices, with $n,m\in\mathbb{N}_{\geq 1}$ is called a \emph{complete bipartite graph} $K_{n,m}=(V,E)$ if $V=U\cup W$, $U\cap W=\emptyset$, with $|U|=n$ and $|W|=m$ and $E$ contains all edges between $U$ and $V$ but no other edges. A graph  with $n\in\N$ vertices is a \emph{path} $P_n=(V,E)$ if there is a linear ordering~$<$ on its vertex set such that $\{u,v\}\in E$ if either $u$ is the immediate predecessor of~$v$ in~$<$ or $u$ is the immediate successor of~$v$ in~$<$.   A graph  with $n\in\N$ vertices is a \emph{cycle} $C_n=(V,E)$ if it contains a path $P_n=(V,E')$ and one additional edge~$e$ that connects the two vertices of this path that have degree one.
$N_n$, $K_n$, $P_n$, $C_n$, $K_{n,m}$ etc. are examples of abstract graphs in the sense that the sets of vertices and edges are not fixed concretely. But we can say that a concrete graph \emph{represents} an abstract graph if it satisfies the demanded properties.
A vertex in a graph is called \emph{universal} if its degree equals the order of the graph minus one\longversion{, i.e., every other vertex is neighbor of a universal vertex}. If $G_1=(V_1,E_1)$ and $G_2=(V_2,E_2)$ are two graphs with disjoint sets of vertices $V_1$ and $V_2$, then the (graph) \emph{union} of $G_1$ and $G_2$ is \longversion{$$G_1\cup G_2=(V_1\cup V_2,E_1\cup E_2)\,,$$}\shortversion{$G_1\cup G_2=(V_1\cup V_2,E_1\cup E_2)$,}
while the (graph) \emph{join}  of $G_1$ and $G_2$ is
\longversion{$$G_1\nabla G_2=(V_1\cup V_2,E_1\cup E_2\cup \{\{x_1,x_2\}\mid x_1\in V_1, x_2\in V_2\})\,.$$}\shortversion{$G_1\nabla G_2=(V_1\cup V_2,E_1\cup E_2\cup \{\{x_1,x_2\}\mid x_1\in V_1, x_2\in V_2\})$.}
It is convenient to think of these operations also on abstract graphs. Then, $K_5\cup N_3$ simply denotes the abstract graph consisting of a complete graph on five vertices, together with three isolates. We can also write $K_5\cup K_5$ and abbreviate this further as $2K_5$, to give one example to illustrate our notation. 
\begin{remark}\label{rem:join-complete-graph}
\longversion{Observe that }$K_{n,m}=N_n\nabla N_m$ and\longversion{ that} $\overline{K_{n,m}}=K_n\cup K_m$.
\longversion{Observe that }$G\nabla K_1$ can be viewed as adding a universal vertex to~$G$. If the graph class $\mathcal{G}$ is closed under adding universal vertices, \longversion{then }this implies by \longversion{simple }induction that, for any $G\in\mathcal{G}$ and any complete graph~$K_n$, $G\nabla K_n\in\mathcal{G}$. 
\end{remark}
More advanced graph-theoretic concepts that are only used in a local way in this paper and also a multitude of graph classes is defined in \Cref{sec:graphtheory}.

\section{Definitions and Simple Observations About the Framework}
\label{sec:defs_and_obs}

In this section, we are fixing the main notions and notations that allow us to define graph classes associated to binary languages. This gives the basis of our framework that allows us to easily derive many interesting properties of different graph classes by unified arguments.

\begin{definition}
Let $V$ be an alphabet and fix $u,v\in V$ with $u\neq v$. Define the monoid morphism $h_{u,v}:V^{\ast}\rightarrow\{0,1\}^{\ast}$ by $u\mapsto 0$, $v\mapsto 1$ and $x\mapsto \emptyword$ for $x\in V\setminus\{u,v\}$.
\end{definition}

In general, $h_{u,v}$ can be viewed as the composition of the projection $h_{\{u,v\}}$ and the renaming isomorphism $\iota:\{u,v\}^*\to\{0,1\}^*$ with $u\mapsto 0$ and $v\mapsto 1$.

\begin{remark}\label{01sym}
Notice that if $L$ is $0$-$1$-symmetric, we have $h_{u,v}(w)\in L\iff h_{v,u}(w)\in L$. 
\end{remark}

These language-theoretic prerequisites are essential for the following definition that introduces the central notion of this paper, giving our framework of research.

\begin{definition}Let $L\subseteq\{0,1\}^*$ be $0$-$1$-symmetric. Let $w\in V^*$ with $V=\alphabet(w)$.
We call a graph $G=(V,E)$ \emph{$L$-represented} by\longversion{ the word}~$w$  \iffl, for all $u,v\in V$, the following \longversion{equivalence }holds:
\(\{u,v\}\in E \Leftrightarrow h_{u,v}(w)\in L. \)
A graph  $G=(V,E)$  is \emph{$L$-representable} if it can be $L$-represented by some\longversion{ word} $w\in V^*$. 
W\longversion{e w}rite $G(L,w)$ for the graph $G$ that is $L$-represented by a given word~$w$.
The class of all $L$-representable graphs \longversion{is denoted by}\shortversion{is written}~$\mathcal{G}_L$, i.e.\longversion{, more formally}, $\mathcal{G}_L=\{G(L,w)\mid w\in\N^*\}$.
\end{definition}

\begin{remark}\label{rem:01sym-convention}
By Remark~\ref{01sym}, $\mathcal{G}_L$ and thus $G(L,w)$ for given $L$ and $w$ is well-defined and defines an undirected graph. As $0$-$1$-symmetry is essential for defining undirected graphs according to the previous definition, we will tacitly assume this condition for any binary language used for defining graphs and graph classes in the following.
\end{remark}

We will also say that, e.g., $C_4$ is $L$-represented by some word $w\in V^*$ (with $|V|=4$), referring to some abstract graph. 

\begin{lemmarep}\label{lem:clique-representation} \shortversion{$(*)$}
Let $V$ be an alphabet, $w\in V^*$ and $F(w)=\{|w|_a\mid a\in V\}$ be the \shortversion{letter }fre\-quent\-nesses\longversion{ of letters in~$w$}. If $0^k\shuffle 1^\ell\subseteq L\subseteq\{0,1\}^*$ for all $k,\ell\in F(w)$, then $K_{|V|}$ is $L$-represented by~$w$.
\end{lemmarep}

\begin{proof}
Consider $G=(V,E)=G(L,w)$. Let $u,v\in V$ be two arbitrary letters. By definition, $|w|_u,|w|_v\in F(w)$. Hence, $h_{u,v}(w)\in 0^{|w|_u}\shuffle 1^{|w|_v}\subseteq L$, i.e., $\{u,v\}\in E$. Therefore, $G$ is a complete graph.
\end{proof}

To be less confusing and more adequate for a human reader, we will also use other vertex names (letters) than natural numbers, but for the formal definition, it is best to fix a countable alphabet over which we form possible finite words and which hence serves as a potentially infinite source of vertex names.

As an example, consider \shortversion{$L_{\text{wrep}}=(1\cup\emptyword)(01)^*(0\cup\emptyword)=(0\cup\emptyword)(10)^*(1\cup\emptyword)$,}\longversion{
$$
L_{\text{wrep}}=\overline{(0\cup 1)^*11(0\cup 1)^*}\cap \overline{(0\cup 1)^*00(0\cup 1)^*}=(1\cup\emptyword)(01)^*(0\cup\emptyword)=(0\cup\emptyword)(10)^*(1\cup\emptyword)\,,
$$
}
which would lead to the class $\cG_{L_{\text{wrep}}}$ of word-representable graphs known from the literature.
Let us look at some more examples now to clarify our definitions.

\begin{example}\label{exa:represented-graphs}
First, we describe some graphs with the help of some formal languages.
\begin{enumerate}
    \item \label{exa:represented-graphs-C4} 
    The cycle $C_4$ is $\langle 0101\rangle$-represented by $w = 14213243$. The same word also  $L_{\text{wrep}}$-represents $C_4$. However, $G(\langle 0101\rangle,w)\simeq G(\langle 0011\rangle,w)$ represents $K_2\cup N_2$, while  $G(\langle 0011,0110\rangle,w)$ represents $2K_2$. Also confer \Cref{fig:languages-C4} for other languages. 
\end{enumerate}
Now, we determine $\cG_L$ for different $L\subseteq\{0,1\}^*$.
\begin{enumerate}\addtocounter{enumi}{1}
    \item  \label{exa:represented-graphs-NullGraphs}
    If $L=\emptyset$  or  $L=\{\emptyword\}$, then $\cG_L$ is the class of null graphs. More formally, $\cG_L=\{N_n\mid n\in\N_{\geq 1}\}$. Namely, in order to represent $N_n=([n],\emptyset)$, take the word $w_n=1\cdot 2\cdots n$\longversion{, or any other word that contains each letter from $[n]$ at least once}.  Clearly, for any pair of vertices $i< j$, $h_{ij}(w_n)=01\notin L$.
    \item \label{exa:represented-graphs-CompleteGraphs} 
    If $L=\{0,1\}^*$ or $L=\{0,1\}^+$, then $\cG_L$ is the class of complete graphs, i.e., $\cG_L=\{K_n\mid n\in\N_{\geq 1}\}$. \longversion{Namely, i}\shortversion{I}n order to represent $\left([n],\binom{[n]}{2}\right)$, take \longversion{the word }$w_n=1\cdot 2\cdots n$\longversion{, or any other word that contains each letter from $[n]$ at least once}.  Clearly, for any pair of vertices $i< j$, $h_{ij}(w_n)=01\in L$. Also see \Cref{lem:clique-representation}, or \Cref{exa:represented-graphs-NullGraphs} and \Cref{prop:compl}\longversion{ below}. 
   \longversion{ \item \label{exa:represented-graphs-CompleteUnionNullGraphs} 
   If $L=\{01,10\}$, then $\cG_L=\{K_n\cup N_m\mid n,m\in\N\}$. In order to represent $K_n\cup N_m$ as $\left([n+m],\binom{[n]}{2}\right)$, take the word $w_{n,m}=w_n\cdot (n+1)(n+1)\cdots (n+m)(n+m)$, or any other word that contains each letter from $[n]$ exactly once and each other letter at least twice.  Clearly, for any pair of vertices $i,j\in [n]$ $i< j$, $h_{ij}(w_{n,m})=01\in L$, while for any pair $i,j\in [n+m]$, $i<j$ and $j>n$, $h_{ij}(w_{n,m})=011\notin L$ or $h_{ij}(w_{n,m})=0011\notin L$. }
   \longversion{    \item \label{exa:represented-graphs-CompleteBipartiteUnionNullGraphs} 
   If $L=(0\shuffle 1)\shuffle\{0,1\}$,  
    then  $\cG_L=\{K_{n,m}\cup N_\ell\mid n,m,\ell\in\N\}$. 
    Namely, take $w_{n,m,\ell}=w_{n,m}\cdot(n+m+1)(n+m+1)(n+m+1)\cdots (n+m+\ell)(n+m+\ell)(n+m+\ell)$. Observe that  $h_{ij}(w_{n,m,\ell})=011\in L$ \iffl $1\leq i\leq n$ and $n<j\leq n+m$. 
    }
\end{enumerate}
\end{example}
From the definitions, we can immediately observe the following.
\begin{lemmarep}\label{lem:representability-reversal} \shortversion{$(*)$}
A graph~$G$ can be $L$-represented \iffl it can be $L^R$-represented.
\end{lemmarep}

\begin{proof} 
For every word $w$ over an alphabet $V$, $G(L,w) = G(L^R,w^R)$. Hence, a graph~$G$ can be $L$-represented (by~$w$) \iffl it can be $L^R$-represented (by~$w^R$). 
\end{proof}
\longversion{By the previous proof, we obtain following corollary.}

\begin{corollary}\label{cor:L-symmetry}
    Let $L\subseteq \{0,1\}^*$ \longversion{be a language with}\shortversion{satisfy} $L=L^R$. Then for each word $w$, $G(L,w)= G(L,w^R)$.
\end{corollary}

The following example is a counterexample to the possible conjecture that $G$ is $L$-representable \iffl $G$ is $(L\cup L^R)$-representable.

\begin{example} \label{exa:notClosedUnderReversal}
Consider $L = \overline{\langle 0001\rangle }$ and $w = \ta\ta\ta\tb$. $G(L,w) = (\{\ta,\tb \}, \emptyset)$, but $L \cup L^R = \{ 0, 1 \}^*$. By \Cref{exa:represented-graphs-CompleteGraphs} of \Cref{exa:represented-graphs}, $\mathcal{G}_{L \cup L^R}$ is the class of complete graphs and $G \notin \mathcal{G}_{L \cup L^R}$.
\end{example}

Because of \Cref{lem:representability-reversal}, we like to study $0$-$1$-symmetric languages~$L$ that are closed under reversal, i.e., for which $L=L^R$ holds, in more detail. For instance, $L_{\text{wrep}}$ is closed under reversal. $L_{\text{wrep}}$ equals \longversion{$$L_{\overline{1^k}}=\overline{\{0,1\}^*\{0^k\}\{0,1\}^*}\cap \overline{\{0,1\}^*\{1^k\}\{0,1\}^*}$$}\shortversion{$L_{\overline{1^k}}=\overline{\{0,1\}^*\{0^k\}\{0,1\}^*}\cap \overline{\{0,1\}^*\{1^k\}\{0,1\}^*}$}
for $k = 2$.
For each $k \geq 2$, 
$L_{\overline{1^k}}$ is $0$-$1$-symmetric and closed under reversal. $\mathcal{G}_{L_{\overline{1^k}}}$ clearly corresponds to the family of graphs called $1^k$-representable in~\cite{JonKPR2015},
where it was shown that, for every $k \geq 3$, \emph{every} graph is  $1^k$-representable. In other words, $\mathcal{G}_{L_{\overline{1^k}}}$ is the class of all (undirected) graphs. Observe that for each $k$, $L_{\overline{1^k}}$ is a very simple regular languages (a so-called local language). 
Likewise, the 112-representable graphs as defined in~\cite{GaeJi2020} are exactly\longversion{ the class} $\cG_{\overline{\langle 001\rangle}}$ in our notation.

With $B_k=\{w\in \{0,1\}^*\mid |w|_0=|w|_1=k\}$ collecting all \emph{$k$-uniform words}, we can easily define $L_{k\text{-uni-wrep}}\coloneqq L_{\text{wrep}}\cap B_k$ and hence the class $\mathcal{G}_{L_{k\text{-uni-wrep}}}$ of \emph{$k$-uniformly word-representable graphs} from the literature.
Observe that $L_{k\text{-uni-wrep}}\subseteq\{0,1\}^{2k}$ is a finite language.
\longversion{This is one of the motivations for us to pay special attention to graphs that are $L$-representable by finite~$L$.} But also these simple languages can possess interesting characterizations. For example, by \cite[Thm. 5.1.7]{KitLoz2015}, $\cG_{\langle 0101\rangle}$ is the class of circle graphs as $\langle 0101\rangle = L_{2\text{-uni-wrep}}$. 
Reinterpreting~\cite{CheKKKP2018}, the family of $k$-11-representable graphs 
can be described by $L_{k\text{-}11}$, the\longversion{ language of} binary words containing each of $00$ and $11$ at most $k$ times. For example, $L_{\text{wrep}}=L_{0\text{-}11}$. Interestingly, $\cG_{L_{2\text{-}11}}$ contains all graphs, and  $\cG_{L_{1\text{-}11}\cap B_2}$ is the class of interval graphs.

\longversion{As it can be quite tedious to list all words of a finite languages, we defined for $L\subseteq\{0,1\}^*$, 
$\langle L\rangle\coloneqq L \cup \{ \widetilde{w} \mid w \in L \}\,.$
Clearly, $\langle \cdot\rangle$ is a hull operator, so that it is, in  particular, idempotent, i.e., $\langle L\rangle=\langle\langle L\rangle\rangle$. 
To further simplify notation and to help recognize patterns, for $w\in\{0,1\}^*$ let $\nu(w)$
denote the lexicographically smallest element of $\{w,\widetilde{w}\}$, 
called \emph{normal form} of~$w$. 
Let $\nu(L)=\{\nu(w)\mid w\in L\}$. Clearly, $\langle\nu(L)\rangle=\langle L\rangle$, so that we can take $\nu(L)$ as the normal form of~$L$. For finite languages, we also omit braces in this notation and list the normal forms in length-lexicographical order.
For example, $L=(0\shuffle 1)\shuffle\{0,1\}=\langle 001,010,011\rangle$ in the normal form presentation as $\nu(L)=\{001,010,011\}$. Without further explicit mentioning, this will be our convention for writing down finite binary languages.}

\longversion{The following notions will set a wording important for our analysis of graph classes.
We will call $L\subseteq\{0,1\}^*$ \emph{length-uniform} if there is some $k\in\N$ such that $L\subseteq 0^k\shuffle 1^k$. Then, we also say that $L$ is \emph{$k$-uniform}. Similarly, $L\subseteq\{0,1\}^*$ is called \emph{nearly length-uniform} if there two different numbers $k,\ell\in \N_{\geq 1}$ such that $L\subseteq (0^k\shuffle 1^\ell)\cup (0^\ell\shuffle 1^k)$. If $k<\ell$, we also say that $L$ is \emph{$(k,\ell)$-uniform}.  Nearly by definition, $L_{k\text{-uni-wrep}}$ is $k$-uniform. \longversion{Moreover, $\langle 001,010,011\rangle$ is $(1,2)$-uniform.}}


\section{Operations on Languages and Graphs}
\label{sec:operations}
In this section, we discuss several language and graph operations and how they relate. First, we present an interplay between graph and monoid morphisms.

\begin{lemma}\label{lem:isomorphism}
    Let $G_1=(V_1,E_1)$, $G_2=(V_2,E_2)$ be graphs, $L\subseteq \{0,1\}^*$ be a language and $w=w_1\cdots w_n$ a word with $G_1=G(L,w)$. Then for each graph isomorphism $h:V_1\to V_2$ from $G_1$ to $G_2$, $G_2\simeq G(L,h(w_1)\cdots h(w_n))=G(L,h(w))$, with $h$ lifted to a monoid morphism. If $h:V_1\to V_2$ is a graph morphism, then $G(L,h(w))$ is isomorphic to a subgraph of~$G_2$.
\end{lemma}

\begin{proof}
    Let $v,u\in V$. Define $h(w)=h(w_1)\cdots h(w_n)$, i.e., interpret the graph isomorphism~$h$ now as a monoid morphism $h:V_1^*\to V_2^*$. Then $ \{h(v),h(u)\} \in E(G(L,h(w))) \iff$
    \begin{equation}
         h_{h(v),h(u)}(h(w))\in L\iff h_{v,u}(w)\in L\iff \{v,u\} \in E_1 \iff \{h(v),h(u)\} \in E_2\,. \label{eq:morphisms}
    \end{equation} Also, as $h:V_1\to V_2$ is a bijection, $|V_2|=|h(w)|=|w|=|V_1|$.
    If $h:V_1\to V_2$ is (only) a graph morphism, then the last equivalence in \Cref{eq:morphisms} will be (only) an implication.
\end{proof}

The set operations on the sets of edges of the graphs $G(L_1,w)$ and $G(L_2,w)$ are compatible with the set operations of $L_1$ and $L_2$ in the following sense:

\begin{lemma} \label{lem:edge_sets}
Let  $L_1, L_2 \subseteq \{ 0, 1 \}^\ast$ be 
languages. Let $w\in V^*$, with 
$V=\alphabet(w)$. 
\begin{enumerate}
\item Let $\Box$ be a binary set operation. Then, $E(G(L_1 \mathbin{\Box} L_2,w)) = E(G(L_1,w)) \mathbin{\Box} E(G(L_2,w))$.
\item  \label{item:op-complement} $E(G(\overline{L_2},w)) = \binom{V}{2} \setminus E(G(L_2,w))$ and $G(\overline{L_2},w)$ is the complement graph of $G(L_2,w)$.
\item \label{item:op-inclusion} If $L_1 \subseteq L_2$, then $E(G(L_1,w)) \subseteq E(G(L_2,w))$ and $G(L_1,w)$ is a subgraph of $G(L_2,w)$.
\end{enumerate} 
\end{lemma}

\begin{proof}
\begin{enumerate}
\item Let $\Diamond:\{0,1\}^2\to\{0,1\}$ be the binary Boolean operation corresponding to the binary set operation $\Box:2^X\times 2^X\to2^X$ in the sense that $A\mathbin{\Box} B=\{x\in X\mid x\in A \mathbin{\Diamond} x\in B\}$.  
For $v \in V$ and $u \in V \setminus \{ v \}$, the following holds: 
\begin{align*}
\{ u, v \} \in E(G(L_1 \mathbin{\Box} L_2,w)) & \iff h_{u,v}(w) \in L_1 \mathbin{\Box} L_2 \\
                                    & \iff h_{u,v}(w) \in L_1 \mathbin{\Diamond} h_{u,v}(w) \in L_2 \\
                                    & \iff \{ u, v \} \in E(G(L_1,w)) \mathbin{\Diamond} \{ u, v \} \in E(G(L_2,w)) \\
                                    & \iff \{ u, v \} \in E(G(L_1,w)) \mathbin{\Box} E(G(L_2,w))
\end{align*}
\item The statement follows from (1) with $L_1 = \{ 0, 1 \}^\ast$.

\item The statement follows from (1) as,  for $A,B\subseteq X$, $A\subseteq B$ \iffl $A\cap B=A$.
\qed
\end{enumerate}\renewcommand{\qed}{}
\end{proof}

This lemma has interesting consequences when we turn our attention towards graph classes instead of single languages and graphs. \longversion{The next statement is a reinterpretation of}\shortversion{Next, we restate} \Cref{lem:edge_sets}, Item~\ref{item:op-complement}.

\begin{proposition}\label{prop:compl}
For $L \subseteq \{ 0, 1 \}^* $, $\mathcal{G}_{\overline{L}}=\{\overline{G}\mid G\in \mathcal{G}_L\}$.
\end{proposition} 
The next statement is a reinterpretation of \Cref{lem:edge_sets}, Item~\ref{item:op-inclusion}, \longversion{moving the focus towards}\shortversion{focusing on} graph classes.

\begin{lemmarep}\label{lem:subgraph} \shortversion{$(*)$}
Let $L_1 \subseteq L_2 \subseteq \{ 0, 1 \}^*$. For every $G \in \cG_{L_1}$ there exists a graph $G' \in \cG_{L_2}$ such that $V(G) = V(G')$ and $G$ is a subgraph of $G'$. 
\end{lemmarep}
\begin{proof}
Let $G = (V,E) \in \cG_{L_1}$. Hence there exists a word $w \in V^*$ such that $G = G(L_1,w)$. Let $G' = G(L_2,w)$. Obviously $V(G') = V$. 
\begin{align*}
\{ u, v \} \in E(G) & \iff h_{u,v}(w) \in L_1  \implies h_{u,v}(w) \in L_2 \iff \{ u, v \} \in E(G')
\end{align*}
Hence $E(G) \subseteq E(G')$ and $G$ is a subgraph of $G'$. 
\end{proof}

\begin{proposition}\label{prop:inclusion}
Let $L_1 \subseteq L_2 \subseteq \{ 0, 1 \}^*$. If $\cG_{L_2}$ is closed under subgraphs, then $\cG_{L_1}\subseteq\cG_{L_2}$.
\end{proposition}

\begin{proof}
Let $G=(V,E)\in\cG_{L_1}$. Hence, there is some $w\in V^*$ with $G\simeq G(L_1,w)$. By \Cref{lem:subgraph}, $G(L_1,w)$ is a subgraph of $G(L_2,w)$. By definition, $G(L_2,w)\in \cG_{L_2}$. As $\cG_{L_2}$ is closed under taking subgraphs, $G\in \cG_{L_2}$. Hence, $\cG_{L_1}\subseteq\cG_{L_2}$.
\end{proof}
\longversion{
The previous proposition also raises the question if the transfer of language inclusion to the inclusion of graph classes is also possible under weaker assumptions, or is possibly the closure under taking subgraphs unnecessary at all? We will come back to this question later.

}The following lemma about projective morphisms is helpful in various settings.

\begin{lemma} \label{lem:induced_subgraph} Let $L\subseteq \{0,1\}^*$ be a language, $V$ an alphabet and $w \in V^*$, with $V=\alph(w)$, and $G = (V,E) = G(L,w)$. Let $A\subseteq V$. 
Then, $G[A] = G(L,h_A(w))$. 
\end{lemma}

\begin{proof}
Let $G' = G(L,h_A(w))$. Obviously, $V(G') = A = V(G[A])$. For every $u \in A$ and $v \in A \setminus \{ u \}$, the following holds: 
\[\{ u, v \} \in E(G[A]) \iff{}  \{ u, v \} \in E 
                       \iff{}  h_{u,v}(w) \in L 
                       \iff{}  h_{u,v}(h_A(w)) \in L 
                       \iff{}  \{ u, v \} \in E(G') 
                       \]
Therefore, $E(G')=E \cap \binom{A}{2}$.
\end{proof}

This lemma immediately entails the following important property of any graph class that is definable in our framework. \longversion{Notice how w}\shortversion{W}e again switch our attention from languages to graph classes.

\begin{theorem}\label{thm:hereditary}
For any $L\subseteq\{0,1\}^*$, $\cG_L$ is hereditary, i.e., closed under induced subgraphs.
\end{theorem}

This property was known before for word-representable graphs \cite[Prop. 3.0.8]{KitLoz2015} and is known for many of the graph classes that we identify in the following as being $L$-representable for some language~$L$, but such claims are often hard to find in the literature, and here we have one single simple argument that works for all of these graph classes.

The next proposition exploits `holes' in the set of frequentnesses.
\begin{propositionrep}\label{propos:closure-adding-isolates} \shortversion{$(*)$}
For any $L\subseteq\{0,1\}^*$, if $\N_{\geq 1}\setminus \mathrm{freq}(L)\neq\emptyset$, then $\cG_L$ is closed under adding isolated vertices and  $\cG_{\overline{L}}$ is closed under adding universal vertices.
\end{propositionrep}

\begin{proof}
Consider some $L\subseteq\{0,1\}^*$ such that there is some `missing frequentness', say, $m\in \N_{\geq 1}\setminus \mathrm{freq}(L)$.
Let $G=(V,E)\in\cG_L$. Hence, there is some $w\in V^*$ with $G=G(L,w)$. Now, let $\ta\notin V$. Define $w'=w\cdot \ta^m$. Then, $G'=G(L,w')$ is isomorphic to $G\cup N_1$.
Consider $H\in \cG_{\overline{L}}$. By \Cref{prop:compl}, $\overline{H}\in \cG_L$. By the previous reasoning, $\overline{H}\cup N_1\in \cG_L$. Now, observe that $\overline{\overline{H}\cup N_1}=H\nabla N_1$ by De Morgan's Law. Inductively, the claims follow.
\end{proof}

Next, we derive another nice property that all graph classes that are representable by some language satisfy.  We say that a graph class $\cG$ is \emph{closed under adding some twins} if for every $G=(V,E)\in \cG$ and for every $v\in V$ and $v'\notin V$, at least one of the following two statements is true:
\begin{itemize}
    \item The graph $G'_v$ obtained from $G$ by adding $v'$ as a true twin of~$v$ belongs to $\cG$.
    
    \item The graph $G''_v$ obtained from $G$ by adding $v'$ as a false twin of~$v$ belongs to $\cG$.
\end{itemize}
Notice that $\{G'_v,G''_v\}\subset\cG$ is possible. \longversion{\footnote{This operation might look weird at first glance, but a very similar one appeared in the study of word-representability of split graphs, see, e.g., \cite{CheKitSai2022}.}}

\begin{theorem}\label{thm:twins}
Let $L\subseteq\{0,1\}^*$. Then, $\cG_L$ is closed under adding some twins.
\end{theorem}

\begin{proof}
Fix $L\subseteq\{0,1\}^*$. Let $G=(V,E)\in \cG_L$. Hence, there is some word~$w$ over the alphabet~$V$ such that $G=G(L,w)$. Let $v\in V$.
Then, we can decompose $w=w_0vw_1v\cdots vw_k$ such that, for each $i\in\N$, $i\leq k$, $v\notin\alph{(w_i)}$. Let $v'\notin V$. Define
\begin{equation}w'\coloneqq w_0vv'w_1vv'\cdots vv'w_k\,.
\label{eq:twin-def}
\end{equation}
Let $G'=(V',E')=G(L,w')$ such that $V'=V\cup\{v'\}$. 
By construction, for any $x,y\in V$, $\{x,y\}\in E\iff \{x,y\}\in E'$.
Moreover, for any $x\in V\setminus\{v\}$, $\{x,v\}\in E\iff \{x,v'\}\in E'$. This means that $v,v'$ are twins. Whether or not they are true or false twins depends on whether or not $h_{v,v'}(w')\in L$.
In any case, $G'\in\cG_L$.
\end{proof}

\begin{remark}
The reader might wonder how both adding true and false twins can result in a graph from our graph class. This is due to the fact that there is a certain arbitrariness in the definition of~$w'$ in \Cref{eq:twin-def}. More generally speaking, we could pick any $w'$ from $$\{w_0\}\cdot (v\shuffle v')\cdot \{w_1\}\cdot (v\shuffle v')\cdots (v\shuffle v')\cdot \{w_k\}\,.$$
For some choice of $w'$, $h_{v,v'}(w')\in L$ might hold, while for other choices of $w'$, $h_{v,v'}(w')\notin L$.
\end{remark}
We will derive a number of interesting consequences from adding twins.
\longversion{
\begin{proposition}\label{prop:false-twins}
Let $L\subseteq\{0,1\}^*$.  If $\cG_L$ is closed under adding false twins, then $\cG_L$ contains all null graphs.
\end{proposition}

\begin{proof}
First, as $\cG_L$ is closed under taking induced subgraphs, $N_1\in \cG_L$. 
By adding more and more false twins, clearly arbitrarily large null graphs can be created.
\end{proof}
In the following, we assume that the reader is acquainted with the notion of treewidth.

\begin{propositionrep}\label{prop:true-twins}Let $L\subseteq\{0,1\}^*$. If $\cG_L$ is closed under adding true twins, then $\cG_L$ contains all complete graphs and hence graphs of unbounded treewidth\longversion{ and non-bipartite graphs}.
\end{propositionrep}

\begin{proof}
First, as $\cG_L$ is closed under taking induced subgraphs, $N_1=K_1\in \cG_L$. By adding more and more true twins, clearly arbitrarily large complete graphs can be created. Hence,  $\cG_L$ contains\longversion{ non-bipartite graphs (like $K_3$) and} graphs of unbounded treewidth.
\end{proof}

}We \shortversion{conclude this section with}\longversion{are next} giving an example how reasoning about operations on words, languages and graphs can yield inclusion results for graph classes. The proof idea of the following theorem is to define two specific operations on words of even length (that are special cases of (literal) shuffle operations; see \cite{Ber87}) so that they model graph union and join.

\begin{theoremrep}\label{thm:cographs}
Let $L\subseteq\{0,1\}^*$ contain either both $1001$ and $0110$ or both $1010$ and $0101$ (but not all four words). Then, ${\cal G}_L$ contains all cographs.
\end{theoremrep}

\begin{proof}
Clearly, a $K_1$ can be described by the word $vv$.
This word, and all words~$w$ that we will construct in the following in order to describe a cograph, has the following property: If $w\in V^*$ describes a graph of order~$n=|V|$, then $|w|=2n$ and $w$ can be \emph{evenly decomposed} into $w=w_1w_2$ such that $|w_1|=|w_2|=n$ and every vertex from~$V$ occurs exactly once in $w_1$ and exactly once in $w_2$.

Hence, let $G_1=(V_1,E_1)$ and $G_2=(V_2,E_2)$, with $V_1\cap V_2=\emptyset$, be two cographs that, by induction hypothesis, can be described by words $w$ and $u$ that can be evenly decomposed as $w=w_1w_2$ and $u=u_1u_2$. Define $x=w_1u_1w_2u_2$ and $y=w_1u_1u_2w_2$. Set $G(x)=(V,E_x)$ and $G(y)=(V,E_y)$ with $V=V_1\cup V_2$. Clearly, $E_1\cup E_2\subseteq E_x\cap E_y$. 

Now, assume that $L$ contains both $1001$ and $0110$ but neither $1010$ nor  $0101$.
Then, $G(x)=G_1\cup G_2$ and $G(y)=G_1 \nabla G_2$. 
Conversely, if $L$ contains both $1010$ and $0101$ but neither $1001$ nor $0110$, then
$G(x)=G_1\nabla G_2$ and $G(y)=G_1 \cup G_2$. 

In both cases, by induction we see that all cographs are contained in the graph class~${\cal G}_L$.
\end{proof}

This theorem implies in particular that the classical word-representable graphs contain all cographs. 
One can find similar properties for $L$ based on words of length $kn$ for any $k\geq 2$ that also guarantee that all cographs can be represented, by giving other decomposition conditions for these words\longversion{, but these conditions grow increasingly complicated without giving too many further insights. Therefore, we refrain from stating them explicitly. But this way, one can easily see that}\shortversion{. Similarly,} any class  $\mathcal{G}_{L_{k\text{-uni-wrep}}}$ (for $k\geq 2$) contains all cographs.

\longversion{Finally, we are formulating a \emph{trash theorem}. This will be useful when trying to apply \Cref{prop:compl} in the context of the study, say, of 2-uniform languages and the corresponding graph classes, because then, we rather need a complementation relative to all 2-uniform words instead of relative to all possible words. To understand the difference between these two types of complements of a binary language~$L$, we now study their difference formally, which gives the set $T_L$ as defined in the following theorem, which can be viewed as a collection of a sort of trash words with only minor influence on the graph classes.



\begin{theorem}[Trash Theorem] \label{trash-theorem} 
Let $L\subseteq \{0,1\}^*$. 
Let $\hat L \coloneqq L \cup T_L$, where we define the trash as $T_L \coloneqq  \{ w \in \{ 0, 1 \}^* \mid |w|_0 \notin \mathrm{freq}(L) \vee |w|_1 \notin \mathrm{freq}(L) \} $.  Then, $L$ and $T_L$ are disjoint. 
\begin{enumerate}
\item If $\mathcal{G}_{\hat L}$ is closed under adding isolated vertices, then $\mathcal{G}_{L} = \mathcal{G}_{\hat L}$ and for every $G \in \cG_L$ there exists a $w \in V(G)^\ast$ such that $|w|_v \in \mathrm{freq}(L)$ for every $v \in V(G)$ and $G(L,w) = G(\hat L, w)$. 
\item If $\mathcal{G}_{L}$ is closed under adding universal vertices, then $\mathcal{G}_{\hat L} \subseteq \mathcal{G}_{L}$. 
\end{enumerate}
\end{theorem}
\begin{proof} By \Cref{rem:01sym-convention}, $L$ is $0$-$1$-symmetric. Hence for every $w\in L$, $|w|_0 \in \mathrm{freq}(L)$ and $|w|_1 \in \mathrm{freq}(L)$. This shows that $L$ and $T_L$ are disjoint. We can assume, w.l.o.g., that $\mathrm{freq}(L) \neq \mathbb{N}_{\geq 1}$. (Otherwise, $L = \hat L$ and the proof is trivial.) 

\begin{enumerate}
\item We will show the inclusion $\mathcal{G}_{L} \subseteq \mathcal{G}_{\hat L}$ first. Consider a graph $G = (V,E) \in \mathcal{G}_{L}$. Hence, there exists a word $w \in V^*$ such that $G = G(L, w)$. Partition $V$ into the sets  \shortversion{$V_1 = \{ v \in V \mid |w|_v \notin \mathrm{freq}(L) \}$ and $V_2 = \{ v \in V \mid |w|_v \in \mathrm{freq}(L) \}$.}\longversion{$$V_1 = \{ v \in V \mid |w|_v \notin \mathrm{freq}(L) \}\text{ and }V_2 = \{ v \in V \mid |w|_v \in \mathrm{freq}(L) \}\,.$$} $V_1$ is a set of isolated vertices, also confer \Cref{propos:closure-adding-isolates}. $G[V_2] = G(L,h_{V_2}(w))$, because of \Cref{lem:induced_subgraph}. By definition of $V_2$, $G(L,h_{V_2}(w)) = G(\hat L,h_{V_2}(w))$. Hence, $G[V_2] \in \mathcal{G}_{\hat L}$. Therefore, $G = G[V_2] \cup (V_1, \emptyset) \in \mathcal{G}_{\hat L}$, because $\mathcal{G}_{\hat L}$ is closed closed under adding isolated vertices. 

We will show $\cG_{\hat L} \subseteq \mathcal{G}_{L}$ and the rest of the statement next. For the sake of contradiction, suppose a graph $G \in \mathcal{G}_{\hat L}$ exists such that, for every $w \in V(G)^\ast$ with $G = G(\hat L, w)$, a vertex $t \in V(G)$ exists such that $|w|_t \notin \mathrm{freq}(L)$. We call this statement $(*)$. We want to add a new isolated vertex $v \notin V(G)$ to~$G$. Let $G' = G \cup (\{v \}, \emptyset)$. Since $G' \in \cG_{\hat L}$, a word $w'$ exists such that $G(\hat L, w') = G'$. If $|w'|_v \notin \mathrm{freq}(L)$, $\{ u, v \} \in E(G')$ for every $u \in V(G)$. This is a contradiction, because $V(G)$ is not empty and $v$ is supposed to be isolated in~$G'$. Hence, $|w'|_v \in \mathrm{freq}(L)$. If $|w'|_u \notin \mathrm{freq}(L)$ for a vertex $u \in V(G)$, $\{ u, v \} \in E(G')$. This is a contradiction. Hence, for every $x \in V(G) \cup \{ v \}$, $|w'|_x \in \mathrm{freq}(L)$. Therefore, $G(\hat L, w') = G(L, w')$ and $G(L, h_{V(G)}(w')) = G$. This is a contradiction to $(*)$. Hence, the negation of $(*)$ is true, i.e., for every $G \in G_{\hat L}$, a word $w$ with $G = G(\hat L, w)$ exists such that, for every $t \in V(G)$, $|w|_t \in \mathrm{freq}(L)$. Hence, $G = G(L,w)$ and $G \in \cG_L$.  

\item Consider a graph $G = (V,E) \in \mathcal{G}_{\hat L}$. Hence, there exists a word $w \in V^*$ such that $G = G(\hat L, w)$. Partition $V$ into the sets  $V_1 = \{ v \in V \mid |w|_v \notin \mathrm{freq}(L) \}$ and $V_2 = \{ v \in V \mid |w|_v \in \mathrm{freq}(L) \}$. For each $u\in V$ and $v\in V \setminus \{ u \}$, the following holds: 
\begin{align*}
\{ u, v \} \in E \iff{} & h_{u,v}(w) \in \hat L \\
                 \iff{} & h_{u,v}(w) \in T_L \vee (h_{u,v}(w) \in L \wedge h_{u,v}(w) \notin T_L) \\
                 \iff{} & |w|_u \notin \mathrm{freq}(L) \vee |w|_v \notin \mathrm{freq}(L) \\
                                 & \vee (h_{u,v}(w) \in L \wedge |w|_u \in \mathrm{freq}(L) \wedge |w|_v \in \mathrm{freq}(L) )\\
                 \iff{} &  u \in V_1 \vee v \in V_1 \vee (h_{u,v}(w) \in L \wedge u \in V_2 \wedge v \in V_2)
\end{align*} 
This shows that $V_1$ is a set of universal vertices. For $u,v \in V_2$, $h_{u,v}(w) = h_{u,v}(h_{V_2}(w))$ and hence $G = G(L,h_{V_2}(w)) \nabla V_1$. Since $G(L,h_{V_2}(w)) \in \mathcal{G}_L$ and $\mathcal{G}_L$ is closed under adding universal vertices, $G \in \mathcal{G}_L$. 
\qed
\end{enumerate}\renewcommand{\qed}{}
\end{proof}

Further, observe that $T_L=T_{\overline{L}}$ \iffl $L\cap (0^\ell\shuffle 1^\ell)\neq 0^\ell\shuffle 1^\ell$ for all $\ell\in \mathrm{freq}(L)$.


\begin{remark}\label{rem:adding-isolates}
The trivial way to add isolated vertices to a word $w$ describing a graph $G(L,w)$ for a finite language $L$ is by choosing the number of occurrences of the new vertex $v$ outside of $\mathrm{freq}(L)$; see \Cref{propos:closure-adding-isolates}. However, this technique cannot be used to guarantee that $\cG_{\hat L}$ is closed under adding isolated vertices. Namely, $\cG_{\hat L}$ is closed under adding isolated vertices \iffl one can add an isolated vertex $v$ in such a way that the number of occurrences of $v$ is in $\mathrm{freq}(L)$. 
\end{remark}
}

\longversion{\section{Structure and Limitations of Language-Representable Graph Classes}
\label{sec:general-results}
}

\longversion{In this section, we show general results concerning representable graphs and graph classes defined by languages $L\subseteq\{0,1\}^*$.
First, we show a decomposition theorem that motivates to study graph classes represented by languages of uniform or nearly uniform lengths in more detail in the following as it shows that graph classes defined by finite languages can be analyzed by focusing on the (nearly) uniform ``pieces'' of the language. Then, we prove that not all interesting graph classes can be defined by languages in our framework. Here, we use two different type of arguments: (a) closure properties as derived in \Cref{sec:operations} that will be used in an iterative fashion, resembling pumping arguments known from formal languages, and (b) information-theoretic arguments that are only available for graph classes as we consider them through the lens of languages.

\subsection{A Decomposition Theorem for Language-Representable Graphs}

As a preparatory step to proving the decomposition theorem, we provide a statement on the type of graphs that can be represented by nearly length-uniform languages. 
We discussed $\cG_{\langle 001,010,011\rangle}$ as \Cref{exa:represented-graphs-CompleteBipartiteUnionNullGraphs} in \Cref{exa:represented-graphs}.

\begin{proposition} \label{prop:bipartite}
If $L$ is a nearly length-uniform language, then all graphs in~$\cG_L$ are bipartite.
\end{proposition}

\begin{proof} For $k,\ell \in \mathbb{N}$ such that $0 < k < \ell$, consider a $(k,\ell)$-uniform, $0$-$1$-symmetric language $L \subseteq \{ 0, 1 \}^*$. For $G \in \mathcal{G}_L$, choose $w \in V(G)^*$ such that $G(L,w) = G$. Let $V_i = \{ v \in V(G) \mid |w|_v = i\}$ for $i\in\{k, \ell \}$. Assume $V_k$ is not independent. Hence, $u,v \in V_k$ exist such that $\{ u, v \} \in E(G)$ and $h_{u,v}(w) \in L$.
$ |h_{u,v}(w)|_0 = |w|_u = k$ and $ |h_{u,v}(w)|_1 = |w|_v = k$. Therefore, $w \notin L$, because this is a contradiction to the $(k,\ell)$-uniformity of~$L$. Hence, $V_k$ is independent, i.e., $G[V_k]$ is a null graph. We show that $V_\ell$ is independent analogously. This proves that $L$ is bipartite.
For $\ell \in \mathbb{N}$ such that $0 < \ell$, the $(0,\ell)$-uniform, $0$-$1$-symmetric languages are $\{ 0^\ell, 1^\ell \}$ and $\emptyset$. Hence $\mathcal{G}_L$ is the class of null graphs. Those are bipartite.
\end{proof}

Combining this with the argument of \Cref{lem:clique-representation}, one obtains, generalizing \Cref{exa:represented-graphs-CompleteBipartiteUnionNullGraphs} of \Cref{exa:represented-graphs}:
\begin{lemma}\label{lem:biclique-representation}
For $0 < k < \ell$, if $L=0^k1^\ell$, then $G(L,w)=K_{n,m}\cup N_\ell$ for some $n,m,\ell\in\N$.
\end{lemma}

\Cref{prop:bipartite} is also one of the building blocks of the following type of decomposition theorem. To formulate it, consider for $w\in V^*$ and $k\leq\ell$ the morphism $g_{w,k,\ell}:V^*\to V^*$ defined for $x\in V$ by $x\mapsto x$ if $|w|_x\in\{k,\ell\}$ and $x\mapsto\emptyword$, otherwise. For $L \in \{ 0,1 \}^\ast$ and a word $w$, we define $E_{k,\ell}(L,w)$ by $G(L, g_{w,k,\ell}(w))=(V_{k,\ell}(w),E_{k,\ell}(L,w))$, i.e., $V_{k,\ell}(w)=\{x\in V\mid |w|_x\in\{k,\ell\}\}$. For $k = \ell$, we define $V_k(w) = V_{k,k}(w)$. If the language $L$ and the word $w$, or the word $w$ respectively, are clear from the context, we might simply write $E_{k,\ell}$, $V_{k,\ell}$ and $V_k$\longversion{ for the respective sets}.

\begin{theorem}\label{thm:graph-decomposition}
 Let $L\subseteq \{0,1\}^*$. 
 Let $G=(V,E)\in \cG_L$ be $L$-represented by $w\in V^*$.
 Let $\textbf{O1}_{L,w}=\{(k,\ell)\in\N^2\mid 0 < k\leq \ell\land \exists x,y\in V:h_{x,y}(w)\in ((0^k\shuffle 1^\ell)\cup (0^\ell\shuffle 1^k))\cap L\}$. Then, $$E=\bigcup_{(k,\ell)\in \textbf{O1}_{L,w}}E_{k,\ell}\,.$$
 Moreover,  $G_{k,\ell}=(V_{k,\ell},E_{k,\ell})\in \cG_{((0^k\shuffle 1^\ell)\cup (0^\ell\shuffle 1^k))\cap L}$.
\end{theorem}

\begin{proof}
For all $(k,\ell) \in \textbf{O1}_{L,w}$, $G[V_{k,\ell}] = G(L,g_{w,k,\ell}(w))= (V_{k,\ell},E_{k,\ell}) $, because $g_{w,k,\ell} = h_{V_{k,\ell}}$. This follows from \Cref{lem:induced_subgraph}. Hence, $G_{k,\ell}$ is well-defined and $E_{k,\ell} =E \cap \binom{V_{k,\ell}}{2}$. 

To prove $E=\bigcup_{(k,\ell)\in \textbf{O1}_{L,w}}E_{k,\ell}$, we only need to show the inclusion $\subseteq$. Let $\{ u, v \} \in E$. W.l.o.g., assume $|w|_u \leq |w|_v$. In this case, $(|w|_u,|w|_v) \in \textbf{O1}_{L,w}$ and $u, v \in V_{|w|_u,|w|_v}$. Hence $\{ u, v \} \in E_{|w|_u,|w|_v} \subseteq \bigcup_{(k,\ell)\in \textbf{O1}_{L,w}}E_{k,\ell}$. 

For all $(k,\ell) \in \textbf{O1}_{L,w}$ and for all $u, v \in V_{k,\ell}$ such that $u \neq v$, $h_{u,v}(g_{w,k,\ell}(w)) \in (0^k\shuffle 1^\ell)\cup (0^\ell\shuffle 1^k)$. Hence, $h_{u,v}(g_{w,k,\ell}(w)) \in L$ \iffl $h_{u,v}(g_{w,k,\ell}(w)) \in ((0^k\shuffle 1^\ell)\cup (0^\ell\shuffle 1^k)) \cap L$. This implies $G_{k,\ell} = G(((0^k\shuffle 1^\ell)\cup (0^\ell\shuffle 1^k))\cap L, g_{w,k,\ell}(w))$ and therefore $G_{k,\ell}\in \cG_{((0^k\shuffle 1^\ell)\cup (0^\ell\shuffle 1^k))\cap L}$. 
\end{proof}

\begin{example}\label{exa:split} 
Consider $L=\langle 001,01\rangle$. By \Cref{thm:graph-decomposition} and \Cref{exa:represented-graphs-CompleteUnionNullGraphs} of \Cref{exa:represented-graphs}, $\cG_L$ is a family of split graphs. More precisely, if $G=(V,E)\in \cG_L$, represented by some $w\in V^*$, then $E=E_{1,1}\cup E_{1,2}$. In particular, all vertices that occur at least thrice in $w$ have degree zero. Hence, the projective morphism that only preserves vertices that occur once in~$w$, let $V_1\subseteq V$ be this set, allows us to apply \Cref{lem:induced_subgraph}, concluding that $G[V_1]=G(L,h_{V_1}(w))$ is a complete graph, while (using a projective morphism preserving vertices from $V\setminus V_1$) $G[V\setminus V_1]=G(L,h_{V\setminus V_1}(w))$ is a null graph, so that~$G$ is a split graph. 
More details on this graph class can be found in \Cref{subsec:atmost-two-occurrences}.
\end{example} 

\Cref{thm:graph-decomposition} motivates to focus on (nearly) length-uniform languages, as graphs described with the help of other, more complex languages can be decomposed in the way described by the preceding theorem.
Also, as (nearly) length-uniform languages are finite, this motivates the study of finite~$L$ and their graph classes.

\subsection{Limitations of Language-Representable Graphs}}
\shortversion{\section{Limitations of Language-Representable Graphs}}

In this \longversion{sub}section, we will use some information-theoretic arguments to show why, for certain classes of graphs, finite languages to not suffice to represent them. Recall \Cref{thm:hereditary}: all $L$-representable graphs are hereditary; therefore, these arguments also relate to the theory of growth classes of hereditary graph families as developed in~\cite{SchZit94}.
First, we will study finite languages and see that none of them is able to represent all graphs.

\begin{theorem}\label{thm:counting-Lfinite}
If $L\subseteq\{0,1\}^*$ is finite, then $\cG_L$ does not contain all graphs.    
\end{theorem}

Actually, in this setting, there are `more' graphs that are not $L$-representable than there are graphs that are $L$-representable.
The reasoning is based on information theory, asking how many graphs of order~$n$ can be represented as $G(L,w)$ if $L$ is finite. \longversion{For our argument, the following lemma is crucial.}\shortversion{Hence, we need:} 

\begin{lemmarep} \shortversion{$(*)$}
There are $2^{\Theta(n^2)}$ many non-isomorphic graphs of order~$n$.
\end{lemmarep}

\begin{proof} 
Unfortunately, we did not find an explicit link of this assertion in the literature. Hence, we provide a proof here.
Let $\text{NI}(n)$ denote the number of non-isomorphic graphs of order~$n$. This sequence of numbers is pretty famous and hence has a relatively low number in the Online Encyclopedia of Integer Sequences (\href{https://oeis.org/A000088}{A000088}). Its study can be traced back to a paper of Pólya, further refined, e.g., by Oberschelp and Wright, see \cite{Obe67,Pol37,Wri69}. For our considerations, it is sufficient to observe the trivial relation (according to Oberschelp):
$$2^{\binom{n}{2}}/n!\leq \text{NI}(n)\leq 2^{\binom{n}{2}}\,.$$
By Stirling's formula, $n!\in 2^{\Theta(n\log n)}$. Also, $\binom{n}{2}\in \Theta(n^2)$. Hence, $$2^{\binom{n}{2}}/n!\in \frac{2^{\Theta(n^2)}}{2^{\Theta(n\log n)}}=2^{\Theta(n^2)-\Theta(n\log n)}=2^{\Theta(n^2)}\,.$$
This proves that $\text{NI}(n)\in 2^{\Theta(n^2)}$ as claimed.
\end{proof}

\begin{proof}[Proof of \Cref{thm:counting-Lfinite}] We are going to give an information-theoretic argument.
If $L$ is finite, then for each $G=(V,E)\in \cG_L$, there is a word $w\in V^*$ of length at most $|L|\cdot \text{ml}(L)\cdot |V|$ (where $\text{ml}(L)$ should denote the maximum length of any word in~$L$) with $G(L,w)=G$. As $L$ is fixed, $|L|\cdot \text{ml}(L)$ can be considered as a constant. Therefore, writing down $w$ costs at most $\cO(|V|\log(|V|))$ many bits.
Hence, at most $2^{\cO(|V|\log(|V|))}$ many graphs of $|V|$ many vertices can belong to $\cG_L$, while there are $2^{\Theta(n^2)}$ many non-isomorphic graphs on $n$ vertices.
\end{proof}
With the same argument, we can conclude: 
\begin{corollaryrep}\label{cor:non-finitely-presentable graph classes} \shortversion{$(*)$}
If $L\subseteq\{0,1\}^*$ is finite, then $\cG_L$ does not contain all bipartite graphs, or all split graphs, or all cobipartite graphs, or all word-representable graphs, or all comparability graphs or all chordal graphs. 
\end{corollaryrep} 
\begin{proof}
Namely, if there are $g(n)$ many non-isomorphic graphs of order~$n$, then there are at least $g(n)$ many  non-isomorphic bipartite (or split, or cobipartite) graphs of order $2n$. Consider the injective mapping that associates to $G=(V,E)$ the bipartite (or split, or cobipartite) graph $G'=(V_1\cup V_2,E')$, where $V_i=\{v_i\mid v\in V\}$ for $i=1,2$ and $\{u_1,v_2\}\in E'$ \iffl $\{u,v\}\in E$. If these are the only edges of $E'$, then $G'$ is bipartite. If $E'$ contains, in addition, all edges between vertices of $V_1$, we get a split graph and, if $E'$ also contains all edges between vertices of $V_2$, we obtain a cobipartite graph. As there are $2^{\Theta(n^2)}$ many non-isomorphic graphs on $n$ vertices, by the given injections, there are also $2^{\Theta(n^2)}$ many non-isomorphic bipartite (or split, or cobipartite) graphs on $n$ vertices. As chordal graphs generalize split graphs, this holds also for this graph class.
According to \cite{KleRot70}, there are  $2^{\Theta(n^2)}$ many different partial orders on an $n$-element set. By \cite{Moh84}, the number of comparability graphs of order~$n$ is of the same order of magnitude as the number of partial orders on an $n$-element set.
Finally, in \cite{ColKitLoz2017}, it was shown that there are $2^{\Theta(n^2)}$ many non-isomorphic word-representable graphs on $n$ vertices.
The results then follow by the information-theoretic argument of the previous proof. 
\end{proof}

This also explains why the languages $L_{\overline{1^k}}$ for $k \geq 3$, $L_{k\text{-}11}$ for $k\geq 2$, and other (so far) known languages capable of describing all graphs, must be infinite.
This also motivates to look at other infinite languages. \longversion{Moreover, it becomes clear why \Cref{exa:split} does not claim a characterization of split graphs. Conversely, this creates the intriguing question to characterize all these language families that can be already represented by finite languages. Recall that all of them are hereditary. Do we know at least some of them from other contexts?}


Next, we are looking at properties of language-representable graph classes that show that various well-established graph classes like forests or planar graphs cannot be represented by any language. In the arguments, we will repeatedly use \Cref{thm:twins} as a closure property of all language-representable graph classes. 
\longversion{To this end, first we considerably strengthen \Cref{prop:true-twins} towards the next statement.}

\begin{theorem}\label{thm:bounded-treewidth}
Let $L\subseteq\{0,1\}^*$. The following four conditions are equivalent.
\shortversion{(1) There is some $t\geq 0$ such that all graphs in $\cG_L$ have treewidth bounded by~$t$. (2) All graphs in $\cG_L$ have treewidth~0. (3) $\cG_L=\{N_n\mid n\in\N_{\geq1}\}$. (4) $L\subseteq 0^*\cup 1^*$.}
\longversion{\begin{enumerate}
    \item There is some $t\geq 0$ such that all graphs in $\cG_L$ have treewidth bounded by~$t$. \item  All graphs in $\cG_L$ have treewidth~0. \item $\cG_L=\{N_n\mid n\in\N_{\geq1}\}$. \item  $L\subseteq 0^*\cup 1^*$.
\end{enumerate}}
\end{theorem}

\begin{proof}
$(1)\implies (2)$:  If all graphs in $\cG_L$ have treewidth at most~$t>0$ and some graph~$G$ in $\cG_L$ has treewidth at least one, then this graph contains an edge $\{u,v\}$. Now, as $\cG_L$ is closed under adding twins, we can first add $t$ twins~$u_1,\dots, u_{t}$ of~$u$, i.e., in the resulting graph $G'$, $V(G')=V(G)\cup \{u_1,\dots,u_{t}\}$ and $N_{G'}[u]=N_{G'}[u_i]$ for $i\in [t]$. Then, we add $t$ twins~$v_1,\dots, v_{t}$ of~$v$, i.e., in the resulting graph $G''$,  $V(G'')=V(G')\cup \{v_1,\dots,v_{t}\}$ and $N_{G''}[v]=N_{G''}[v_i]$ for $i\in [t]$. Now, $G''[\{u_i,v_i\mid i\in[t]\}\cup\{u,v\}]$ contains $K_{t+1,t+1}$ as a subgraph, a graph of treewidth $t+1$. By construction, $G''\in \cG_L$ according to \Cref{thm:twins}, contradicting our assumption that all graphs in~$\cG_L$ have treewidth at most~$t>0$. Therefore, if all graphs in $\cG_L$ have bounded treewidth, then all these graphs are null graphs. 
$(2)\implies (1)$: trivial.
$(2)\implies (3)$: If all graphs in $\cG_L$ have treewidth~0, then these graphs must all be null graphs. As $\cG_L\neq\emptyset$ and $\cG_L$ is hereditary, $N_1\in\cG_L$. As $\cG_L$ is closed under adding some twins and as adding true twins would introduce edges into graphs of $\cG_L$, we know that $\cG_L$ is in fact closed under adding false twins. \longversion{By \Cref{prop:false-twins}}\shortversion{Hence}, $\cG_L=\{N_n\mid n\in\N_{\geq1}\}$. 
$(3)\implies (2)$:\shortversion{ trivial.}\longversion{ Clearly, $\{N_n\mid n\in\N_{\geq1}\}$ has treewidth~0.} 
$(4)\implies (3)$: Extending \Cref{exa:represented-graphs-NullGraphs} of \Cref{exa:represented-graphs}, it is clear that $\cG_L=\{N_n\mid n\in\N_{\geq1}\}$ if $L\subseteq 0^*\cup 1^*$.
$(3)\implies (4)$ by contraposition: If $L\cap (0^+\shuffle 1^+)\neq\emptyset$, then there exists some pattern~$p\in L$ with $|p|_0>0$ and $|p|_1>0$. Hence, considering $V=\{0,1\}$, $G(L,p)=(V,\{\{0,1\}\})\simeq K_2\in \cG_L$, so that $\cG_L\neq \{N_n\mid n\in\N_{\geq1}\}$.
\end{proof}

This already rules out many interesting graph classes. We assume the reader to be familiar with notions like $k$-outerplanar and series-parallel.

\begin{corollaryrep}\label{cor:treewidth} \shortversion{$(*)$}
There is no binary language $L$ such that $\cG=\cG_L$ for any of the following graph classes~$\cG$:
\longversion{\begin{itemize}
    \item forests, i.e., graph unions of trees;
    \item forests of paths, i.e., graph unions of paths;
    \item for any $k\geq 1$: $k$-outerplanar graphs;
    \item forests of series-parallel graphs.
\end{itemize}}\shortversion{(a) forests, i.e., graph unions of trees;
    (b) forests of paths, i.e., graph unions of paths;
    (c) for any $k\geq 1$: $k$-outerplanar graphs;
    (d) forests of series-parallel graphs.}
\end{corollaryrep}

\begin{proof}
    Forests can be characterized the graphs of treewidth at most~1.
    Similarly, forests of series-parallel graphs are exactly the graphs of treewidth at most~2. Also, union of paths and $k$-outerplanar graphs have bounded treewidth (to be more specific, at most~1 and at most $3k-1$, respectively, see~\cite{Bod98}). Hence, the claims follow with \Cref{thm:bounded-treewidth}.
\end{proof}

\Cref{thm:twins} has a further interesting consequence (we make only one further explicit in this place, in fact it has a number of important consequences). Now, we assume the reader is acquainted with notions like degeneracy and genus of a graph.

\begin{theorem} \label{thm:d-degenerate}
 Let $d$ be a fixed positive integer. Then, the class  $d$-$\cG$ of all $d$-degenerate graphs cannot be characterized as $\cG_L$ for any binary language~$L$, nor can any hereditary subclass of $d$-$\cG$ apart from the class of null graphs. 
\end{theorem} 

\begin{proof}
    Let $d\in \N \setminus \{ 0 \}$ and $\mathcal{G}\subseteq d\text{-}\mathcal{G}$ be a graph class. Assume there is a language $L\subseteq \{0,1\}^{\ast}$ such that $\cG=\cG_L$. Let $G=(V,E)=G(L,w)\in \cG$ be a graph with at least one edge and  a word $w\in V^{\ast}$. We can assume~$G$ has no isolated vertices, as otherwise we can delete the vertex from the word and the graph (also cf. \Cref{thm:hereditary}). Since $G$ is $d$-degenerate, there exists a $v \in V$ with $1 \leq \deg(v)\leq d$. Let $u\in N(v)$. Now add $u_1, \dots, u_d$ as twins of~$u$ and $v_1,\dots,v_d$ as twins of~$v$  to~$G$ as in \autoref{thm:twins}. Since for all $i\in [d]$, $\deg(v), \deg(v_i), \deg(u), \deg(u_i)>d$ and the degree of the remaining vertices is never decreasing, the number of vertices with degree at most~$d$ is strictly decreasing.
    By performing this step inductively, we obtain a graph in~$\cG$ such that no vertex has degree at most~$d$. This contradicts $\cG \subseteq d$-$\cG$. 
\end{proof}

The last theorem allows to make important notes about the non-representability of various well-known graph classes.\longversion{ Some are listed already in \Cref{cor:treewidth}, but here we can show this with a different proof.}

\begin{corollaryrep}\label{cor:degeneracy} \shortversion{$(*)$}
There is no binary language $L$ such that $\cG=\cG_L$ for any of the following graph classes~$\cG$:
\longversion{\begin{itemize}
    \item forests, i.e., graph unions of trees;
    \item forests of paths, i.e., graph unions of paths;
    \item for any $d\geq 1$: graphs of maximum degree~$d$, especially
    forests of paths and cycles;
    \item for any $g\geq 0$: graphs of maximum genus~$g$, especially
    planar graphs.
\end{itemize}}\shortversion{(a) for any $d\geq 1$: graphs of maximum degree~$d$, especially forests of paths and cycles; (b) for any $g\geq 0$: graphs of maximum genus~$g$, especially planar graphs.}
\end{corollaryrep}
\begin{proof}
It is well-known that forests (including forests of paths) are 1-degenerate. In fact, forests can be both characterized as the 1-degenerate graphs. Trivially, graphs with maximum degree~$d\in \N_{\geq 1}$ are are $d$-degenerate. Particularly forest of cycles are 2-degenerate graphs. By Euler's formula, every planar graph has a vertex of degree 5 or less  and therefore planar graphs are $5$-degenerate.  Hence $d \in \N_{\geq 1}$ exists such that $\cG \subseteq d$-$\cG$. The statement follows from \Cref{thm:d-degenerate}. A similar argument applies to graphs of any bounded genus.
\end{proof}

The results of this section also have some interesting algorithmic consequences for decision problems that one might like to solve with our framework.
\begin{theorem}\label{thm:deciding-bounded-treewidth-and-degeneracy}
There is an algorithm that, given a context-free grammar~$G$, decides if the language $L$ generated by~$G$ represents a graph class $\cG_L$ that has bounded treewidth or that has bounded degeneracy.
\end{theorem}

\begin{proof}
Consider first the treewidth question that is (more formally) the following one:
Given  a context-free grammar~$G$ (generating~$L$), does there exist some integer $k>0$ such that, for all graphs $H\in\cG_L$, $\text{tw}(H)\leq k$?
According to \Cref{thm:bounded-treewidth}, this is equivalent to the question if $L\subseteq 0^*\cup 1^*$. This in turn is equivalent to asking if $L\cap \overline{0^*\cup 1^*}=\emptyset$. As $\overline{0^*\cup 1^*}$ is a \longversion{(even star-free)} regular language and as the context-free languages are effectively closed under intersection with regular languages, one can construct, from $G$, another context-free grammar $G'$ that generates the language $L\cap \overline{0^*\cup 1^*}$. Now, recall that the emptiness problem for context-free grammars is decidable (see \cite{HopUll79}), i.e., there exists an algorithm that decides, given $G'$, if its language is empty. This algorithm clearly solves our original problem. 

Concerning degeneracy, observe that the arguments given in \Cref{thm:d-degenerate} lead to the conclusion that a language-representable graph class of bounded degeneracy must contain null graphs only, so that we can take the same decision procedure as described for the bounded treewidth question also in this case for the question of bounded degeneracy. 
\end{proof}

\longversion{\section{Special Representable Graph Classes from Finite Languages}
\label{sec:special-results}
\begin{figure}
\centering

\begin{tabular}{| c | c | c |}
\hline 
$L$ & $\mathcal{G}_L$ & reference \\
\hline 
$\langle 01 \rangle$ & $\{K_n\cup N_m\mid n,m\in\N\}$ & \Cref{exa:represented-graphs} (\ref{exa:represented-graphs-CompleteUnionNullGraphs}) \\
\hline
$\langle 001 \rangle$ & bipartite chain graphs & \Cref{thm:bipartite-chain}\\
$\langle 010 \rangle$ & convex graphs & \Cref{thm:convex} \\
$\langle 011 \rangle$ & bipartite chain graphs & \Cref{cor:bipartite-chain2} \\
$\langle 010,011 \rangle$ & bipartite chain graphs & \Cref{cor:bica-classes}\\
$\langle 001, 011\rangle$ & bico-convex graphs & \Cref{cor:bica-classes}\\
$\langle 010,001 \rangle$ & bipartite chain graphs & \Cref{cor:bica-classes} \\
$\langle 001,010,011 \rangle$ & $\{K_{n,m}\cup N_\ell\mid n,m,\ell\in\N\}$ & \Cref{exa:represented-graphs} (\ref{exa:represented-graphs-CompleteBipartiteUnionNullGraphs}) \\
\hline
$\langle 0011 \rangle$ & co-interval graphs & \Cref{cor:co-intervalgraphs} \\
$\langle 0101 \rangle$ & circle graphs & \cite[Thm. 5.1.7]{KitLoz2015} \\
$\langle 0110 \rangle$ & permutation graphs & \Cref{thm:permGraphs} \\
$\langle 0101,0110 \rangle$ & interval graphs & \Cref{thm:intervalgraphs} \\
$\langle 0011,0110 \rangle$ & $\{ G \cup N_m \mid G \text{ is a co-circle graph}, m \in \mathbb{N} \}$ & \Cref{thm:co-circle} \\
$\langle 0011,0101 \rangle$ & permutation graphs & \Cref{cor:co-permutation} \\
$\langle 0011,0101,0110 \rangle$ & $\{K_n\cup N_m\mid n,m\in\N\}$ & \Cref{lem:00shuffle11} \\
\hline
\end{tabular}
\caption{Graph classes represented by certain length-uniform and nearly length-uniform languages} \label{tab:graphClasses}
\end{figure}
\Cref{tab:graphClasses} shows the results of this section for $1$-, $(1,2)$- and $2$-uniform languages. Notice that mostly, we can characterize graph classes known from the graph theory literature, but sometimes (namely, when this class is not closed under adding isolated vertices), we have to make unions of null graphs explicit due to \Cref{propos:closure-adding-isolates}, which applies to all finite languages and their related graph classes.

\subsection{Interval models with respect to $(1,2)$- and $2$-uniform languages}

In this subsection, we show that there is a uniform way to interpret graph classes defined by $(1,2)$- and $2$-uniform languages~$L$ geometrically. This explains why we encounter many well-known graph classes as such classes~$\cG_L$.
Interestingly, not all of these geometric models have been described before to the best of our knowledge.

\begin{definition} \label{def:interval-model}
Let $L$ be a $(1,2)$-uniform or a $2$-uniform language such that $\emptyset \subsetneq L \subsetneq (0 \shuffle 11 \cup 00 \shuffle 1)$, or $\emptyset \subsetneq L \subsetneq 00 \shuffle 11$ respectively. Let $\mathcal{I} = \{ I_v \}_{v \in V}$ be a family of closed intervals on a linearly ordered set $\mathcal{D}$ such that for each interval $I \in \mathcal{I}$ an endpoint of $I$ is never an endpoint of another interval.\longversion{ We allows intervals where left and right endpoint coincide.} We say that $\mathcal{I}$ \emph{ is an interval model of $G$ with respect to the language $L$} if, for each $v \in V$ and for each $u \in V \setminus \{ v \}$, $\{u,v\} \in E$ \iffl the pair $(I_u,I_v)$ matches the pattern of a word contained in $L$. The pair $([x_\alpha, x_\beta],[y_\alpha, y_\beta])$ \emph{matches the pattern of the word $w \in L$} if 
\begin{itemize}
\item $ x_\alpha < x_\beta < y_\alpha < y_\beta$ for $w = 0011$,
\item $x_\alpha <y_\alpha < x_\beta < y_\beta$ for $w = 0101$,
\item $x_\alpha < y_\alpha < y_\beta < x_\beta$ for $w = 0110$,
\item $x_\alpha < x_\beta < y_\alpha = y_\beta$ for $w = 001$, 
\item $x_\alpha < y_\alpha = y_\beta < x_\beta$ for $w = 010$,
\item $x_\alpha = x_\beta < y_\alpha < y_\beta$ for $w = 011$, or
\item $([y_\alpha, y_\beta],[x_\alpha, x_\beta])$ matches the pattern of the word $\widetilde{w}$. 
\end{itemize} 
\end{definition}
The different kinds of interval models yielded by \Cref{def:interval-model} are illustrated in \Cref{fig:interval-model}. 
\begin{figure} \centering
\begin{tikzpicture}
\newcommand{\dy}{0.1}
\newcommand{\drawinterval}[3]{\draw (#1,#3+\dy) -- (#1,#3-\dy) -- (#1,#3) -- (#2,#3) -- (#2,#3+\dy) -- (#2,#3-\dy);}

\draw (0,1) node[anchor=west,align=left]{0 0 1 1};
\drawinterval{0.1}{0.6}{0.7}
\drawinterval{0.7}{1.2}{0.7}
\draw (0,0) node[anchor=west,align=left]{disjoint\\intervals};

\draw (4,1) node[anchor=west,align=left]{0 1 0 1};
\drawinterval{4.1}{4.9}{0.7}
\drawinterval{4.4}{5.2}{0.5}
\draw (4,0) node[anchor=west,align=left]{intersecting intervals\\but no inclusion};

\draw (8,1) node[anchor=west,align=left]{0 1 1 0};
\drawinterval{8.1}{9.2}{0.7}
\drawinterval{8.4}{8.9}{0.5}
\draw (8,0) node[anchor=west,align=left]{one interval\\contains the other};

\draw (0,-1) node[anchor=west,align=left]{0 0 1};
\drawinterval{0.1}{0.6}{-1.3}
\filldraw[black] (0.8,-1.3) circle (2pt);
\draw (0,-2) node[anchor=west,align=left]{point to the right\\of the interval};

\draw (4,-1) node[anchor=west,align=left]{0 1 0};
\drawinterval{4.1}{4.9}{-1.3}
\filldraw[black] (4.5,-1.5) circle (2pt);
\draw (4,-2) node[anchor=west,align=left]{point contained\\in the interval};

\draw (8,-1) node[anchor=west,align=left]{0 1 1};
\drawinterval{8.4}{8.9}{-1.3}
\filldraw[black] (8.2,-1.3) circle (2pt);
\draw (8,-2) node[anchor=west,align=left]{point to the left\\of the interval};
\end{tikzpicture} 
\caption{Interval pairs matching words from $(1,2)$- and $2$-uniform languages (with points $p$ considered intervals of the form $[p,p]$)} \label{fig:interval-model}
\end{figure}
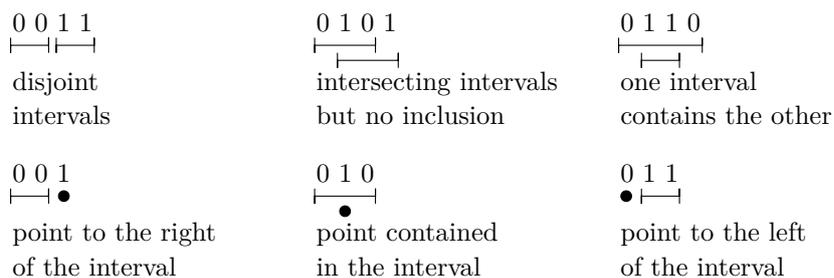

\begin{example} Consider\longversion{ the language} $L = \langle 010 \rangle$. For each graph described by an interval model with respect to $L$, the vertices can be partitioned into the vertices described by intervals that are points and intervals that are not points. There are only edges between those two vertex sets. A vertex described by a point is adjacent to a vertex described by an interval that is not a point \iffl the point is contained in the interval. This is illustrated in \Cref{fig:interval-model-convex}.
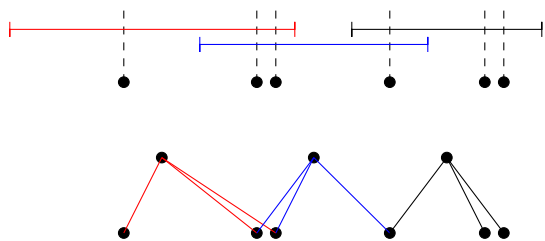
\begin{figure} \begin{center}
\begin{tikzpicture}
\newcommand{\dy}{0.1}
\newcommand{\drawinterval}[4]{\draw[#1] (#2,#4+\dy) -- (#2,#4-\dy) -- (#2,#4) -- (#3,#4) -- (#3,#4+\dy) -- (#3,#4-\dy);}

\filldraw (1.50,0) circle (2pt); \draw[dashed] (1.50,0) -- (1.50,1);
\filldraw (3.25,0) circle (2pt); \draw[dashed] (3.25,0) -- (3.25,1);
\filldraw (3.50,0) circle (2pt); \draw[dashed] (3.50,0) -- (3.50,1);
\filldraw (5.00,0) circle (2pt); \draw[dashed] (5.00,0) -- (5.00,1);
\filldraw (6.25,0) circle (2pt); \draw[dashed] (6.25,0) -- (6.25,1);
\filldraw (6.50,0) circle (2pt); \draw[dashed] (6.50,0) -- (6.50,1);

\drawinterval{red}  {0.0}{3.75}{0.7}
\drawinterval{blue} {2.5}{5.50}{0.5}
\drawinterval{black}{4.5}{7.00}{0.7}

\filldraw (1.50,-2) circle (2pt);
\filldraw (3.25,-2) circle (2pt);
\filldraw (3.50,-2) circle (2pt);
\filldraw (5.00,-2) circle (2pt);
\filldraw (6.25,-2) circle (2pt);
\filldraw (6.50,-2) circle (2pt);

\filldraw (2.00,-1) circle (2pt);
\draw[red] (2.00,-1) -- (1.50,-2);
\draw[red] (2.00,-1) -- (3.25,-2);
\draw[red] (2.00,-1) -- (3.50,-2);

\filldraw (4.00,-1) circle (2pt);
\draw[blue] (4.00,-1) -- (3.25,-2);
\draw[blue] (4.00,-1) -- (3.50,-2);
\draw[blue] (4.00,-1) -- (5.00,-2);

\filldraw (5.75,-1) circle (2pt);
\draw[black] (5.75,-1) -- (5.00,-2);
\draw[black] (5.75,-1) -- (6.25,-2);
\draw[black] (5.75,-1) -- (6.50,-2);
\end{tikzpicture}
\end{center}
\caption{An interval model with respect to\longversion{ the language} $\langle 010 \rangle$ and the graph described by the model} \label{fig:interval-model-convex}
\end{figure}
\end{example}

The following lemma is not only interesting for understanding when we can build interval models, see \Cref{thm:interval-model}, but it is also crucial for the graph class property of adding isolated vertices as explained in \Cref{rem:adding-isolates} in connection with the the Trash \Cref{trash-theorem}. Also confer \Cref{prop:2-uniform-plus-vertices}.

\begin{lemma} \label{lem:uniform_word}
\begin{enumerate}
\item For every $L \subsetneq (0 \shuffle 1^2) \cup (1 \shuffle 0^2) $ and every $G \in \cG_L$, there exists a word $w \in V(G)^\ast$ such that $G = G(L,w)$ and $|w|_v \in \{ 1, 2 \}$ for every $v \in V(G)$. 
\item For every $L \subsetneq 00 \shuffle 11 $ such that $0110 \notin L$ or $0011 \notin L$ and every $G \in \cG_L$, there exists a word $w \in V(G)^\ast$ such that $G = G(L,w)$ and $|w|_v = 2$ for every $v \in V(G)$. 
\end{enumerate}
\end{lemma}

\begin{proof}
We only prove the first statement. The second statement can be proven analogously. 

Let $G \in \cG_L$. Hence there exists a word $w' \in V(G)^\ast$ such that $G = G(L,w')$. Define $V' = \{ v \in V(G) \mid |w'|_v \notin \{ 1, 2\} \}$ and let $V' = \{ v_1, \dots, v_k \}$ for $k = |V'|$. All the vertices in $V'$ are isolated. Define $$
w=
\begin{cases}
v_1^2 \cdots v_k^2 \cdot h_{V(G) \setminus V'}(w'), & \text{if } 001 \notin L\\
v_1 \cdots v_k \cdot h_{V(G) \setminus V'}(w') \cdot v_1 \cdots v_k, & \text{if } 010 \notin L \wedge 001 \in L\\
h_{V(G) \setminus V'}(w') \cdot v_1^2 \cdots v_k^2, & \text{otherwise.}
\end{cases}
$$
If $001 \notin L$ and $010 \notin L$, then $011 \notin L$, because $L \neq 0 \shuffle 1^2 \cup 1 \shuffle 0^2 $ and $L$ is 0-1-symmetric. Hence $G(L,w) = G(L,w') = G$ and $|w|_v \in \{ 1, 2 \}$ for every $v \in V(G)$.
\end{proof}



The necessity of the conditions  $L \neq (0 \shuffle 1^2) \cup (1 \shuffle 0^2) $ and  $L \neq 00 \shuffle 11 $ in \Cref{lem:uniform_word} is easy to see if one observes that 
$$(0 \shuffle 1^2) \cup (1 \shuffle 0^2)=\langle  001,010,011\rangle $$ and that
$$00 \shuffle 11=\langle 0011, 0101,0110\rangle\,.$$
Consulting \Cref{tab:graphClasses}, one sees that the corresponding graph classes contain the complete (bipartite) graphs and all null graphs. Non-trivial null graphs can be only represented by letters whose frequentness is different from 1 and 2 (in the case of $(0 \shuffle 1^2) \cup (1 \shuffle 0^2)$) or different from~2 (in the case of $00 \shuffle 11$). The necessity of the condition ``$0110 \notin L$ or $0011 \notin L$'' is explained in the following, less trivial example.

\begin{example}\label{exm:very_long_example}
    Interestingly, there are graphs $G\in \cG_L$ for $L \coloneqq \langle 0011,0110\rangle$ for which there exists no 2-uniform word $w$ such that $G=G(L,w)$. 

    One such example is $G=(V \coloneqq \{a,b,c,d,e,f\},E \coloneqq\{\{a,b\},\{b,c\},\{c,d\},\{d,e\},\{e,a\}\})$. $G$ is isomorphic to $C_5\cup K_1$. 
    First of all, $G=G(L,eacdabdebcf)$. Assume, for the sake of contradiction, that there is a 2-uniform word $w$ with $G=G(L,w)$. 

    As $I_1\coloneqq \{a,c,f\}$ is a independent set, $$h_{I_1}(w)\in \{acfacf, afcafc, cafcaf, cfacfa, facfac, fcafca\}\,.$$ The function $h: V \to V$ with $h(x)=x$ for $x\in \{ b, f\}$, $h(c)=a$, $h(a)=c$, $h(d)=e$ and $h(e)=d$, is a graph automorphism, i.e., an isomorphism from $G$ to~$G$. Further, $L$ is 0-1-symmetric. By \Cref{lem:isomorphism} and \Cref{cor:L-symmetry}, we only need to consider $h_{I_1}(w) \in \{acfacf, afcafc\}$. This gives two main cases with a number of subcases.
        
    \noindent \textbf{Case 1.} $h_{I_1}(w) = acfacf$. Consider $I_2 \coloneqq \{a,d,f\}$.
    
    \noindent \textbf{Case 1.a} $h_{I_2}(w) = dafdaf$.  Then $h_{I_1 \cup I_2}(w)= dacfdacf$ and $h_{d,c}(w)=0101$.
     
    \noindent \textbf{Case 1.b} $h_{I_2}(w) = adfadf$.  Then $h_{I_1 \cup I_2}(w) \in \{adcfacdf, acdfadcf\}$ as otherwise $h_{d,c}(w)=0101$. The function $h': V \to V$ with $h'(x)=x$ for $x\in \{ a, f\}$, $h'(c)=d$, $h'(d)=c$, $h'(b)=e$ and $h'(e)=b$, is a graph automorphism. Since $h(adcfacdf) = acdfadcf$, we only need to consider $h_{I_1 \cup I_2}(w) = adcfacdf$. Consider $I_3 \coloneqq \{ b,d,f\}$
     
    \noindent \textbf{Case 1.b.I} $h_{I_3}(w) = bdfbdf$. $h_{I_1 \cup I_3}(w) = badcfacbdf$ , as otherwise $h_{a,b}(w)=0101$ or $h_{b,c}(w)=0101$. Consider $I_4 \coloneqq \{e,b,f\}$.
     
    \noindent \textbf{Case 1.b.I.i} $h_{I_4}(w) = ebfebf$. Thus, $h_{e,d}(w) = 0101$.
     
    \noindent \textbf{Case 1.b.I.ii} $h_{I_4}(w) = befbef$. Hence, $h_{a,e}(w) = 0101$ or $h_{c,e}(w) = 0110$. 
     
    \noindent \textbf{Case 1.b.I.iii} $h_{I_4}(w) = bfebfe$. Then, $h_{d,e}(w) = 0101$. 
     
    \noindent \textbf{Case 1.b.II} $h_{I_3}(w) = dbfdbf$. This implies $h_{a,b}(w)=0101$.
     
    \noindent \textbf{Case 1.b.III} $h_{I_3}(w) = dfbdfb$. $h_{I_1 \cup I_3}(w) = adcfacbdfb$ , as otherwise $h_{a,b}(w)=0101$ or $h_{b,c}(w)=0101$. Consider $I_4 \coloneqq \{e,b,f\}$.

    \noindent \textbf{Case 1.b.III.i} $h_{I_4}(w) = efbefb$. Thus, $h_{a,e}(w)=0101$ or $h_{c,e}(w)=0110$.   
     
    \noindent \textbf{Case 1.b.III.ii} $h_{I_4}(w) = febfeb$. Therefore,  $h_{d,e}(w)=0101$.  
     
    \noindent \textbf{Case 1.b.III.iii} $h_{I_4}(w) = fbefbe$. This implies $h_{c,e}(w)=0011$. 
     
    \noindent \textbf{Case 1.c} $h_{I_2}(w) = afdafd$. So, $h_{c,d}(w)=0101$.
    
    \noindent \textbf{Case 2} $h_{I_1}(w) = afcafc$. 
    
    \noindent \textbf{Case 2.a} $h_{I_2}(w) = dafdaf$. Hence, $dafdcafc$. Otherwise, $h_{d,c}(w)=0101$. 
    
    \noindent \textbf{Case 2.a.I} $h_{I_3}(w) = bdfbdf$. Thus $h_{I_1\cup I_3}=bdafbdcafc$ and $h_{a,b}(w) = 1010$. 
    
    \noindent \textbf{Case 2.a.II} $h_{I_3}(w) = dbfdbf$. Therefore,  $h_{I_1\cup I_3}=dabfdbcafc$. Otherwise, $h_{b,c}(w) = 0101 $ or $ h_{a,b}(w) = 1010$. 
    
    \noindent \textbf{Case 2.a.II.i} $h_{I_4}(w) = ebfebf$. This implies,  $h_{c,e}(w)=1100$.

    \noindent \textbf{Case 2.a.II.ii} $h_{I_4}(w) = befbef$. So  $h_{d,e}(w)=0101$.
    
    \noindent \textbf{Case 2.a.II.iii} $h_{I_4}(w) = bfebfe$. Hence  $h_{a,e}(w)=0101$.
    
    \noindent \textbf{Case 2.a.III} $h_{I_3}(w) = dfbdfb$. Therefore,   $h_{a,b}(w) = 0101$. 
    
    \noindent \textbf{Case 2.b} $h_{I_2}(w) = adfadf$. Hence, $adfcadfc$ and $h_{d,c}(w)=0101$.
        
    \noindent \textbf{Case 2.c} $h_{I_2}(w) = afdafd$. Then there are two cases for $h_{I_1 \cup I_2}(w)$: $afdcafcd$ and $afcdafdc$. We can use again $h'$. Therefore, we only have to consider $afdcafcd$.

    \noindent \textbf{Case 2.c.I} $h_{I_3}(w) = bfdbfd$. Hence,  $h_{I_1\cup I_3}(w) = abfdbcafcd$. Otherwise, $h_{b,c}(w) = 0101 $ or $ h_{a,b}(w) = 1010$. 

    \noindent \textbf{Case 2.c.I.i} $h_{I_4}(w) = ebfebf$. Thus,  $h_{c,e}(w) = 1100$.

    \noindent \textbf{Case 2.c.I.ii} $h_{I_4}(w) = befbef$. Thus,  $h_{d,e}(w) = 1010$.

    \noindent \textbf{Case 2.c.I.iii} $h_{I_4}(w) = bfebfe$. Thus,  $h_{a,e}(w) = 0101$.

    \noindent \textbf{Case 2.c.II} $h_{I_3}(w) = fbdfbd$. Then $h_{a,b}(w)=0101$.
    
    \noindent \textbf{Case 2.c.III} $h_{I_3}(w) = fdbfdb$. Then $h_{a,b}(w)=0101$ or $h_{b,c}(w)=1010$.

    Since we discussed all of the cases, there is no 2-uniform word $w$ with $G=G(L,w)$.
\end{example} 

\begin{theorem} \label{thm:interval-model}
Let $L$ be a $(1,2)$-uniform language such that $L \subsetneq (0 \shuffle 1^2) \cup (1 \shuffle 0^2)$ or let $L$ be a $2$-uniform language such that $0110 \notin L$ or $0011 \notin L$. Then, $G \in \mathcal{G}_L$ \iffl an interval model of $G$ with respect to $L$ exists.
\end{theorem}
\begin{proof}
We are going to show both directions of the logical equivalence.
\begin{itemize}
\item Let $G =(V,E) \in \mathcal{G}_L$. According to \Cref{lem:uniform_word}, the following holds. If $L$ is $(1,2)$-uniform, then a word $w \in V^\ast$ exists such that $G = G(L,w)$ and $|w|_v \in \{ 1, 2 \}$ for every $v \in V$. If $L$ is $2$-uniform, then actually a word $w \in V^\ast$ exists such that $G = G(L,w)$ and $|w|_v = 2$ for every $v \in V$. For every $v \in V$, let $v_\alpha$ be the index of the first and $v_\beta$ be the index of the last occurrence of $v$ in $w$ and define $I_v$ as $[v_\alpha, v_\beta]$. Note that $I_v$ has the form $[v_\alpha, v_\alpha]$ if $v$ only occurs once in $w$. Let $\mathcal{I} = \{ I_v \}_{v \in V} $. For every $v \in V$ and $u \in V \setminus \{ u \}$, $(I_v,I_u)$ matches the pattern of $h_{v,u}(w)$ by construction. Hence $\mathcal{I}$ is an interval model of $G$ with respect to $L$. 
\item Let $G =(V,E)$ and $\mathcal{I} = \{ I_v \}_{v \in V}$ be an interval model of $G$ with respect to $L$. Let $i_1 < i_2 < \dots < i_m$ be  the endpoints of the intervals contained in $\mathcal{I}$, such that for an interval $I$ of the form $[i_j, i_j]$ the endpoint $i_j$ only occurs once in the list. For each $j \in [m]$ and each $v \in V$, choose $w_j = v$ if $i_j$ is an endpoint of the interval $I_v$. Define $w = w_1 \cdots w_m$. For every $v \in V$ and $u \in V \setminus \{ u \}$, $(I_v,I_u)$ matches the pattern of $h_{v,u}(w)$ by construction.
Hence $G(L,w) = G$ and $G \in \mathcal{G}_L$.\qed
\end{itemize}\renewcommand{\qed}{}
\end{proof}

\begin{remark} For intervals $[x_\alpha, x_\beta]$ such that $x_\alpha < x_\beta$, the intervals do not have to be closed because of our model assumption that endpoints of different intervals are distinct. 
\end{remark}

\subsection{$(1,2)$-uniform languages}

In the following, we consider all $(1,2)$-uniform  languages~$L\subseteq\{0,1\}^*$ where each letter occurs either once or twice. This leads us to consider the following cases:
$L_1=\langle 001\rangle$,
$L_2=\langle 010\rangle$,
$L_3=\langle 011\rangle$,
$L_4=\langle 001,010\rangle$,
$L_5=\langle 001,011\rangle$,
$L_6=\langle 010,011\rangle$, and
$L_7=\langle 001,010,011\rangle=(00\shuffle 1)\cup (0\shuffle 11)$.
By \Cref{cor:L-symmetry}, $\cG_{L_1}=\cG_{L_3}$ as 
$L_1^R=\{001,110\}^R=\{100,011\}=L_3$. $\cG_{L_4}=\cG_{L_6}$. This leaves us with five languages to consider in this subsection. All graph classes will contain only bipartite graphs due to \Cref{prop:bipartite}, but we will make this bipartition explicit sometimes in the following. Disregarding one trivial case, we will see bipartite chain graphs, convex graphs and their respective complements characterized by these five languages. This is also summarized in \Cref{tab:graphClasses}.

\Cref{thm:interval-model} implies the following; also confer \Cref{fig:interval-model-convex}: 
\begin{corollary} \label{thm:convex} 
    A graph $G=(V,E)$ is 
    $\langle 010\rangle$-representable \iffl $G$ is a convex graph.
\end{corollary} 

\begin{theorem} \label{thm:bipartite-chain}
$\mathcal{G}_{\langle 001 \rangle}$ is the class of bipartite chain graphs.
\end{theorem}
\begin{proof}
To simplify the notation define $L \coloneqq  \langle 001\rangle$.
Let $G=(V,E)$ be $\langle 001\rangle$-representable. Therefore, there exists a word $w$ (of length $\ell$) such that $G=G(L,w)$. By \Cref{lem:uniform_word}, $w$ contains only letters  which appear exactly once or twice. Set $A$ to the set of vertices which appear exactly once in $w$ and set $B$ to the set the vertices which appears exactly twice. Clearly, $G$ is a bipartite graph with the classes $A,B$.

For $v\in V$, define $\max_{v}\coloneqq \max\{i\in [\ell]\mid w_i=v\}$. For all $v,u\in A$, define the ordering relation $v \leq_A u$ \iffl $\max_{v}\leq \max_u$. Let $v,u\in A$ with $v \leq_A u$ and $x\in B\cap N(v)$. Hence, $h_{x,v}(w)=001$. As $\max_{v} \leq \max_{u}$, $h_{x,u}(w)=001$. Thus, $v \leq_A u$ implies $N(v)\subseteq N(u)$.

Let $G$ a bipartite chain graph with the classes $A,B$ and the ordering $\leq_A$ on~$A$. Let $B\coloneqq  \{b_1,\dots,b_q\}$ and $A=\{a_1,\dots, a_p\}$ with $a_i \leq_A a_j$ \iffl $i \leq j$ for all $i,j\in [p]$. Let $\leq_B$ be the corresponding linear ordering on $B$. Without loss of generality, let $b_i \leq_B b_j$ \iffl $i \leq j$ for all $i,j\in [q]$.
Furthermore, let $r_i\coloneqq  \max\{r \in [q]\mid b_r\in N(a_i)\}$.

Define $w_{(0)}\coloneqq b_1 \cdots b_q$ and $w_{(i)}$ for $i\in [p]\cup \{ 0 \}$ inductively.
If  $N(a_{i-1})\subsetneq N(a_i)$, set $w_{(i)}\coloneqq  w_{(i-1)}b_{r_{i-1}+1}\cdots b_{r_{i}}a_i$. Otherwise, $w_{(i)}\coloneqq w_{(i-1)}a_i$. 

Define $V_i\coloneqq B \cup \{a_1,\dots,a_i\}$, $G_i\coloneqq G[V_i]=(V_i,E_i)$ and $G'_i\coloneqq G(L,w_{(i)})$ for $i \in [p] \cup \{ 0 \}$. We prove by induction over $i \in [p]\cup \{ 0 \}$ that $G'_i=G_i$. 

For $i=0$, this holds trivially.    
Now assume $G_i=G_i'$, for some $i \in [p]\cup \{ 0 \}$. If $r_i=r_{i+1}$, then $N(a_i)= N(a_{i+1})$ are the only edges which will be added to $G_i$ and $G_i'$ to get $G_{i+1}$ and $G'_{i+1}$. Therefore, this would imply $G_{i+1}=G'_{i+1}$.
Now assume $r_i\neq r_{i+1}$. Let $j,j_1 \in [r_i]$ and $k,k_1 \in [q]\setminus [r_i]$. As above, $a_1,b_k,b_{k_1}$ as well as $b_j,b_{j-1}$ are not connected by one edge. By construction, $h_{b_j,a_i}(w_{(1)}) = 001 \in L$ for $j\in [r_i]$.  Furthermore, $h_{b_j,b_k}=010\notin L$. Thus, $G_i=G'_i$. 
Hence, $G(w_{(p)},L)= G_p=G$.
\end{proof}
By \cref{lem:representability-reversal}, we get the next result.
\begin{corollary}\label{cor:bipartite-chain2}
$\mathcal{G}_{\langle 011 \rangle}$ is the class of bipartite chain graphs.
\end{corollary}
With \Cref{prop:compl}, we get the following characterization of 112-representable graphs\longversion{ as} defined in~\cite{GaeJi2020}. \longversion{Recall that t}\shortversion{T}hese are different from the 112-representable labelled graphs \longversion{as defined in}\shortversion{of}~\cite{Kit2017a}.
\begin{corollary} \label{cor:bipartite-chain-complement}
The class $\cG_{\overline{\langle 001\rangle}}$ of 112-representable graphs  is the class of complements of bipartite chain graphs. 
\end{corollary}
This corollary is also interesting, as it shows, together with \cite[Thm. 3.13]{GaeJi2020}, that all word-representable graphs (in the classical sense) are complements of bipartite chain graphs. This implies, among other things, that 3-colorable, comparability and circle graphs are co-bipartite chain graphs, see \cite{KitPya2018}. We are not aware of any papers that make such statements explicitly.

Let $G=(V,E)$ be a bipartite graph with the partition classes $A,B \subseteq V$. Then the \emph{bipartite complement graph} of $G$ with respect to the partition $V = A \cup B$ is $(V, E')$ with $E' = \{ \{ a, b \} \mid a \in A, b \in B \} \setminus E$. For a graph class $\cG$ of bipartite graphs, $\bico\cG$  is the class of bipartite complement graphs of graphs $G \in \cG$ with respect to any partition of $V(G)$ into two disjoint, independent sets. 

\begin{theorem} \label{thm:bico-complement}
For every\longversion{ language}~$L$ such that $\emptyset \neq L \subsetneq 0 \shuffle 1^2 \cup 1 \shuffle 0^2 $, $\cG_{\overline{L} \cap (0 \shuffle 1^2 \cup 1 \shuffle 0^2)} = \bico\cG_L$. 
\end{theorem}

\begin{proof}
\begin{enumerate}
\item Let $G \in \cG_{\overline{L} \cap (0 \shuffle 1^2 \cup 1 \shuffle 0^2)}$. According to \Cref{lem:uniform_word}, there exists a word $w \in V(G)$ such that $G = G(\overline{L} \cap (0 \shuffle 1^2 \cup 1 \shuffle 0^2),w)$ and $|w|_v \in \{ 1, 2 \}$ for every $v \in V(G)$. Clearly, $\overline{L} \cap (0 \shuffle 1^2 \cup 1 \shuffle 0^2) = (0 \shuffle 1^2 \cup 1 \shuffle 0^2) \setminus L$. Hence according to \Cref{lem:edge_sets}, $E(G) = E(G(0 \shuffle 1^2 \cup 1 \shuffle 0^2, w)) \setminus E(G(L, w))$. Let $V_1 = \{ v \in V(G) \mid |w|_v = 1 \}$ and $V_2 = \{ v \in V(G) \mid |w|_v = 2 \}$. Then $V_1$ and $V_2$ are two disjoint sets that partition $V(G)$ and both sets are independent in $G(L, w)$. Clearly, $E(G(0 \shuffle 1^2 \cup 1 \shuffle 0^2, w)) = \{ \{ a, b \} \mid a \in V_1, b \in V_2 \}$. Hence, $G(L, w))$ is the bipartite complement graph of $G$ with respect to the partition of $V(G)$ into $V_1$ and $V_2$. 
\item For the other direction, consider a graph $G \in \bico\cG_L$. Hence, $G$ is the bipartite complement of a graph $G' \in \cG_L$. According to \Cref{lem:uniform_word}, there exists a word $w \in V(G')^\ast$ such that $G' = G(L,w)$ and $|w|_v \in \{ 1, 2 \}$ for every $v \in V(G')$. Let $V_1 = \{ v \in V(G') \mid |w|_v = 1 \}$ and $V_2 = \{ v \in V(G') \mid |w|_v = 2 \}$. 
Let $A, B$ be two independent sets such hat $A \cup B$ is a partition of $V(G')$. Let $V(G') = K_1 \cup \dots \cup K_k$ be the partition of $V(G')$ into the connected components of $G'$. Assume there exists an $i \in [k]$ such that $K_i \cap V_1 \cap A \neq \emptyset$ and $K_i \cap V_2 \cap A \neq \emptyset$. Choose $v_1 \in K_i \cap V_1 \cap A$ and $v_2 \in K_i \cap V_2 \cap A$. Because $v_1$ and $v_2$ are in the same connected component, there must be a path $(p_0 = v_1, p_1, \dots, p_m = v_2)$. Because of $L \subseteq 0 \shuffle 1^2 \cup 1 \shuffle 0^2$, $E(G') \subseteq \{ \{ a, b \} \mid a \in V_1, b \in V_2 \}$. Hence for every $1 \leq j \leq m$, $j$ is even \iffl $p_j \in V_1$. Because of the choice of $A$ and $B$ for every $1 \leq j \leq m$, it also holds that $j$ is even \iffl $p_j \in A$. But this means that $m$ must be odd (because $p_m \in V_2$) and $j$ must be even (because $p_m \in A$). This is a contradiction. Analogously, we show that no $i \in [k]$ exists such that $K_i \cap V_1 \cap B \neq \emptyset$ and $K_i \cap V_2 \cap B \neq \emptyset$. Therefore, we can conclude the following: 
For each partition that could have been used to form the bipartite complement of the graph $G' \in \cG_L$, $E(G) = \{ \{ a, b \} \mid a \in V_1, b \in V_2 \} \setminus E(G')$. 

We use \Cref{lem:edge_sets} again to prove $\{ \{ a, b \} \mid a \in V_1, b \in V_2 \} \setminus E(G') = E(G(\overline{L} \cap (0 \shuffle 1^2 \cup 1 \shuffle 0^2), w))$. Hence, $G = G(\overline{L} \cap (0 \shuffle 1^2 \cup 1 \shuffle 0^2),w)$. 
\qed
\end{enumerate}\renewcommand{\qed}{}
\end{proof}

\begin{corollary} \label{cor:bica-classes}
\begin{enumerate}
\item $\cG_{\langle 001, 011\rangle}$ is the class of bico-convex graphs.
\item $\cG_{\langle 010,011 \rangle}$ and $\cG_{\langle 010,001 \rangle}$ are equal to the class of bipartite chain graphs.
\end{enumerate}
\end{corollary}
\begin{proof}
\begin{enumerate}
\item The statement follows from \Cref{thm:convex} and \Cref{thm:bico-complement}.
\item The statement follows from \Cref{lem:representability-reversal}, \Cref{thm:bipartite-chain}, \Cref{thm:bico-complement} and the fact that the class of bipartite chain graphs is closed under bipartite complements, which can be obtained by redefining a linear ordering $\leq_{A'}$ on $A$ such that $a_1 \leq_{A'} a_2 \Leftrightarrow a_1 \geq_A a_2$, for each pair $a_1,a_2\in A$. 
Alternatively, this can be also understood by the characterization of bipartite chain graphs as bisplit graphs.
\qed
\end{enumerate}\renewcommand{\qed}{}
\end{proof}

\begin{remark}
With \Cref{thm:interval-model} we can define interval models for the following classes:
\begin{enumerate}
\item bipartite chain graph (interval models with respect to the languages $\langle 001 \rangle$, $\langle 011 \rangle$, $\langle 010,011 \rangle$ and $\langle 010,001 \rangle$),
\item bico-convex graphs (interval model with respect to the language $\langle 001, 011\rangle$).
\end{enumerate}
We are not aware of any previously published geometric intersection models for these graph classes.
\end{remark}

\begin{remark}
    Different from bipartite chain graphs, convex graphs are not closed under bipartite complement. To see this, we consider a concrete example in the following. Let $G=(A\cup B, E)$ with $A:=[6]$, $B:=\{a,b,c,d,e\}$ and 
    $$E:=\{\,\{1,a\},\{2,a\},\{2,b\},\{3,b\},\{3,c\},\{4,c\},\{4,d\},\{5,d\},\{5,e\},\{6,e\}\,\}.$$
    It is easy to see that $G=G(L,w)$ with $w=a1b2ac3bd4ce5d6e$. Hence, $G$ is convex. Let $H=(A \cup B, E')$ be the bipartite complement of $G$. We want to show that $H$ is no convex graph. Clearly, $A$ and $B$ are the partition classes of $H$ and there are no other ways to partition these vertices. By \Cref{lem:uniform_word}, there exists a $w' \in (A\cup B)^*$ with $H=G(L,w')$ and $\vert w'\vert_{v}\in \{1,2\} $ for each $v \in A \cup B$. Since $A,B$ is a unique partition for $H$, for all $v\in A$ and $u \in B$, $\vert w'\vert_v \neq \vert v \vert_u$.
    
    We first assume  $\vert w'\vert_v=1$ for each $v \in A$. We start by considering $h_{a,1}(w')=001= h_{a,1}(w')$. By \Cref{cor:L-symmetry}, $h_{a,1}(w')=100= h_{a,1}(w')$ works analogously. Since $N_H(a)=\{3,4,5,6\}$, for each $v \in \{ 3,4,5,6\}$ and $j \in \{1,2\}$, $h_{v,j}(w')=01$. $3\notin \{1,4,5,6\} =N_H(b)$ implies $h_{\{a,b,1,2,3,5\}}(w')=a3b5a1b2$. By $3,5,2\in N_H(e)$, $h_{\{b,e,1,2,3,5\}}(w')=e3b5a1be$ holds which contradicts $\{5,e\}\in E'$.

    So consider $h_{\{a,1,2\}}(w')=1aa2$ (note that $h_{\{a,1,2\}}(w')=2aa1$ works analogously). Since $\{a,5\}\in E'$ and $1,2\in N_H(e)$, $h_{\{a,e,1,2,5\}}(w')=e2a5a1e$ which contradicts $\{e,5\}\notin E'$. 

    Now we assume that for each $v \in A$, $\vert w'\vert_v=2$. Consider $h_{\{a,b,2\}}(w')=a22b$  ($h_{\{a,b,2\}}(w')=b22a$ is analogously). By $a,b\in N_H(6)$, $h_{\{a,b,2,6\}}(w')=6a22b6$. This contradicts $\{e,6\}\notin E'$, as $\{2,e\}\in E'$ and so $h_{\{a,b,e,2,6\}}(w')=6a2e2b6$. 

    Therefore, assume $h_{2,a}(w') = 001 = h_{2,b}(w')$ ($h_{2,a}(w')=100= h_{2,b}(w')$ is analogously). $\{b,3\},\{c,3\} \notin E'$ and $\{a,3\},\{e,3\},\{e,2\}\in E'$ imply $h_{\{a,b,c,e,2,3\}}(w') = 2c3e2a3b$. Since $a,b\in N_H(6)$, $h_{\{b,c,e,6\}}(w') = 6ceb6$ which contradicts $\{6,e\}\notin E'$. 

    As there is no case remaining, $G$ is convex but its bi-complement is not.
\end{remark}

\subsection{2-uniform languages}

We will now study all 2-uniform languages. Our study covers again all (seven) cases, as can be easily verified by looking at \Cref{tab:graphClasses}. 

The following Lemma shows (in view of \Cref{exa:represented-graphs}) that there are many ways to describe the same graph class, and the reader should find it easy to find many more representations of this class.

\begin{lemma}\label{lem:00shuffle11}
For $L = 0^2\shuffle 1^2$, 
$\mathcal{G}_{L}=\{K_n\cup N_m\mid n,m\in\N\}$.
\end{lemma}
\begin{proof}
\begin{enumerate}
    \item Let $G = (V,E) \in \cG_L$. Hence, there exists a word $w \in V^*$ such that $G = G(L,w)$. Let $V' = \{ v \in V \mid |w|_v = 2 \}$ and $E' = \{ \{u,v\} \mid u,v \in V' \}$.
\begin{align*}
\{ u, v \} \in E  \iff h_{u,v}(w) \in L 
                  \iff |w|_u = |w|_v = 2  \iff u, v \in V'  \iff \{ u, v \} \in E'\,.
\end{align*}
This shows $E = E'$.
\item Let $G = (V,E)$ be such that there exists a subset $V' \subseteq V$ with $E = \{ \{u,v\} \mid u,v \in V' \} $. Let $V' = \{ v_1, \dots, v_m \}$ and $V \setminus V' = \{ v_{m+1}, \dots, v_n \}$. 
Choose $w = v_1 \dots v_m \cdot v_1 \dots v_m \cdot v_{m+1} \dots v_n $. Let $G_2 = G(L,w)$. We prove that $G_2 = G$. Obviously, $V(G_2) = V$. 
\begin{align*}
\{ u, v \} \in E(G_2)   \iff h_{u,v}(w) \in L \iff |w|_u = |w|_v = 2 \iff u, v \in V' 
                     \iff \{ u, v \} \in E\,.
\end{align*}
This shows that $G(L,w) = G_2 = G$.
\end{enumerate}
Hence, $\mathcal{G}_{L}=\{K_n\cup N_m\mid n,m\in\N\}$.
\end{proof}

We next provide a collection of useful closure properties that allow us to apply the Trash \Cref{trash-theorem} in a number of circumstances. Recall the definition of $T_L$ from that theorem.


\begin{proposition}\label{prop:2-uniform-plus-vertices}
Let $L\subseteq\{0,1\}^+$ be 2-uniform. Assume that for each $G=(V,E)\in \cG_L$, there exists a 2-uniform word~$w\in V^*$ such that $G=G(L,w)$.
\begin{enumerate}
    \item If $L\cap \{0110,0011\}\neq\emptyset$, then 
$\cG_L$ is closed under adding universal vertices.  \item If ${L}\cap \{0110,0011\}\neq\emptyset$, then  $\cG_{\overline{L}}$  is closed under adding isolated vertices.
\end{enumerate}
Now assume that for each $G=(V,E)\in \cG_{\overline{L}\cap (0^2 \shuffle 1^2)}$, there exists a 2-uniform word~$w\in V^*$ such that $G=G(\overline{L}\cap (0^2 \shuffle 1^2),w)$.
\begin{enumerate} \addtocounter{enumi}{2}
    \item If $\overline{L}\cap \{0110,0011\}\neq\emptyset$, then  $\cG_{\overline{L}\cap (0^2\shuffle 1^2)}$ is closed  under adding universal vertices.
    \item If $\overline{L}\cap \{0110,0011\}\neq\emptyset$, then 
    $\cG_{{L}\cup T_L}$ is closed under adding isolated vertices.
\end{enumerate} 
\end{proposition}
\begin{proof}
Let $L\subseteq\{0,1\}^+$ be 2-uniform.
\begin{enumerate}
    \item Let $G=(V,E)\in \mathcal{G}_{L}$. Hence, there is a $2$-uniform $w\in V^*$ with $G=G(L,w)$. Let $\ta\notin V$. First, assume $0110\in L$.  Consider $w'=\ta w\ta$. Then, $\ta$ is a universal vertex in $G'=G(L,w')$ and $G=G'[V]$, because $w$ is $2$-uniform. Otherwise, by our assumption $0011\in L$.   Consider $w''=w\cdot \ta^2$. Then, $\ta$ is a universal vertex in $G''=G(L,w'')$ and $G=G''[V]$.
    \item By the previous item,  $\cG_L$ is closed under adding universal vertices. Hence, its corresponding class of complement graphs ${\cG_L'}$ is closed under adding isolated vertices. By \Cref{prop:compl}, $\cG_{\overline{L}}={\cG_L'}$. 
    \item As $\widehat L\coloneqq {\overline{L}\cap (0^2\shuffle 1^2)}$ is 2-uniform, the claim follows with the first item, where $\widehat L$ takes on the role of~$L$, as $\overline{L}\cap \{0110,0011\}=\widehat L\cap \{0110,0011\}$.
    \item As $T_L=\overline{0^2\shuffle 1^2}$, this follows by the third item by the reasoning of the second item (using \Cref{prop:compl}) with De Morgan's Law as $\overline{{\overline{L}\cap (0^2\shuffle 1^2)}} = L \cup T_L$. 
\qed 
\end{enumerate}
\renewcommand{\qed}{}
\end{proof}

The discussions in \Cref{rem:adding-isolates} and in the paragraph preceding \Cref{lem:uniform_word} are also related in particular to the last statement of the previous proposition.

After this type of prelude, well-known non-trivial graph classes are characterized in the following. We start with the class of interval graphs and their complements.

\begin{theorem}\label{thm:intervalgraphs}
For $L = \langle 0101,0110\rangle$, 
$\mathcal{G}_{L}$ is the class of interval graphs. 
\end{theorem}

This result follows from the characterization of interval graphs as 1-11-representable graphs that can be represented by uniform words containing two copies of each letter, see~\cite{CheKKKP2018}. However, this is also a direct consequence of our interval representation model with respect to $\langle 0101,0110\rangle$ that also supports the intuition that the pattern $0101$ covers the case of proper interval overlap, while $0110$ covers the case of containment. From a graph-theoretic point of view, looking at their intersection graph models, a natural generalization of interval graphs are chordal graphs. \Cref{thm:intervalgraphs} shows (in particular) that interval graphs can be represented by a finite language. However, \Cref{cor:non-finitely-presentable graph classes} shows that chordal graphs cannot be represented by any finite language. The related information-theoretic argument also proves that there must be chordal graphs which are not interval graphs. This fact is of course well-known in the graph-theory community, but this argument is clearly different from the standard proofs, see \cite{Gol2004i,Gol2004l}.

\begin{theorem}
$\mathcal{G}_{\overline{ \langle 0011 \rangle }}$ is the class of interval graphs.
\end{theorem}
\begin{proof}
Let $L = \langle 0101, 0110 \rangle$. Using the notation from \Cref{trash-theorem}, we get $\hat L = \overline{ \langle 0011 \rangle }$. According to \Cref{thm:intervalgraphs}, $\mathcal{G}_{L}$ is the class of interval graphs and therefore $\mathcal{G}_{L}$ is closed under adding universal vertices. 
Due to \Cref{lem:uniform_word}, we can apply the second part of \Cref{prop:2-uniform-plus-vertices} to the language $\langle 0011 \rangle$. 
Thus, $\mathcal{G}_{\hat L}$ is closed under adding isolated vertices. Hence, we get the equality $G_L = G_{\hat L}$ from \Cref{trash-theorem}. This proves the statement. 
\end{proof}

\noindent
By \Cref{prop:compl}, we obtain:
\begin{corollary} \label{cor:co-intervalgraphs}
For $L = \langle 0011 \rangle $, $\mathcal{G}_{L}$ is the class of co-interval graphs. 
\end{corollary}

\begin{remark}
According to \Cref{thm:interval-model} the interval model with respect to the language $\langle 0011\rangle$ is a geometric model for the class of co-interval graphs.
\end{remark}

The next rather famous class of graphs that we study are permutation graphs.
According to graphclasses.org, permutation graphs are exactly the containment graphs of intervals, and this is also the interpretation of our intersection model according to the next theorem, confer \Cref{thm:interval-model}. As we could not trace back a proof for this statement in the literature, we are providing a short one next, using our general geometric interpretations of 2-uniform languages and the corresponding graph classes.

\begin{theorem} \label{thm:permGraphs}
For $L = \langle 0110 \rangle$, $\mathcal{G}_L$ is the class of permutation graphs.
\end{theorem}

\begin{proof}
Let $G=(V,E)$ be a permutation graph with $V=[n]$, where $n\in \mathbb{N}_{\geq 1}$, and we have the permutation $\pi:V\to V$ such that $\{i,j\}\in E\iff (i-j)(\pi(i)-\pi(j))<0$. Then, define $w=w_1 \cdots w_{2n}$ with $w_i=i$ for $i\in V$ and $w_i=\pi^{-1}(i-n)$ for $i\in \{n+1,\dots,2n\}$. We want to show $G=G'=(V,E')\coloneqq G(\langle0110\rangle,w)$. To this end, we only have to prove that $E=E'$. Let $\{u,v\}\in E$. By the definition of permutation graphs, $(v-u)(\pi(v)-\pi(u)) < 0$. Assume $v<u$ (for $u<v$, the proof works analogously). Then, $\pi(u)<\pi(v)$ and $h_{v,u}(w)=0110 \in \langle0110\rangle$. Hence, $\{u,v\}\in E'$.
    
    Now assume $\{u,v\}\notin E$. Then, $(v-u)(\pi(v)-\pi(u)) > 0$. We assume $v<u$ which implies $\pi(v)<\pi(u)$ and $h_{v,u}(w)=0101\notin \langle0110\rangle$. Thus, $\{u,v\}\notin E'$.

    For the only-if part, let $G=(V,E)$ be a graph with $w\in V^{\ast}$ such that $G=G(\langle 0110 \rangle,w)$. Without loss of generality, assume $V=[n]$ for $n\in \mathbb{N}$ such that, for all $v,u\in V$, $v < u$ implies that $v$'s first occurrence in $w$is before $u$'s  first occurrence. Define $\pi:V\to V$ such that, for all $v,u\in V$, $\pi(v) < \pi(u)$ implies that $v$'s second occurrence in $w$ is before $u$'s second occurrence. Let $G'=(V,E')$ be the permutation graph implied by the permutation~$\pi$. 

    We want to show that $E=E'$. Consider $\{v,u\}\subseteq V$ with $v<u$. First, assume $\{v,u\}\in E$. As $h_{v,u}(w)=0110$, $\pi(u)<\pi(v)$. This implies that  $(v-u)(\pi(v)-\pi(u)) < 0$. Hence, $\{v,u\}\in E'$.  
    Secondly, assume $\{v,u\}\notin E$.  As $h_{v,u}(w)=0110$, $\pi(v)<\pi(u)$. Hence, $(v-u)(\pi(v)-\pi(u))>0$ and thus $\{v,u\}\notin E'$.
\end{proof}

The following lemma should be known for permutation graphs, but we have made this explicit for two reasons: first, such lemmas (or observations) are hard to find in the literature and secondly, it is an example how to reason within our framework in a very concise way, using \Cref{prop:2-uniform-plus-vertices} and \Cref{lem:uniform_word}.

\begin{lemma}\label{lem:0110-plus-universal-vertices}
$\cG_{\langle 0110 \rangle}$ is closed under adding universal vertices.\qed
\end{lemma}

\begin{lemma} \label{lem:permGraphs2} For $\hat L = \langle 0110 \rangle  \cup  \overline{0^2 \shuffle 1^2}$,
$\cG_{\hat L}$ is the class of permutation graphs.
\end{lemma}

\begin{proof} Let $L=\langle 0110 \rangle$.
Observe that $T_L=\overline{0^2 \shuffle 1^2}$ in the sense of \Cref{trash-theorem} as $L$ is 2-uniform, so that  $\hat L =L\cup T_L$.
Combining \Cref{lem:0110-plus-universal-vertices} with \Cref{trash-theorem}, $ \cG_{\hat L}= \cG_L$ follows.
\end{proof}

The class of permutation graphs are closed under taking graph complements.  Therefore, \Cref{prop:compl} and \Cref{lem:permGraphs2} imply the following:
\begin{corollary}\label{cor:co-permutation}
$\mathcal{G}_{\langle 0011, 0101 \rangle}$ is the class of permutation graphs.
\end{corollary}

We now turn our attention to yet another well-studied graph class, namely circle graphs.
Kitaev and Lozin showed in~\cite[Theorem 5.1.7]{KitLoz2015} the next result that, reformulated in our framework, reads as follows:
\begin{proposition} \label{prop:circlegraphs}
$\mathcal{G}_{\langle 0101\rangle}$ is the class of circle graphs. 
\end{proposition}

We again also get this proposition as a corollary of \Cref{thm:interval-model}, because of the fact that circle graphs are exactly the (interval) overlap graphs, see \cite{Gol2004f}.

\begin{theorem} \label{thm:co-circle}
For $L = \langle 0011,0110 \rangle$, $\cG_L = \{ G \cup N_m \mid G \text{ is a co-circle graph}, m \in \mathbb{N} \} $.
\end{theorem}

\begin{proof} We will show that $\cG_{\overline{L}} = \{ G \nabla K_m \mid G \text{ is a circle graph}, m \in \mathbb{N} \} $. This is sufficient because of \Cref{prop:compl}. 
\begin{enumerate}
\item Let $G \in \cG_{\overline{L}}$. Hence, a word $w \in V(G)^\ast$ exists such that $G = G(\overline{L},w)$. Because of \Cref{lem:edge_sets}, the following holds:
$E(G) = E(G(\langle 0101 \rangle),w) \cup \binom{V(G)}{2} \setminus E(G(0^2 \shuffle 1^2),w)$.
The graph $G(\langle 0101 \rangle),w)$ is a circle graph according to \Cref{prop:circlegraphs}. In the graph $G(0^2 \shuffle 1^2)$, the set $V_2$ is a clique and $V \setminus V_2$ is a set of isolated vertices. Hence, $G \in \{ G \nabla K_m \mid G \text{ is a circle graph}, m \in \mathbb{N} \}$.
\item Let $G' = G \nabla K$ such that $K$ is isomorphic to $K_m$ with $m \in \mathbb{N}$ and $G = (V,E)$ is a circle graph with $n$ nodes. According to \Cref{prop:circlegraphs}, a word $w$ exists such that $G = G(\langle 0101 \rangle,w)$. Let $V(K) = \{ u_1, \dots u_m \}$ and $w' = w \cdot u_1 \cdots u_m$. According to \Cref{thm:graph-decomposition}, $E(G(\overline{L}, w')) = E_{1,1}(\overline{L}, w') \cup E_{1,2}(\overline{L}, w') \cup E_{2,2}(\overline{L}, w')$. Clearly, $E_{1,1}(\overline{L}, w') = \binom{V(K)}{2}$, $E_{1,2}(\overline{L}, w') = \{ \{ v, u \} \mid v \in V, u \in V(K) \}$ and $(V_{2,2}(\overline{L}, w'), E_{2,2}(\overline{L}, w')) = G $. Hence, $G(\overline{L}, w') = G'$ and $G' \in \cG_{\overline{L}}$.\qed
\end{enumerate}\renewcommand{\qed}{}
\end{proof}

Note that we cannot prove \Cref{thm:co-circle} using \Cref{prop:2-uniform-plus-vertices}, because not every graph in $\cG_{\langle 0011,0110 \rangle}$ can be represented by a 2-uniform word (see \Cref{exm:very_long_example}). 

\subsection{Languages with at most two occurrences of $0$ or $1$}
\label{subsec:atmost-two-occurrences}

As we have seen in the preceding subsections, already the (nearly) length-uniform languages themselves offer quite a rich spectrum of graph classes, many of them well-known from other contexts. In this subsection, we only study some of the graph classes that can be obtained from languages with at most two occurrences of 0 or~1. We will encounter two different graph classes this way.
Our first example was already studied in \Cref{exa:split}. Now, we give an exact classification of the class.

\begin{theorem} \label{thm:threshold}
    Let $L=\langle 01,001\rangle$. Then $\cG_L$ is the class of threshold graphs.
\end{theorem}

\begin{proof}
Let $G=(V,E)$ be a threshold graph. This means that $G$ can be obtained from a sequence of operations that add isolated and universal vertices, starting from the graph $G_1=(\{\ta_1\},\emptyset)$, i.e., we have a corresponding sequence of graphs $G_1,\ldots,G_n$, with $n=|V|$, such that $G_{i+1}$ is obtained from $G_i$ by adding the vertex $\ta_{i+1}$, for $i=1,\dots,n-1$. Hence, with $V_i=V(G_i)$, $G_i=G[V_i]$. Also, $V=\{\ta_1,\dots,\ta_n\}$. Now, set $w_1=\ta_1\ta_1$ and construct $w_{i+1}$ from $w_i$ as follows inductively:
\begin{itemize}
    \item If $\ta_{i+1}$ is an isolated  vertex in $G_{i+1}$, then set $w_{i+1}\coloneqq w_i\cdot \ta_{i+1}\ta_{i+1}$.
    \item If $\ta_{i+1}$ is  a universal vertex in $G_{i+1}$, then set $w_{i+1}\coloneqq w_i\cdot \ta_{i+1}$.
\end{itemize}
We claim that $G_i=G(L,w_i)$ for all $1\leq i\leq n$.
The correctness of the construction follows by a simple inductive argument.

Conversely, let $G=(V,E)\in \cG_L$. Consider a word $w\in V^*$ that $L$-represents~$G$, i.e., $G=G(L,w)$. Let $V=\{\ta_1,\dots,\ta_n\}$.
Assume that $G$ contains $k$ isolated vertices. Without loss of generality, let these be $\ta_{n-k+1},\dots,\ta_n$.
Hence, $V'=\{\ta_1,\dots,\ta_{n-k}\}$ collect the vertices of degree at least one. Let $w'=h_{V'}(w)$. Then, clearly $G(L,w'\cdot \ta_{n-k+1}^2\cdots\ta_n^2)$ is isomorphic to~$G$. Now, consider $G'=G(L,w')$; clearly, $G'=G[V']$.
We are now modifying $w'$ towards a string $w''$ such that, for $1\leq i\leq n-k$, $|w'|_{\ta_i}=|w''|_{\ta_i}$ and if $|w''|_{\ta_i}=2$, then $\ta_i^2$ is a factor of $w''$.

We achieve this by executing the following until no further changes are possible, starting with setting $w''\coloneqq w'$:
\begin{quote}
If there is any letter $\ta$ that occurs twice in $w''$ such that there is some decomposition $w''=w_1\ta w_2\ta w_3$ with $w_2\neq\emptyword$, then set $w''\coloneqq w_1w_2\ta^2 w_3$. 
\end{quote}
We claim that the graph that is $L$-represented does not change by this operation, which also proves $G'=G(L,w'')$ for the final string~$w''$ by a simple induction.
Namely, the described operation clearly does not affect any edge that does not contain the vertex~$\ta$. By the structure of $L$ and by the definition of $L$-representability, any letter that occurs in $w_1$ or $w_3$ forms an edge together with $\ta$ with respect to the old $w''$ if and only if it forms an edge with $\ta$ with respect to the new $w''$. Letters in $w_2$ cannot form an edge with $\ta$ neither before nor after performing the operation, as none of them can follow the pattern $001$. 

Now, consider $G(L,w'')$ and observe that this graph is again isomorphic to~$G$. Moreover, $w''$ can now be thought of being formed by subsequently adding $\ta_i$ or $\ta_i^2$ to the empty string. As seen before, this therefore means that $G$ can be build by adding universal or isolated vertices. Hence, $G$ is a threshold graph.
\end{proof}

Recall that $\hat L \coloneqq L \cup T_L$, where  $T_L \coloneqq  \{ w \in \{ 0, 1 \}^* \mid |w|_0 \notin \mathrm{freq}(L) \vee |w|_1 \notin \mathrm{freq}(L) \} $ for any binary language~$L$, see \Cref{trash-theorem}.

\begin{corollary}
For $L = \langle 01, 001 \rangle$, $\cG_{\hat L}$ and $\cG_{\langle 010, 011 \rangle \cup 0^2 \shuffle 1^2}$ are the threshold graphs. 
\end{corollary}
\begin{proof}
For $L = \langle 01, 001 \rangle$, $T_L=\{ w \in \{ 0, 1 \}^* \mid |w|_0 \notin \{1,2\} \vee |w|_1 \notin \{1,2\} \}$.
Let $G = (V,E)\in \cG_{\hat L}$ and $w \in V^*$ such that $G = G(\hat L, w)$. If $v \in V$ such that $|w|_v > 2$, define $w' = v \cdot h_{V \setminus \{ v \}}(w)$. Clearly $G = G(L,w')$. By inductive reasoning we can show that a word $x \in V^*$ exists such that $|x|_v \in \{ 1, 2\}$ for every $v \in V$ and $G=G(L,x)$. Let $u$ be a vertex such that $u \notin V$. $G(\hat L,x \cdot uu) = G \cup (\{ u \}, \emptyset)$. Hence $\cG_{\hat L}$ is closed under adding isolated vertices. By \Cref{trash-theorem}, we get $\cG_{\hat L} = \cG_{L}$. 
Clearly, ${\langle 010, 011 \rangle \cup 0^2 \shuffle 1^2} ={\overline{\hat L}}$. By \Cref{prop:compl}, we get that $\cG_{\langle 010, 011 \rangle \cup 0^2 \shuffle 1^2}$ is the class of threshold graphs, because the threshold graphs are closed under graph complement and because $\cG_{L}$ are the threshold graphs as we know from \Cref{thm:threshold}. 
\end{proof}

In \cite{KlaPet87}, Klavžar and Petkovšek discussed intersection graphs of halflines, i.e., \longversion{one-side bounded}\shortversion{unbounded} intervals. These are a special class of cobipartite graphs. The following theorem might look a bit awkward in its formulation, but it gives a nice characterization ``up to isolates'', and similar exceptions can also be found in \Cref{tab:graphClasses}. Again, this is related to the fact that the class of halfline intersection graphs is not closed under adding isolates, see \Cref{propos:closure-adding-isolates}.

\begin{theorem} \label{thm:intersection-graph-halflines} Let $L_{\text{HL}}\coloneqq\langle 01,011,0101,0011,0110\rangle$.
    A graph $G=(V,E)$ is $L_{\text{HL}}$-representable \iffl $G$ is the disjoint union of $G_1$ and $G_2$ where $G_1$ is an intersection graph of halflines and $G_2$ is a null graph.
\end{theorem}

\begin{proof}
%
%
Let $G=(V,E)$ be an intersection graph of halflines of order $n$ with the interval representation $\mathcal{I}=\{I_v\mid v\in V\}$. Furthermore, $V_1\subseteq V$ denotes the set such that for each $v\in V_1$ there exists an $a_v \in \mathbb{R}$ with $I_v = [a_v , \infty)$, i.e., $I_v$ is a left-bounded halfline. $V_2\subseteq V$ is the set of vertices $v\in V$  such that there is an $a_v \in \mathbb{R}$ with $I_v = ( -\infty, a_v ]$, i.e., $I_v$ is a right-bounded halfline. Thus $V_1,V_2$ are cliques in~$G$. Define the an order  on~$V$ by $$\mbox{$\preceq'$}\coloneqq \{ (v,u)\in V^2 \mid (a_v<a_u)\vee (a_v=a_u \wedge v\in V_1 \wedge u\in V_2) \}\,.$$ Extend $\preceq'$ to a linear ordering that we write as~$\preceq$. Let $w'$ be the word that enumerates~$V$ in the ordering~$\preceq$, i.e., $|w'|=n$. Let $w$ denote the word $w'w''$ where $w''$ is an enumeration of~$V_2$. Define $G'=(V,E')=G(L_{\text{HL}},w)$. Clearly, for each $i\in \{1,2\}$ and $v\in V_i$, $v$ appears~$i$ times in~$w$. As $L_{\text{HL}}$ contains all 1-uniform and 2-uniform words,  $V_1$ and $V_2$ are cliques in~$G'$. This leaves to show that for all  $v\in V_1, u\in V_2$, $\{v,u\}\in E$ \iffl $h_{v,u}(w)=011$. By definition of~$w$, the last position of $h_{v,u}(w)$ is~1. Thus, $\{v,u\}\in E$ \iffl $a_v \leq a_u$ \iffl $h_{v,u}(w)=011$. Thus, $G=G'$. As in \Cref{propos:closure-adding-isolates}, we can easily add isolated vertices to~$G'$.

Now let $G=(V,E)\in \mathcal{G}_{L_{\text{HL}}}$. Then there exists a word $w=w_1\cdots w_{\ell}\in V^*$ with $w_i\in V$ for $i\in [\ell]$ such that $G=G(L_{\text{HL}},w)$. Let $V_j:=\{v\in V \mid \vert w\vert_v = j\}$ for $j \in \{1,2\}$. Define 
    $$a_v:=\begin{cases}
        \min \{i\mid w_i=v\}, & v\in V_1 \cup V_2,\\
        0, & \text{otherwise,}
    \end{cases} \text{ and }I_v:=\begin{cases}
        (-\infty, a_v], & v\in V_2,\\
        [a_v, \infty), & v\in V_1\,.
    \end{cases}$$
    Furthermore, $G'=(V,E')$ denotes the intersection graph of $\mathcal{I}=\{I_v\mid v\in V\}$.  Now we want to show $E=E'$. Let $v,u\in V$. If $v \notin V_1 \cup V_2$, then $v$ is isolated in $G$. 
    Thus, we can assume $v,u \in V_1 \cup V_2$. If $v,u\in V_1$ (respectively $v,u\in V_2$) then $\{v,u\}\in E$ and $\ell \in I_v\cap I_u$ (respectively $0\in I_v \cap I_u$). 
    So, we can assume $v\in V_1$ and $u\in V_2$. $\{v,u\}\in E$ if an only if $a_v\leq a_u$ \iffl $ I_v \cap I_u \neq \emptyset$.
\end{proof} 

We can characterize the class of halfline intersection graphs by the infinite language $L:= \langle011\rangle \cup \bigcup_{(i,j)\in \mathbb{N}_{\geq 1}^2\setminus \{(1,2),(2,1)\}}0^i \shuffle 1^j$. The proof is nearly identical to the previous one, apart from the discussion of isolated vertices. If a vertex (letter) has frequentness different from $1$ or~$2$, then this corresponds to the interval $[0,\infty)$.

Our last theorem of this subsection does not really fit here, as we also see vertices of frequentness~3. However, it shows that one can still expect new graph classes to show up when increasing frequentnesses further. At least, we saw no way to describe this graph class with vertices of frequentness at most~2. Moreover, this theorem is needed later to disprove a conjecture that one might tend to believe.

\begin{theorem} \label{thm:intervalBigraphs}
    $\mathcal{G}_{\langle 01110,01101,01011,01100,01010,01001\rangle}$ is the set of interval bigraphs.
\end{theorem}

\begin{proof}
    To simplify the notation, define $L := \langle 01110,01101, 01011,01100,01010,01001\rangle$. Note that the words $w$ with $\{ |w|_0, |w|_1 \} = \{ 2, 3 \}$ missing in $L$ are exactly the words that start with two identical letters. Let $G = (V,E)\in \mathcal{G}_L$. Then there exists a $w\in V^*$ with $G=G(L,w)$. Assume there exists a $v\in V$ with $\vert w\vert_v \notin \{2,3\}$. Then define $w':=vv\cdot h_{V\setminus\{v\}}(w)$. 
    Note that $v$ is an isolated vertex in~$G$  and $G(L,w')$. Since $h_{x,y}(w)=h_{x,y}(w')$ for all $x,y\in V\setminus \{v\}$, $G=G(L,w')$. Thus, we can assume that each $v\in V$ appears exactly twice or thrice in~$w$.

    Clearly, $G$ is bipartite with the classes $V_2$ and $V_3$. Let $w=w_1\cdots w_{\ell}$ with $w_1,\dots, w_{\ell}\in V$. We denote by $i_{v,j}\in[\ell]$ for $v\in V$ and $j\in \{1,2\}$ the number such that $w_{i_{v,j}}=v$ and $\vert w_1 \cdots w_{i_{v,j}}\vert_v = j$, i.e., $i_{v,j}$ is the index of the $j^{\text{th}}$ occurrence of~$v$ in~$w$.

    For $v\in V$, define $I_v=[i_{v,1},i_{v,2}]$. Let $\ta\in V_2$ and $\tb \in V_3$. Consider $\{\ta,\tb\}\notin E$. This holds \iffl $h_{\ta,\tb}(w) \in \{00111,11001,11010,11100\}$. This are exactly the cases where $I_\ta \cap I_\tb = \emptyset$. Therefore, $G$ is an interval bigraph. 

    Now let $G=(V,E)$ be an interval bigraph of order~$n$ with the classes $A,B \subseteq V$ and the family of intervals $\mathcal{I}=\{I_v:=[l_v,r_v]\}_{v \in V}$. In the proof of Lemma 8 of \cite{DasSah2024} it is proven that we can assume the endpoints of the intervals are different. Define $x_1,\dots,x_{2n}\in \{l_v,r_v\mid v\in V\}$ such that for $i,j\in [2n]$, $x_i<x_j$ \iffl $i<j$. Furthermore define $w':=w_1\cdots w_{2n} \in V^{2n}$ with $w_i=v$ if $x_i\in \{l_v,r_v\}$ for $v\in V$. By adding each vertex from~$A$ once in an arbitrary order at the end of~$w'$, we obtain $w$. Clearly, each vertex from~$A$ appears three times and each from~$B$ two times in~$w$. Define $G'=(V,E')=G(V,w)$. Let $v\in A$ and $u\in B$. Consider $\{v,u\}\in E$. Hence $\max\{l_v,l_u\}\leq \min\{r_v,r_u\}$. This implies $h_{v,u}(w)  \in \{01010, 01100, 10010, 10100\}\subseteq L$ and $\{v,u\}\in E'$.

    Assume $\{v,u\}\in E'$. Since at the end of~$w$, $A$ is enumerated, there the last symbol of $h_{v,u}(w)$ is a~$0$. Hence, $h_{v,u}(w) \in \{01010,01100,10010,10100\}$. In this case  $\max\{l_v,l_u\}\leq \min\{r_v,r_u\}$ and $\{u,v\}\in E$. Thus, $G=G'$.
\end{proof} 

If we use the first and the last occurrence of a letter to describe the bounds of an interval (instead of the first two occurrences), we get the following corollary: 

\begin{corollary}
   $\mathcal{G}_{\langle (0^3 \shuffle 1^2) \setminus \{ 00011, 11000\} \rangle}$ is the set of interval bigraphs. 
\end{corollary}

\subsection{Some simple infinite languages and their graph classes}

In this subsection, we are going to study the very simple regular language $\langle 0^*1^*\rangle$ and variations thereof. We will see that they nicely link to some previously studied languages (and hence graph classes).
The reader might find it interesting to compare these results with our findings in \Cref{subsec:Dyck} where we study the context-free language $\langle \{0^n1^n\mid n\in\N\} \rangle$ and more complicated versions of Dyck languages.

\begin{lemma}\label{lem:0ast1astappear2}
    Let $G\in \mathcal{G}_{\langle 0^*1^*\rangle}$.  Then there exists a $w\in V^{\ast}$ such that $G=G(\langle 0^*1^*\rangle,w)$ and $\vert w\vert_v \geq 2$ for each $v\in V$.
\end{lemma}
\begin{proof}
    To simplify the notation $L\coloneqq \langle 0^*1^*\rangle$.
    Let $G=(V,E)\in \mathcal{G}_{L}$. Then there exists $w\in V^{\ast}$ with $G=G(L,w)$. Assume there exists a $v\in V$ with $\vert w\vert_v < 2$. We want to construct a $w'\in V^{\ast}$ with $G=G(L,w')$, $\vert w'\vert_v \geq 2$ and $\vert w'\vert_u = \vert w\vert_u$ for each $u\in V\setminus\{v\}$.

    Assume $\vert w\vert_v=1$. Then there exist $w_1,w_2\in (V\setminus \{v\})^{\ast}$ such that $w = w_1 v w_2$. Define $w' = w_1 vv w_2$ and $G'=(V,E') =  G(L,w')$. Clearly, for $s,t \in V\setminus \{v\}$, $\{s,t\}\in E$ \iffl $\{s,t\}\in E'$, as $h_{s,t}(w)=h_{s,t}(w')$. Let $u\in V\setminus \{v\}$. As $h_{v,u}$ is a homomorphism, $h_{v,u}(w') = h_{v,u}(w_1)h_{v,u}(vv)h_{v,u}(w_2) =h_{v,u}(w_1)00h_{v,u}(w_2)$ while $h_{v,u}(w) = h_{v,u}(w_1)h_{v,u}(v)h_{v,u}(w_2) =h_{v,u}(w_1)0h_{v,u}(w_2)$. Hence, $\{v,u\}\in E'$ \iffl $\emptyword \in \{h_{v,u}(w_1),h_{v,u}(w_2)\}$. The same holds for $\{v,u\}\in E$. Thus, $G=G'$.
\end{proof}

\begin{theorem}\label{thm:0ast1ast}
    $\cG_{\langle 0^*1^*\rangle}=\cG_{\langle 0011\rangle}$.
\end{theorem}
\begin{proof}
    To simplify our notation, $L_1\coloneqq{\langle 0^*1^*\rangle}$ and $L_2\coloneqq \langle 0011 \rangle$. 
    Let $G =(V,E)\in \cG_{L_2}$. Then by \Cref{lem:uniform_word} there is a word $w\in V^{\ast}$ such that $G=G(L_2, w)$ and $\vert w\vert_v=2$ for all $v\in V$.  
    Hence, $h_{v,u}(w)\in \langle0011,0101,0110\rangle$ for all $v,u\in V$ with $v \neq u$. As $\langle0011,0101,0110\rangle\cap L_1= L_2$, $G=G(L_1,w)$.

    Now assume $G=(V,E)\in G_{L_1}$. By \Cref{lem:0ast1astappear2} there exists a word $w \in V^{\ast}$ with $G=G(L,w)$ and $\vert w\vert_v\geq 2$ for all $v\in V$. Let $w=w_1\cdots w_\ell$ with $w_1,\dots, w_\ell\in V$. Let $l_v\coloneqq  \min\{i\in [\ell] \mid w_i = v \}$ and $r_v\coloneqq  \max\{i\in [\ell] \mid w_i = v \}$. Define $x_0\coloneqq \lambda$ and for $i\in [\ell]$ with $w_i = v$
    $$x_i\coloneqq \begin{cases}
        x_{i-1}v,& i \in \{l_v,r_v\}\\
        x_{i-1},& i \notin \{l_v,r_v\}
    \end{cases}.$$
    To simplify the notation $x= x_\ell$. So, $x$ contains exactly the first and the last occurrence of each symbol. Clearly $\vert x \vert_v=2$ for all $v\in V$. Let $G'=G(L_2,x)=(V,E')$ and $v,u\in V$ with $v \neq u$. Assume $\{v,u\}\in E$. Then $h_{v,u}(w) \in L_1$. With out loss of generality, there exists $a,b\in \mathbb{N}\setminus \{0,1\}$, such that $h_{v,u}(w)=0^a1^b$. Thus, $l_v<r_v<l_u<r_u$. Since we do not change the order of $v$ and $u$, $h_{v,u}(x)= 0011 \in L_2 $. Hence, $\{v,u\}\in E'$.

    Assume $\{v,u\}\in E'$. Without loss of generality let $h_{v,u}(x)=0011$. Therefore, for all $i,j \in [\ell]$ with $w_i=v$ and $w_j=u$, $l_v \leq i \leq r_v<l_u \leq j \leq r_u$. Thus, $h_{v,u}(w) \in L_1$.
\end{proof}
From \Cref{cor:co-intervalgraphs} as well as from \Cref{thm:0ast1ast} and \Cref{lem:0ast1astappear2}, we can infer:

\begin{corollary}
For any $i,j\leq 2$, the class $\cG_{\langle 0^i0^*1^*1^j\rangle}$ is the class of co-interval graphs.
\end{corollary}

\begin{lemma}\label{lem:0_0or1_ast_1appear2}
    Let $G=(V,E)\in \cG_{\langle 0\{0,1\}^{\ast}1 \rangle}$. Then there exists a $w\in V^{\ast}$ such that $G=G(\langle 0\{0,1\}^{\ast}1 \rangle,w)$  and $\vert w\vert_v \geq 2$ for each $v\in V$.
\end{lemma}
\begin{proof}
    To simplify the notation $L\coloneqq \langle 0\{0,1\}^{\ast}1 \rangle$.
    Let $G=(V,E)\in \mathcal{G}_{L}$. Then there exists $w\in V^{\ast}$ with $G=G(L,w)$. 

    Assume $\vert w\vert_v=1$. Then there exist $w_1,w_2\in (V\setminus \{v\})^{\ast}$ such that $w = w_1 v w_2$. Define $w' = w_1 vv w_2$ and $G'=(V,E')=G(L,w')$. Clearly, for $s,t \in V\setminus \{v\}$, $\{s,t\}\in E$ \iffl $\{s,t\}\in E'$, as $h_{s,t}(w)=h_{s,t}(w')$. Let $u\in V\setminus \{v\}$. As $h_{v,u}$ is a homomorphism, $h_{v,u}(w') = h_{v,u}(w_1)h_{v,u}(vv)h_{v,u}(w_2) =h_{v,u}(w_1)00h_{v,u}(w_2)$ while $h_{v,u}(w) = h_{v,u}(w_1)h_{v,u}(v)h_{v,u}(w_2) =h_{v,u}(w_1)0h_{v,u}(w_2)$. Hence, $\{v,u\}\in E'$ \iffl $\emptyword \in \{h_{v,u}(w_1),h_{v,u}(w_2)\}$. The same holds for $\{v,u\}\in E$. Thus, $G=G'$.
\end{proof}

\noindent
Hence, $\cG_{\langle 0\{0,1\}^{\ast}1 \rangle}=\cG_{\langle 0(0\shuffle 1)\{0,1\}^*1\rangle }$.

\begin{theorem}\label{thm:0_0or1_ast_1}
    $\cG_{\langle 0\{0,1\}^{\ast}1 \rangle}=\cG_{\langle 0011,0101 \rangle}$.
\end{theorem} 
\begin{proof}
    To simplify the notation $L_1\coloneqq \langle 0\{0,1\}^{\ast}1 \rangle$ and $L_2\coloneqq \langle 0011,0101 \rangle$. 
    Let $G =(V,E)\in \mathcal{G}_{L_2}$. Then by \Cref{lem:uniform_word} there is a word $w\in V^{\ast}$ such that $G=G(L_2, w)$ and $\vert w\vert_v=2$ for all $v\in V$.  
    Hence, $h_{v,u}(w)\in \langle0011,0101,0110\rangle$ for all $v,u\in V$ with $v \neq u$. As $\langle0011,0101,0110\rangle\cap L_1= L_2$, $G=G(L_1,w)$. 

    Now assume $G=(V,E)\in G_{L_1}$. By \Cref{lem:0_0or1_ast_1appear2} there exists a word $w \in V^{\ast}$ with $G=G(L,w)$ and $\vert w\vert_v\geq 2$ for all $v\in V$. Let $w=w_1\cdots w_\ell$ with $w_1,\dots, w_\ell\in V$. Let $l_v\coloneqq  \min\{i\in [\ell] \mid w_i = v \}$ and $r_v\coloneqq  \max\{i\in [\ell] \mid w_i = v \}$. Define $x_0\coloneqq \lambda$ and for $i\in [\ell]$ with $w_i = v$
    $$x_i\coloneqq \begin{cases}
        x_{i-1}v,& i \in \{l_v,r_v\}\\
        x_{i-1},& i \notin \{l_v,r_v\}
    \end{cases}.$$ 
    So, $x$ contains exactly the first and the last occurrence of each symbol. Now we want to show the following claim. 
    \begin{claim}
    For all $i\in [\ell]$ and distinct $v,u\in V$ with $\vert x_i\vert_{v} = 2 =  \vert x_i\vert_{u}$, $\{v,u\} \in E$ \iffl $h_{u,v}(x_i)\in L_2$. 
    \end{claim}
    \begin{proof}
        We show this by an inductive argument. Clearly this holds for $x_1$. The claim holds for a $i\in [\ell-1]$. For $i+1 \notin \{l_v,r_v \mid v\in V\}$, then $x_i=x_{i+1}$ and the claim holds. Assume there exists $v\in V$ with $i+1=l_v$. Then $\vert x_{i+1}\vert_v =1 $ and $\vert x_{i+1}\vert_u = \vert x_i\vert_u$. Thus $\{u\in V \mid \vert x_{i+1}\vert_2 \}=\{u\in V \mid \vert x_{i}\vert_2 \} $ and the claim holds by induction in this case. So this leaves to assume there is a $v\in V$ with $i+1=r_v$. As $l_v<r_v$, $\vert x_{i+1}\vert_v =2$ and $\{u\in V \mid \vert x_{i+1}\vert_2 \}=\{u\in V \mid \vert x_{i}\vert_2 \}\cup \{v\}$. Let $u\in V\setminus \{ v \}$ with $\vert x_{i+1}\vert_u = 2$. Hence, $\vert x_i\vert_u =2$ and $h_{v,u}(x_{i+1})\in  \{1100,1010,0110\}$. Assume $h_{v,u}(x_{i+1})=0110$. This equivalent to $l_v<l_u<r_u<r_v$. Because of maximality of $r_u$, $\vert w_{i+1}\cdots w_\ell\vert_u=0$ (so $r_u < r_v$). So $\{v,u\}\notin E$ \iffl $l_v<l_u$. Hence for $u\in V\setminus \{v\}$ with $\vert x_{i+1}\vert_u$, $\{v,u\} \in E$ \iffl $h_{u,v}(x_i)\in L_2$. The remaining cases follow directly by the inductive assumption. Therefore the claim holds by induction.\renewcommand{\qed}{\hfill$\Diamond$}
    \end{proof}
     Clearly $\vert x_{\ell} \vert_v=2$ for all $v\in V$. By the claim $G(L_2,x_{\ell})=G$.
\end{proof}
By \Cref{cor:co-permutation}, we can see:
\begin{corollary}
$\cG_{\langle 0\{0,1\}^{\ast}1 \rangle}$ and $\cG_{\langle 0(0\shuffle 1)\{0,1\}^*1\rangle }$ describe the class of permutation graphs.
\end{corollary}

Let us now return to interval graphs.
We will see next that also for this class of graphs, we can find more (finite and infinite) languages for its characterization.

\begin{theorem}\label{thm:interval-infinite}
For $k\in \N_{\geq 1}$, let $L_k=\langle 01\cdot (0^k\shuffle 1^k)\rangle$ and $L=\langle 01 \cdot (0^+\shuffle 1^+)\rangle$. Then, $\cG_{L_k}$ and $\cG_L$ are the interval graphs.
\end{theorem}

\begin{proof}We already know that $\cG_{L_1}$ are the interval graphs, see \Cref{thm:intervalgraphs}.

Given an interval graph~$G$, we can construct a word $w_k$ such that $G=G(L_k,w_k)$ as follows. For $k=1$, we refer again to  \Cref{thm:intervalgraphs}, thus defining $w_1$. For $k>1$, set $w_k=w_{k-1}01$. In other words, the question whether or not an edge $\{\ta,\tb\}$ exists in $G(L_k,w_k)$ is completely determined by the first two occurrences of $\ta$ and $\tb$ in~$w_k$. 

Conversely, consider some $G=(V,E)$ that is $L_k$- (or $L$-) represented by some  word~$w\in V^*$. First, observe that we can represent any isolated vertex $v\in V$ by taking $v^{2} \cdot h_{V\setminus\{v\}}(w)$ instead of~$w$, so that we can assume that $w\in (I^{2})^*\cdot (V\setminus I)^+$ where $I$ collects all isolated vertices, such that $w=w_0\cdot w_1$ and $w_1$ does not start with $\tb\tb$ for any $\tb\in V$ (and $|w_1|_{\ta}=k+1$ for all $\ta\in V\setminus V_0$ if $w$ $L_k$-represents~$G$). Now,
consider a finite automaton with output that gets $w$ as an input, reads it letter by letter, and produces an output~$w_1$ such that it outputs~$\ta$ when reading the first and second occurrence of~$\ta$, but it outputs the empty word when reading any further occurrence of~$\ta$.
This means that $w'$ is 2-uniform. Moreover, if $\{\ta,\tb\}\in E$, then $h_{\ta,\tb}(w)\in L_k=\langle 01\cdot (0^k\shuffle 1^k)\rangle$, so that $h_{\ta,\tb}(w')\in L_1=\langle 01\cdot (0\shuffle 1)\rangle=\langle0101,0110\rangle$, i.e., $\{\ta,\tb\}\in E(G(L_1,w'))$. Conversely, if $\{\ta,\tb\}\in E(G(L_1,w'))$, i.e., $h_{\ta,\tb}(w')\in L_1$, then these two first occurrences of $\ta$ and of $\tb$ must have a reason provoked by a subsequence of $k+1$ occurrences of $\ta$ and of $\tb$ in $w$, i.e., $h_{\ta,\tb}(w)\in L_k$. This proves that $G=G(L_1,w')$. By \Cref{thm:intervalgraphs}, $G$ is an interval graph. The case of $L$-representability is treated similarly.
\end{proof}

\subsection{Consequences of this section for general questions}

In this section, we collect some results that could also be seen as counterexamples to conjectures that one could suggest when first looking at our framework.

\begin{corollary} \label{cor:noSubsetInclusionImplied_1} In general, 
$\cG_{L_1} \subseteq \cG_{L_2}$ neither implies $L_1 \subseteq L_2$ nor $L_2 \subseteq L_1$.
\end{corollary}
\begin{proof} 
Consider $L_1 = \langle 01 \rangle$ and $L_2 = 0^2\shuffle 1^2$. According to \Cref{exa:represented-graphs-CompleteUnionNullGraphs} of \Cref{lem:00shuffle11}, $\cG_{L_2} = \{ K_n \cup N_m \mid n,m \in \mathbb{N}\}$. According to \Cref{exa:represented-graphs} this set also equals $\cG_{L_1}$. This shows $\mathcal{G}_{L_1} = \mathcal{G}_{L_2}$, but $L_1$ and $L_2$ are disjoint. 
\end{proof}

\begin{corollary} \label{cor:noSubsetInclusionImplied_2} In general, 
$L_1 \subseteq L_2$ implies neither $\mathcal{G}_{L_1} \subseteq \mathcal{G}_{L_2}$ nor $\mathcal{G}_{L_2} \subseteq \mathcal{G}_{L_1}$.
\end{corollary}
\begin{proof} Clearly $\langle 0101\rangle\subseteq \langle 0101,0110\rangle$. By \Cref{prop:circlegraphs} and \Cref{thm:intervalgraphs} we know that $\cG_{\langle 0101\rangle}$ is the class of circle graphs and $\cG_{\langle 0101,0110\rangle}$ is the class of interval graphs. 

By \Cref{exa:represented-graphs-C4} of \Cref{exa:represented-graphs}, $C_4\in \cG_{\langle 0101\rangle}$, i.e., $C_4$ is a circle graph, 
but $C_4$ is not an interval graph, see \cite[p.~46]{LekBol62}. Hence the circle graphs are not a subset of the interval graphs. The interval graphs are not a subset of the circle graphs either, see \cite[p.~5]{CzeDurGra2001}. 
\end{proof}

This shows that the condition of \Cref{prop:inclusion} cannot be dropped. Also, it proves that interval graphs are not closed under taking subgraphs, a fact that is known but is now a trivial consequence.

Although we have already seen that inclusion relations on the language side need not have any influence on inclusion relations on the graph class side, the following conjecture might be tempting.

\begin{conjecture}\label{conj:three-languages-inclusion}
Let $L_1\subseteq L_2\subseteq L_3$ such that $\cG_{L_1}=\cG_{L_3}$. Then, $\cG_{L_1}=\cG_{L_2}=\cG_{L_3}$.
\end{conjecture}
However, even this type of conjecture is not true in general.

\begin{proposition}\label{prop:three-languages-inclusion}
There exist languages $L_1\subseteq L_2\subseteq L_3$ with $\cG_{L_1}=\cG_{L_3}$ but $\cG_{L_1}\neq \cG_{L_2}$.
\end{proposition}

\begin{proof}
Let $L_1=\langle 01\cdot (0\shuffle 1)\rangle$, $L_2'=\langle 01110,01101,01011,01100,01010,01001\rangle$, and $L_3'=\langle 01\cdot (0^2\shuffle 1^2)\rangle$. Furthermore, let $L_2=L_1\cup L_2'$ and $L_3=L_2\cup L_3'$.
By \Cref{thm:interval-infinite}, $L_1$ and $L_3'$ describe the class of interval graphs. By \Cref{thm:intervalBigraphs},  $\mathcal{G}_{L_2'}$ is the class of interval bigraphs.

We will now show that not only  $\mathcal{G}_{L_2'}$ but also  $\mathcal{G}_{L_2}$ contains the 4-cycle $C_4$, which is not an interval graph. To this end, consider $w=\ta\tb\td\ta\tc\tb\td\tc\tb\td$ that describes the graph $G(L_2',w)=(V,E)$. We see:
\begin{itemize}
    \item $h_{\ta,\tb}(w)=01011\in L_2$, i.e., $\{\ta,\tb\}\in E$,
    \item $h_{\ta,\tc}(w)=0011\notin L_2$, i.e., $\{\ta,\tc\}\notin E$,
    \item $h_{\ta,\td}(w)=01011\in L_2$, i.e., $\{\ta,\td\}\in E$,
    \item $h_{\tb,\tc}(w)=01010\in L_2$, i.e., $\{\tb,\tc\}\in E$,
    \item $h_{\tb,\td}(w)=010101\notin L_2$, i.e., $\{\tb,\td\}\notin E$,   
    \item $h_{\td,\tc}(w)=01010\in L_2$, i.e., $\{\tc,\td\}\in E$.
\end{itemize}
To complete the proof, we have to show that also $L_3$ represents the interval graphs. First, as any interval graph~$G$ can be represented by a 2-uniform word~$w$ with respect to $L_1$ and as $L_1\subseteq L_3$ and $L_3\cap (0^2\shuffle 1^2)=L_1$, $G=G(L_3,w)$. Conversely, consider any word $w\in V^*$ that represents $G(L_3,w)=(V,E)$.
We want to prove that $(V,E)$ is an interval graph.
As interval graphs are closed under adding isolated vertices, we can assume that $(V,E)$ has no vertices of degree zero. More formally, let $I$ collect all vertices of degree zero and consider $h_{V\setminus I}(w)$ instead of~$w$ in the following argument. Clearly, $w$ contains only letters with frequentnesses~2 or~3. 
Consider now the word $w'$ that is obtained from $w$ by deleting the third occurrence of every letter (if possible). We will show that $G(L_3,w')=(V,E)$. As obviously $G(L_3,w')=G(L_1,w')$ (because there are no isolated vertices), $(V,E)$ must be an interval graph.
In order to prove that $G(L_3,w)=G(L_3,w')$, we have to consider cases.
\begin{itemize}
    \item If $|w|_\ta=|w|_\tb=2$, then $\{\ta,\tb\}\in E$ $\iff$ $h_{\ta,\tb}(w)\in L_3$ $\iff$   $h_{\ta,\tb}(w)\in L_1$ $\iff$   $h_{\ta,\tb}(w')\in L_3$.
    \item If $|w|_\ta=|w|_\tb=3$, then $\{\ta,\tb\}\in E$ $\iff$ $h_{\ta,\tb}(w)\in L_3$ $\iff$   $h_{\ta,\tb}(w)\in L_3'$ $\iff$   $h_{\ta,\tb}(w')\in L_3$.   
    \item If $|w|_\ta=2$ and $|w|_\tb=3$, then $\{\ta,\tb\}\in E$ $\iff$ $h_{\ta,\tb}(w)\in L_3$ $\iff$   $h_{\ta,\tb}(w)\in L_2'$ $\iff$   $h_{\ta,\tb}(w')\in L_3$. Here, also observe that $L_2'=\langle 01110,01101,01011,01100,01010,01001\rangle=\{ 01110,01101,01011,01100,01010,01001,10001,10010,10100,10011,10101,10110\}$ and that the set $L_2''$ obtained from $L_2'$ by deleting any third occurrence of a letter (if possible) yields $L_2''=\{0110,0101,1001,1010\}=\langle 0110,0101\rangle$, i.e., $L_2''=L_1$.
\end{itemize}
This concludes presenting our counterexample to \Cref{conj:three-languages-inclusion}.
\end{proof}
}

\section{Palindromes and Other Well-known Languages}
\label{sec:palindromes}
\longversion{So far, we have taken the approach to look at very small language classes in a systematic fashion and study what graph classes they yield. }In this section, we turn our attention to rather `famous' (infinite) languages as the set of palindromes, the-copy language, the Dyck-language or the set of Lyndon words and look at their associated graph classes. We will see that most of these languages allow to describe all graphs, and this can be done in a rather efficient way. Illustrations of these representations can be found in \Cref{fig:languages-C4}.

\longversion{\subsection{Palindromes}}

We first turn to palindromes, i.e., words $w$ satisfying $w=w^R$. We first derive some simple properties. 
Notice that $h_{\ta,\tb}(w^R)=h_{\ta,\tb}(w)^R$ for some word $w\in\Sigma^{\ast}$ and $\ta,\tb\in\Sigma$. Our first lemma shows that being a palindrome can be characterized by the palindrome property of all projections to binary sub-alphabets. Besides being a nice word-combinatorial property, it also has a graph-theoretic interpretation, see \Cref{cor:palindromes-and-complete-graphs}. We will see similar properties for other well-known languages below.

\begin{lemma}\label{lem:binary-palindromes}
A word $w\in\Sigma^{\ast}$ is a palindrome \iffl, for all $\ta,\tb\in\Sigma$ with $\ta\neq\tb$, $h_{\ta,\tb}(w)$ is a palindrome.
\end{lemma} 
\begin{proof}
W.l.o.g., we can assume $|\Sigma| > 1$ because for $|\Sigma| = 1$, the statement holds trivially. 
Let $w$ be a palindrome and $\ta,\tb\in\Sigma$. Choose $u\in\Sigma^{\ast}$ and $v\in\Sigma\cup\{\emptyword\}$ with $w=uvu^R$.
We have that $h_{\ta,\tb}(w)=h_{\ta,\tb}(u)h_{\ta,\tb}(v)h_{\ta,\tb}(u)^R$ which is a palindrome.\longversion{

}
Let now be all projections onto two letters be palindromes. If $w[1]\neq w[|w|]$, then $h_{\{w[1],w[|w|]\}}(w)$ would not be a palindrome and thus we have $w[1]=w[|w|]$. Consider now $w[2..|w|-1]$. By the same argument, we get $w[2]=w[|w|-1]$. Inductively, we get that $w$ is a palindrome.
\end{proof}

Let $L_\mathcal{P}=\{u\in\{0,1\}^+\mid u=u^R\}$ be the set of all binary palindromes. This language is $0$-$1$-symmetric. The previous lemma links the set of all palindromes to  $L_\mathcal{P}$. 

\begin{corollary}\label{cor:palindromes-and-complete-graphs}
A word~$w$ is a palindrome \iffl $G(L_\mathcal{P},w)\simeq K_{\alph(w)}$.
\end{corollary}

We will first look at binary patterns that are both palindromes and alternating words in the sense of classical word-representability, which is formalized in our framework by $L_{\wrep}$.

\begin{theorem} \label{thm:alternating_palindromes}
    $\cG_{L_{\mathcal{P}}\cap L_{\wrep}}$ is the class of bipartite graphs. 
\end{theorem} 
\begin{proof}
Let $G=(V,E) \in \cG_{L_{\mathcal{P}}\cap L_{\wrep}}$. Clearly, $L_{\mathcal{P}} \cap L_{\wrep} = 1(01)^\ast \cup (01)^\ast 0 \cup \{ \emptyword \}$. Consider $\{ u, v \} \in E$. Obviously, $h_{\{u,v\}}(w) \in v(uv)^+ \cup (uv)^+ u$ because $h_{u,v}(w) \in L_{\mathcal{P}} \cap L_{\wrep}$. If $h_{\{u,v\}}(w) \in v(uv)^+$, $|w|_v = |w|_u +1$ and if $h_{\{u,v\}}(w) \in u(vu)^+$, $|w|_u = |w|_v +1$. Hence, $E \subseteq \{ \{u, v \} \mid \exists \ell \in \mathbb{N}_{\geq 1}: |w|_v=\ell \wedge |w|_u=\ell+1\}$. Define $A\coloneqq\{v\in V\mid \exists \ell\in \mathbb{N}_{\geq 1}: |w|_v=2\ell\}$ and $B\coloneqq\{v\in V\mid \exists \ell\in \mathbb{N}_{\geq 1}: |w|_v=2\ell-1\}$. There is no edge between any two vertices of $A$ (or $B$), that is, $G$ is bipartite. \longversion{Also confer the argument of \Cref{prop:bipartite}.}

    Let $G=(V,E)$ be a bipartite graph with the classes $A=\{a_1,\dots,a_s\}$ and $B=\{b_1,\dots,b_t\}$. For each $i\in [s]$, define $v_i\coloneqq a_1\cdots a_{i-1}a_{i+1}\cdots a_s$, $x_i$ as the word which enumerates $N(a_i)$ in linear order of indices of $B$, and $y_i$ as the word which enumerates $B\setminus N(a_i)$ in linear order of the indices of $B$, giving $u_i\coloneqq x_ia_iy_i$; define that $v_0=a_1\cdots a_s$, then,
    \longversion{$$w\coloneqq v_0u_1v_1\cdots u_sv_s\,.$$}\shortversion{$w\coloneqq v_0u_1v_1\cdots u_sv_s$.}
    
    Define $G'=(V,E')\coloneqq G(L_{\mathcal{P}}\cap L_{\wrep},w)$. Clearly, for\longversion{ all} $j,k\in [s]$ with $j<k$, $h_{a_j,a_k}(u_iv_i)=01$ if $i\in [s]\setminus\{j,k\}$, $h_{a_j,a_k}(u_jv_j)=01$ and $h_{a_j,a_k}(u_kv_k)=10$, so totally, $h_{a_j,a_k}(w)=(01)^k10(01)^{s-k}\notin L_{\mathcal{P}}\cap L_{\wrep}$, so there is no edge between vertices of $A$. For\longversion{ all} $j,k\in [t]$ with $j\neq k$, $h_{b_j,b_k}(w) \in (0\shuffle 1)^s$. Hence, $h_{b_j,b_k}(w) \notin L_{\mathcal{P}}\cap L_{\wrep}$ because it has even length, so there is also no edge between vertices of~$B$.
    First consider $\{a_j,b_k\}\notin E$. Then $h_{a_j,b_k}(v_{j-1}u_j)=001$ is an infix of $h_{a_j,b_k}(w)$. Hence, $h_{a_j,b_k}(w)\notin L_{\mathcal{P}}\cap L_{\wrep}$, i.e., $\{a_j,b_k\}\notin E'$.
    Secondly, consider $\{a_j,b_k\}\in E$. For $i\in [s]\setminus \{j\}$, $h_{a_j,b_k}(u_iv_i)=10$. Also $h_{a_j,b_k}(u_jv_j)=10$. Hence, $h_{a_j,b_k}(w)=0(10)^s\in L_{\mathcal{P}}\cap L_{\wrep}$, so that $\{a_j,b_k\}\in E'$. Thus, $G=G'$.
\end{proof}

Let $\cG_{\star}$ be the set of all extended star graphs, i.e., $\cG_{\star}=\{K_{1,n}\cup N_m\mid n,m\in\N\}$.
By \cite{FleHLN2024}, the non-empty palindromes represent exactly the extended star graphs in the following sense:
$$\cG_{\star}=\{G(L_{\wrep},w)\mid w\text{ is a palindrome}\,\}\,.$$ \longversion{

}By \Cref{thm:alternating_palindromes}, we see\shortversion{: this}\longversion{ that this seemingly tiny logical} tweak from $\cG_{\star}$ to $\cG_{L_{\mathcal{P}}\cap L_{\wrep}}$  makes a big difference.

\begin{corollary}$\cG_{\star}\subsetneq \cG_{L_{\mathcal{P}}\cap L_{\wrep}}$.
\end{corollary} 

But $\cG_{L_{\mathcal{P}}}$ itself is much richer. This can be already observed by noticing that $w=w_1w_2$ represents the graph union of $G(L_{\mathcal{P}},w_1)$ and $G(L_{\mathcal{P}},w_2)$. Together with  \Cref{lem:binary-palindromes}, this already means that $\cG_{L_{\mathcal{P}}}$ contains all cluster graphs, a proper superclass of $\cG_{\star}$ and clearly not bipartite in general. But in fact, the representability goes far beyond, as we show next.

\begin{theorem}\label{thm:palindromes-get-all}
  The binary palindrome language  $L_{\mathcal{P}}$ can represent every graph. 
\end{theorem}
\begin{proof}
    Let $G=([n],E)$ be a graph with $n\in \N$. Define, for $i\in [n]$, $u_i\in [i-1]^*$ as the word which enumerates $ \overline{N(i)}\cap [i-1]$ in linear order as well as \shortversion{$w_1\coloneqq 11$ and $w_i\coloneqq iu_iw_{i-1}iu_i^R$ if $i\in\N_{\geq 2}$.}\longversion{    $$w_i\coloneqq \begin{cases}11,& i=1,\\
    iu_iw_{i-1}iu_i^R, & i\neq 1
    \end{cases}.$$} We want to show that $h_{j,k}(w_i) = h_{j,k}(w_i)^R$ \iffl $\{j, k\}\in E$ for all $i\in [n]$ and $j,k\in [i]$ with $j\neq k$. 
    For $i=1$ this trivially holds. So let $i>1$. Since $h_{j,k}$ is a morphism, 
    \begin{equation}
        \begin{split}
            h_{j,k}(iu_i w_{i-1} iu_i^R)^R &= h_{j,k}(u_i^R)^R h_{j,k}(i)^R h_{j,k}(w_{i-1})^R h_{j,k}(u_i^R) h_{j,k}(i)^R\\ &= h_{j,k}(u_i) h_{j,k}(i) h_{j,k}(w_{i-1})^R h_{j,k}(u_i^R) h_{j,k}(i)\\
            &=h_{j,k}(u_i iw_{i-1}^R u_i^Ri)\,.
        \end{split}
    \end{equation}
    First we assume $i\notin \{j,k\}$. Then $h_{j,k}(w_i)= h_{j,k}(u_i) h_{j,k}(w_{i-1}) h_{j,k}(u_i^R)$. This equals $h_{j,k}(w_i^R) = h_{j,k}(u_i) h_{j,k}(w_{i-1})^R h_{j,k}(u_i^R)$ \iffl $h_{j,k}(w_{i-1})^R= h_{j,k}(w_{i-1})$. By induction, we know this is the case \iffl $\{j,k\}\in E$. 
    Now we can assume $j=i$. If $\{i,k\}\notin E$, then $h_{i,k}(w_i)= h_{i,k}(u_i)\cdot h_{i,k}(i)\cdot h_{i,k}(w_{i-1})\cdot h_{i,k}(u_i^R)\cdot h_{i,k}(i) = 1\cdot 0\cdot h_{i,k}(w_{i-1})\cdot 1\cdot 0 $ as $k<i$, which is no palindrome. For $\{i,k\}\in E$, $h_{i,k}(w_i)=0 h_{i,k}(w_{i-1})0$. This is a palindrome as $h_{i,k}(w_{i-1})\in 1^+$. By induction, $G\simeq G(L_\cP,w_n)$ follows.
\end{proof}
Also, we can find two further characterization results.

\begin{corollary}\label{cor:palindromes-yield-split-and-bipartite}
There are languages related to the palindrome language that characterize the class of split graphs and the class of cobipartite graphs.    
\end{corollary}

\begin{proof}
By \Cref{thm:alternating_palindromes}, we get:\begin{itemize}
\item $(L_\cP\cap L_{\wrep})\cup ((0^2)^+\shuffle (1^2)^+)$ turns one partition side of a bipartite graph represented by $(L_\cP\cap L_{\wrep})$ into a clique, so that we represent all split graphs this way.
    \item Using \Cref{prop:compl}, we see that $\overline{L_{\cP}}\cup \overline{L_{\wrep}}$ represents all cobipartite graphs.   \qed
\end{itemize}\renewcommand{\qed}{}
\end{proof}

\longversion{\subsection{Copy-language}}

Now, we are looking at the copy-language  $L_{\text{copy},\Sigma}\coloneqq\{ww\mid w\in\Sigma^*\}$.
Apart from the fact that the language of palindromes is context-free, while the copy-language is not, both languages enjoy quite similar properties. We exemplify this by an analogue to \Cref{lem:binary-palindromes}.\longversion{ To formulate it, let us call the elements of $L_{\text{copy},\Sigma}$ \emph{copy-words}\longversion{. These are}\shortversion{,} also known as \emph{repetitions}\longversion{ in combinatorics on words}. We follow a presentation similar to the\longversion{ one on} palindromes, starting with a morphic characterization\longversion{ lemma}.}
\begin{lemmarep}\label{lem:binary-copies} \shortversion{$(*)$} For every alphabet $\Sigma$ such that $|\Sigma|\geq 2$, a word $w\in\Sigma^{\ast}$ is a copy-word \iffl, for all $\ta,\tb\in\Sigma$ with $\ta\neq\tb$, $h_{\ta,\tb}(w)$ is a copy-word.
\end{lemmarep}
\begin{proof}
Consider a copy-word $ww$ over $\Sigma$. For each $\ta \in \Sigma$ and $\tb \in \Sigma \setminus \{ \ta \}$, $h_{\ta, \tb}(ww) = h_{\ta, \tb}(w) h_{\ta, \tb}(w)$ is a copy-word. Now consider a word $v \in \Sigma^\ast$ such that $h_{\ta, \tb}(v)$ is a copy-word for every $\ta \in \Sigma$ and $\tb \in \Sigma \setminus \{ \ta \}$. $|h_{\ta, \tb}(v)|_0$ and $|h_{\ta, \tb}(v)|_1$ are even for every $\ta \in \Sigma$ and $\tb \in \Sigma \setminus \{ \ta \}$. Hence, $|v|_{\ta}$ is even for every $\ta \in \Sigma$ and $v$ has even length. 
Let $\ta$ be the first letter of $v$. Hence, $u,u' \in \Sigma^\ast$ exist such that $v = \ta u \ta u'$ and $|u|_{\ta} = |u'|_{\ta}$. For every $\tb \in \Sigma \setminus \{ \ta \}$, $|u|_{\tb} = |u'|_{\tb}$ because, otherwise, $h_{\ta,\tb}(\ta u \ta u')$ would not be a copy-word. Hence $|u|_{\tc} = |u'|_{\tc}$ for every $\tc \in \Sigma$. Assume $v = x \td y x \td' y'$ for $\td, \td' \in \Sigma$ and $x, y, y' \in \Sigma^\ast$ such that $|y|_{\tc} = |y'|_{\tc}$ for every $\tc \in \Sigma$. Assume $\td \neq \td'$. 
$h_{\td, \td'}(v) = h_{\td, \td'}(x) \cdot 0 \cdot h_{\td, \td'}(y) \cdot h_{\td, \td'}(x) \cdot 1 \cdot h_{\td, \td'}(y')$ with $|y|_{\td} = |y'|_{\td}$ and $|y|_{\td'} = |y'|_{\td'}$. Since $h_{\td, \td'}(v)$ is a copy-word, $\td = \td'$. This is a contradiction. Hence, $\td = \td'$ and $v = x \td y x \td y'$ with $|y|_{\tc} = |y'|_{\tc}$ for every $\tc \in \Sigma$. By inductive reasoning, we get that $v$ is a copy-word.
\end{proof}

Let $L_{\mathcal{C}}=\{ww \mid w\in \{0,1\}^{\ast}\}=L_{\text{copy},\{0,1\}}$ denote the binary copy-language which is $0$-$1$-symmetric. 
Similarly to our treatment of the palindromes, we can infer from \Cref{lem:binary-copies}:

\begin{corollary}\label{cor:copy-language-clique}
A word~$w$ is a copy-word \iffl $G(L_\mathcal{C},w)\simeq K_{\alph(w)}$.
\end{corollary}
Again, $\cG_{L_{\mathcal{C}}}$ is much richer, as we show next. The construction is even simpler than before.

\begin{theorem}\label{thm:copy-words-get-all}
The binary copy-language  $L_{\mathcal{C}}$ can represent every graph. 
\end{theorem}

\begin{proof}
    Let $G=([n],E)$ be a graph. Define $w = u_1 1 u_2 2 \cdots u_n n 1 u_1 \cdots n u_n$ where $u_i$ enumerates each vertex of $\overline{N(i)}\cap [i]$ exactly once in the natural order for $i\in [n]$. Consider $G'=([n],E')=G(L_{\mathcal{C}},w)$. 
    Let $i,j\in [n]$ with $i<j$. By definition, $h_{i,j}(u_i)= 0$. Also, $h_{i,j}(u_{j})\in \{\emptyword,0\}$ with $h_{i,j}(u_{j})= \emptyword$ being equivalent to $i\notin N(j)$. For $k<i$, $h_{i,j}(u_k)= \emptyword$. Hence, $h_{i,j}(w)= w_1w_2$ with  \begin{eqnarray*}w_1&\coloneqq&h_{i,j}(u_i)0h_{i,j}(u_{i+1})\cdots h_{i,j}(u_{j})1h_{i,j}(u_{j+1})\cdots h_{i,j}(u_{n})\text{ and}\\w_2 &\coloneqq& 0h_{i,j}(u_i)h_{i,j}(u_{i+1})\cdots 1 h_{i,j}(u_{j})h_{i,j}(u_{j+1})\cdots h_{i,j}(u_{n})\,.\end{eqnarray*} It is easy to see that $\vert w_1 \vert = \vert w_2 \vert$. Furthermore, $w_1=w_2$ holds \iffl $1 h_{i,j}(u_{j}) = h_{i,j}(u_{j})1$. Since $h_{i,j}(u_{j})\in \{\emptyword,0\}$, $1 h_{i,j}(u_{j}) = h_{i,j}(u_{j})1$ \iffl $h_{i,j}(u_{j})= \emptyword$. As $h_{i,j}(u_{j})= \emptyword$ \iffl $i\notin N(j)$, we can conclude $E=E'$ and $G=G'$.
\end{proof}

A reader estimating the complexity of a language by its situation within the Chomsky hierarchy might consider the copy-language as quite complicated. However, notice that by \Cref{prop:compl}, we get another characterization of all graphs through complementation, and the complement of the copy-language is a one-counter language, \longversion{i.e., it is in }a subclass of context-free.
\begin{corollary}
The complement of the  binary copy-language  $L_{\mathcal{C}}$ can represent every graph. 
\end{corollary}

So far, we investigated the copy language, i.e., we only looked at a very specific case of what is called a repetition in the field of combinatorics of words. A word $w\in\Sigma^{\ast}$ is called a {\em repetition of order $t$} or simply a {\em $t$-repetition} for some $t>1$ if there exists $u\in\Sigma^{\ast}$ with $w=u^t$. If no such $u$ and $t>1$ exists $w$ is called {\em primitive.}
Regarding arbitrary repetitions, one can observe that for some $w\in\Sigma^{\ast}$ and $i\in\N_{\geq 1}$, $w^i$ and $w$ $L_{\wrep}$-represent the same graph \iffl for all $\ta,\tb\in\Sigma$ we have $h_{\{\ta,\tb\}}(w)\in \ta\{\ta,\tb\}^{\ast}\tb\cup\tb\{\ta,\tb\}^{\ast}\ta$. Suppose for a contradiction that for some $\ta,\tb\in\Sigma$ we have that $h_{\{\ta,\tb\}}(w)$ starts and end with the same letter, w.l.o.g., with~$\ta$, and that $\ta,\tb$ alternate in~$w$. Then, they do not alternate in $h_{\{\ta,\tb\}}(w^2)$. On the other hand, if $\ta,\tb$ do not alternate in~$w$ (no matter whether or not they start and end in the same letter in $h_{\{\ta,\tb\}}(w)$), then they do not alternate in $w^2$, either. The same argument leads to the following proposition.

\begin{proposition}\label{prop:rep}
Let $L=\langle 0\{0,1\}^{\ast}1\rangle$ (see  \Cref{thm:0_0or1_ast_1}). 
If $w$ $L$-represents~$G$ then also $w^i$ $L$-represents $G$ for all $i\in\N_{\geq 1}$.
\end{proposition}

\Cref{prop:rep} establishes that primitive words can be seen as representatives (of the equivalence class which contains all repetitions of a given word) for $L$-representing graphs. Notice that they are the smallest elements w.r.t. to the word length. This insight encourages to look at another famous class of words which are primitive by definition: the Lyndon words.

\longversion{\subsection{Lyndon words}}

Another well-studied class of words are the Lyndon words~\cite{Lyn54}\longversion{ that have been introduced in the context of group theory}. To fix notation, let $\Sigma$ be an alphabet with strict linear ordering~$\prec$. 
Also, recall that it was shown already in \cite{KitPya2008} that if $w$ represents a graph~$G$ of order $n\geq 3$ in the classical sense, then every word in the conjugacy class of $w$ also represents~$G$, i.e., apart from trivial cases, for representing any word-representable graphs, Lyndon words are sufficient. Graphs are special in the following sense: since Lyndon words are by definition non-repetitive there does not exist a Lyndon word in which both letters alternate. 
As we focus on 0-1-symmetric languages and as not both $xy$ and $yx$ can be Lyndon for any two different letters $x,y$, we clearly have to look at the symmetric closure of binary Lyndon words. Therefore, define $L_{\mathcal{L}}\coloneqq\langle L_{\text{Lyndon},\{0,1\},<}\rangle$ with the ordering $0<1$. Observe that $L_{\mathcal{L}}=L_{\text{Lyndon},\{0,1\},<}\cup L_{\text{Lyndon},\{0,1\},>}$. 
The results for the copy-language imply that repetitions cannot occur in projections when dealing with Lyndon words. Therefore, we can, w.l.o.g., focus on the property of being the smallest conjugate in its class during the investigation of Lyndon words.


Not all renaming projections $h_{\ta,\tb}(w)$ of a Lyndon word $w$ have to be Lyndon: we have for the Lyndon word $w=\ta\tb\tc\ta\tb\tc\ta\tc$ that \longversion{$h_{\{\ta,\tb\}}(w)=\ta\tb\ta\tb\ta$, i.e.,  }$h_{\ta,\tb}(w)=01010$, is not Lyndon. 
Hence, we only get\shortversion{:}\longversion{ a weaker projection lemma.}

\begin{lemmarep} \shortversion{$(*)$}
Let $w\in\Sigma^{\ast}$ such that $h_{\ta,\tb}(w)$ is $<$-Lyndon for all $\ta, \tb\in\Sigma$ with $\ta  \prec \tb$. Then $w$ is $\prec$-Lyndon.
\end{lemmarep}

\begin{proof}
Suppose that there exist $u,v\in\Sigma^+$ such that $w=uv$ and $vu\prec uv$. 
If $v$ were not a prefix of $u$, there existed $x\in\Sigma^{\ast}$ and some letters $\ta\prec \tb$ with $u=x\tb$ and $v=x\ta$. Thus, we have $w=x\tb x\ta$ and therefore $h_{\ta,\tb}(x)$ is not Lyndon. If $v$, on the other hand, were a prefix of~$u$, there existed $x_1\in\Sigma^+$ with $u=vx_1$. Thus, we have by our assumption $vvx_1\prec vx_1v$ and therefore $vx_1\prec x_1v$. Inductively, we obtain that $v=\ta$ and $u=\ta x_1\cdots x_k$ for some $k\in\N$. This leads to $w=\ta x_1\cdots x_k\ta$ and therefore $h_{\ta,\tb}(w)$ is not Lyndon for any other\longversion{ letter} $\tb\in\Sigma$. 
\end{proof}

\begin{lemmarep}\label{lem:Lyndon_observation} \shortversion{$(*)$}
Let $\ta,\tb\in\Sigma$ with $\ta \neq \tb$ and $w \in \Sigma^{\ast}$.
\begin{enumerate}
    \item $h_{\ta,\tb}(w)$ is a $<$-Lyndon word 
    \iffl $h_{\tb,\ta}(w)$ is a $>$-Lyndon word. 
    \item If $h_{\ta,\tb}(w)$ is a $<$-Lyndon word 
    then $h_{\ta,\tb}(w)$ is no $>$-Lyndon word. 
\end{enumerate}
\end{lemmarep}

\begin{proof}
    \begin{enumerate}
        \item This statement is obvious as $h_{\ta,\tb}(w)[i]=0$ \iffl $h_{\tb,\ta}(w)[i]=1$ for $i\in [\vert h_{\tb,\ta}(w)\vert]\longversion{=[\vert h_{\ta,\tb}(w)\vert]}$.
        \item If $h_{\ta,\tb}(w)$ is a $<$-Lyndon word, there are $i,j\in [\vert h_{\ta,\tb}(w)\vert]$ such that $h_{\ta,\tb}(w)[i]=0$ and $h_{\ta,\tb}(w)[j]=1$. Especially, $h_{\ta,\tb}(w)[1]=0$. Otherwise, the conjugate starting at $i$ would be $<$-smaller than $h_{\ta,\tb}(w)$. This implies that the  conjugate starting at~$j$ is $>$-smaller than $h_{\ta,\tb}(w)$.\qed
    \end{enumerate}\renewcommand{\qed}{}
\end{proof}

\begin{theorem}\label{thm:Lyndon-words-get-all}
    $\cG_{L_\cL}$ includes every graph.
\end{theorem}
\begin{proof}
    Let $G=([n],E)$ be an graph with $n\in \mathbb{N}$. Define for $i\in [n]$ the words $v_i:=i\cdots n$, $x_i:=j_{i,1} \cdots j_{i,\ell_i}$, $y_i:=k_{i,1} \cdots k_{i,\ell'_i}$ and $u_i:=i^2 x_i i^2 y_i$ where $N(i) \cap [i..n] = \{j_{i,1},\dots, j_{i,\ell_i}\}$, $\overline{N(i)} \cap [i..n] = \{i,k_1,\dots, k_{i,\ell'_i}\}$ and $j_{i,1},\dots, j_{i,\ell_i}$ as well as $k_{i,1},\dots, k_{i,\ell'_i}$ are strictly monotone increasing. 
    We will show $G=G'=(V,E'):=G(L_{\mathcal{L}},w)$\longversion{ with $$w\coloneqq 1^3 2^3\cdots n^3 v_1^2 u_1 v_2^2 u_2\cdots v_n^2 u_n\,. $$}\shortversion{, $w\coloneqq 1^3 2^3\cdots n^3 v_1^2 u_1 v_2^2 u_2\cdots v_n^2 u_n$.} 

    Let $p,q,r\in [n]$ with $p<q$. By \Cref{lem:Lyndon_observation} only $h_{p,q}(w),h_{q,p}(w)$ with respect to\longversion{ the order} $0< 1$ needs to be considered. Since $p$ appears before $q$ in $w$, it is enough to consider $h_{p,q}(w)$.
    
    Observe that for $q<r$,  $h_{p,q}(v_r u_r) = \emptyword$. Hence, $h_{p,q}(v_{q+1}u_{q+1} \cdots v_{n}^2u_{n}) = \emptyword$. For $p<r<q$, $h_{p,q}(v_r^2u_r) = 111$ and $h_{p,q}(v_q^2u_q) = 111111$. Thus,  $h_{p,q}(v_{q}u_{q} \cdots v_{n}^2u_{n}) \in 111111(111)^*$.

    Next, assume $r<p<q$. Then $h_{p,q}(v_i)=01$. If $p,q\in N(r)$, $h_{p,q}(x_r)=01$ and $h_{p,q}(y_r)=\emptyword$. Analogously, $h_{p,q}(x_r)=\emptyword$ and $h_{p,q}(y_r)=01$ for $p,q\notin N(r)$. In the case of $p\in N(r)$ and $q \notin N(r)$, $h_{p,q}(x_r)=0$ and $h_{p,q}(y_r)=1$. For $p\notin N(r)$ and $q \in N(r)$, $h_{p,q}(x_r)=1$ and $h_{p,q}(y_r)=0$. Therefore, $h_{p,q}(v_r v_r u_r)\in \{010101,010110\}$ and $h_{p,q}(v_1^2u_1\cdots v_{p-1}^2 u_{p-1})\in \{010101,010110\}^{p-1}$. Clearly, there does not exist an $i \in [6(p-1)]$ with $w[3n+i..3n+i+2]=000$. 
    Finally, consider $r=p$. If $\{p,q\}\in E$ then $h_{p,q}(v_p^2,u_q)=010100100$ and $h_{p,q}(w)\in 000111\{010101,010110\}^{p-1}010100100111111(111)^*.$ 
    Clearly, $i=1$ is the only $i\in \vert h_{p,q}(w)\vert$ such that $h_{p,q}(w)[i..i+2]=000$. Since $h_{p,q}(w)[\vert h_{p,q}(w)\vert]=1$, $h_{p,q}(w)$ is Lyndon and $\{p,q\}\in E'$.

    Assume $\{p,q\}\notin E$. Then $h_{p,q}(v_pp^2 x_p p^2 y_p)=0100001$. As $v_p u_p= v_pp^2 x_p p^2 y_p$ is an infix of~$w$, $0100001$ is an infix of $h_{p,q}(w)$. As $000111$ is a prefix of $h_{p,q}(w)$, $h_{p,q}(w)$ is no Lyndon word (a conjugate with prefix $h_{p,q}(u_p)$ is $<$-smaller). Thus, $\{p,q\}\notin E$. 
    In total, \longversion{$E=E'$ and }$G=G'$. 
\end{proof}
We can use this idea to construct a language $L$ that represents exactly all bipartite graphs.

\begin{theoremrep}\label{thm:bipartite_Language}
\shortversion{$(*)$}  \longversion{Let} $L\coloneqq \{w\in L_\cL\mid |w| \text{ is odd}\}$\longversion{ be the set of Lyndon words of odd length}. Then $\mathcal{G}_L$ is the class of bipartite graphs.
\end{theoremrep}
\begin{proof}
    Assume $G=(V,E)$ is an $L$-representable graph. Then there exists a $w\in V^*$ such that $G=G(L,w)$. Define $A\coloneqq\{v\in V\mid \exists k\in \N :\: \vert w\vert_v =2k\}$ and $B\coloneqq\{v\in V\mid \exists k\in \N :\: \vert w\vert_v =2k + 1\}$. Clearly, there are no edges $\{a_1,a_2\}\in E \cap \binom{A}{2}$ as $h_{a_1,a_2}(w)$ is even. Analogously, there is no edge $\{b_1,b_2\}\in E \cap \binom{B}{2}$. Hence, $G$ is bipartite.

    Let $G=(V,E)$ be a bipartite graph with the classes $A=\{a_1,\dots,a_s\},B=\{b_1,\dots,b_t\}$. Let $y_i$ be the word that enumerates each element in $N(a_i)$ for $i\in [s]$ exactly once in order of the indices of the elements in $B$. Define \longversion{$$w\coloneqq a_1^3\cdots a_{s}^3 b_1^2 \cdots b_t^2 v_1\cdots v_s$$}\shortversion{$w\coloneqq a_1^3\cdots a_{s}^3 b_1^2 \cdots b_t^2 v_1\cdots v_s$} with $x_i \coloneqq a_i^2 y_i^2 a_i^2$ for $i\in [s]$. Clearly, for all $i\in [s]$ and $v\in V$, $\vert x_i\vert_v\in \{0,2,4\}$. Hence, $\vert w\vert_a$ is odd for $a\in A$ and $\vert w\vert_b$ is even for $b\in B$.
    
    Define $G'=(V,E'):=G(L,w)$.
    Since $L$ only includes words of odd length, there are no edges $\{a_1,a_2\}\in E' \cap \binom{A}{2}$ and $\{b_1,b_2\}\in E' \cap \binom{B}{2}$. 
    Hence, $G'$ is a bipartite graph with the classes $A,B$. Let $i\in [s]$ and $j \in [t]$.\longversion{

   } Assume $\{a_i,b_j\}\notin E$. Then $h_{a_i,b_j}(v_i)=0000$ is an infix and $h_{a_i,b_j}(a_1^3 \cdots a_s^3 b_1^2 \cdots b_t^2)=00011$ is a prefix of $h_{a_i,b_j}(w)$. Hence, $h_{a_i,b_j}(w)$ is no Lyndon word and $\{a_i,b_j\} \notin E'$.\longversion{

   } Assume $\{a_i,b_j\}\in E$. For $k \in [s] \setminus \{ i \}$, $\vert v_k\vert_{a_s} =0$ which implies $h_{a_i,b_j}(v_k)\in \{\emptyword,11\}$. Furthermore, $h_{a_i,b_j}(v_k) = 001100$. In total $h_{a_i,b_j}(w)\in 00011(11)^{\ast}001100(11)^{\ast}$. Therefore, $h_{a_i,b_j}(w)$ is Lyndon and $ \{a_i, b_j\}\in E$. Thus, $G=G'$. 
\end{proof}
\longversion{
Together with \Cref{thm:graph-decomposition}}\shortversion{Similar to \Cref{cor:palindromes-yield-split-and-bipartite}}, we obtain:

\begin{corollary}
 Let $L\coloneqq \{w\in L_\cL\mid |w| \text{ is odd}\}$\longversion{ be the set of Lyndon words of odd length}, $L_{\text{odd}}:=\{w\in \{0,1\}^{\ast}\mid \vert w\vert_0,\vert w\vert_1 \text{ are odd}\}$ and $L_{\text{even}}:=\{w\in \{0,1\}^{\ast}\mid \vert w\vert_0,\vert w\vert_1 \text{ are even}\}$. Then  $\cG_{L \cup L_{\text{odd}}}$ is the class of split graphs and  $\cG_{L \cup L_{\text{odd}} \cup L_{\text{even}}}$ is  the class of  cobipartite graphs.
\end{corollary}
\longversion{\subsection{Dyck-language}\label{subsec:Dyck}}
Next, we consider the Dyck-language.\footnote{The language $\{w\in \{0,1\}^* \mid \vert w\vert_0 =\vert w\vert_1 \wedge \forall i\in [\vert w\vert ]:\vert w[1..i]\vert_0 \leq \vert w[1..i]\vert_1 \}$ is also known as the restricted Dyck-language with one pair of parentheses, see the textbook~\cite{Ber79}, and would be written as $D_1^{\prime\ast}$ in the notation of that book. The equivalence to our definition is clarified in \cite[Exercise 3.6]{Ber79}. More generally, $D_k^{\prime\ast}$ is the set of correctly parenthesized words formed with $k$  different types of parentheses, forming an alphabet $\Sigma$ of size~$2k$. Clearly, if $w\in D_k^{\prime\ast}$, then $h_{(,)}(w)\in L_{\mathcal{D}}$ for any pair of parentheses $(,)\in\Sigma$. However, the converse is not true, as $w=([)]$ shows.}

Define the symmetric hull of the Dyck-language as $$L_{\mathcal{D}} := \langle \{w\in \{0,1\}^* \mid \vert w\vert_0 =\vert w\vert_1 \wedge \forall i\in [\vert w\vert ]:\vert w[1..i]\vert_0 \leq \vert w[1..i]\vert_1 \}\rangle\,.$$ 
$L_{\mathcal{D}}$ can be viewed as the language of all well-bracketed expressions, interpreting $0$ as the opening and $1$ as the closing bracket, or vice versa.

\begin{theorem}
    $\cG_{L_{\mathcal{D}}}$ is the class  of comparability graphs. 
\end{theorem}
\begin{proof}
    Let $G=(V,E)$ be a comparability graph with respect to the strict order $\prec$ on~$V$, i.e., $\{u,v\}\in E$ \iffl $u,v$ are comparable with respect to~$\prec$. Define the upper set $U_\prec(v)\coloneqq\{u\in V \mid v\prec u\}$ for $v\in V$. $\prec_{lin}$ denotes an arbitrary but fixed linear order on $V$ that extends~$\prec$. Define $z\coloneqq v_1\cdots v_n$ to be an enumeration of $V$ in the order $\prec_{lin}$\longversion{, i.e., $\vert V\vert = n$}.

    For $v\in V$, let $x_v,y_v$ be words such that $x_v$ enumerates $U_\prec(v)$ and $y_v$ enumerates $\overline{U_\prec(v)}\setminus \{v\}$ according to the order~$\prec_{lin}$ (increasingly). Further, denote $z_v\coloneqq y_v v x_v$, i.e., $z_v$ first enumerates the elements of~$V$ that are in the upper set of~$v$, then $v$, and then  the remaining elements of~$V$.
    Define $w\coloneqq z z_{v_1}\cdots  z_{v_n}$ and $G'=(V,E')=G(L_{\mathcal{D}},w)$. Observe that for each $v,u\in V$, $\vert z_v\vert_u=\vert z \vert_u = 1$. Hence, for all $v,a,b\in V$, $h_{a,b}(z_v),h_{a,b}(z)\in \{01,10\}$. 
    
    Let $a,b\in V$ with $a\prec_{lin}b$ in our following discussion (without loss of generality). Thus, $h_{a,b}(z) = 01$ $(*)$.
    
    Assume $\{a,b\} \in E$. As $G$ is a comparability graph, this implies $a\prec b$ $[*]$, as $b\prec a$ is impossible because $a\prec_{lin}b$. We want to show that $h_{a,b}(z_v) =01$ for all $v\in V$. Then, together with $(*)$, we can conclude $h_{a,b}(w) =(01)^{n+1}$. Hence, $h_{a,b}(w)\in L_{\mathcal{D}} $. For the sake of contradiction, assume that there exists a $v \in V$ such that $h_{a,b}(z_v) =10$. Hence, $b$ precedes $a$ in the enumeration~$z_v$. If $v=a$ (or $v=b$, respectively), then $b\in \alph(y_v)=\overline{U_\prec(v)}\setminus \{v\}$ (or $a\in \alph(x_v)=U_\prec(v)$, respectively), contradicting $[*]$. Therefore, $v\notin \{a,b\}$. Since $x_v$, $y_v$ enumerate in accordance with~$\prec_{lin}$ and as $a\prec b$, $b$ has to appear in $y_v$ and $a$ in $x_v$. This is a contradiction to the transitivity of $\prec$, as $a\prec b$ and $v \prec a$ (because 
    $a\in U_\prec(v)$) would imply $b\in U_\prec(v)=\alph(x_v)$. Hence, $h_{a,b}(z_v) =01$ for all $v\in V$. 

    Now assume $\{a,b\} \notin E$. Then $\vert y_a\vert_b =1$ and $h_{a,b}(z_a) =10$. Let $a=v_i$. This implies $h_{a,b}(w) \in 01\{01,10\}^{i-1}10\{01,10\}^{n-i}$. Then $\{a,b\}\notin E'$, since $\vert h_{a,b}(w[1..2i+1])\vert_0 = i < i+1 = \vert h_{a,b}(w[1..2i+1])\vert_1$ and $\vert h_{a,b}(w[1])\vert_1 = 0 < 1 = \vert h_{a,b}(w[1])\vert_0$.

    Conversely, let $G=(V,E)\in \cG_{L_{\mathcal{D}}}$. Then, there exists a $w\in V^{\ast}$ such that $G=G(L_{\mathcal{D}},w)$. Define the order $\prec_V$ such that, for all $v,u\in V$, $v \prec_V u$ \iffl $v\neq u$, $\vert w\vert_v= \vert w\vert_u$ and, for all prefixes $w'$ of $w$, $\vert w'\vert_v\leq \vert w'\vert_u$. We show that $v \prec_V u$ is a strict partial order and that $\{v,u\}\in E$ \iffl $v\prec_V u$. Hence, $G$ is a comparability graph.

    By definition, $\prec_V$ is irreflexive. Let $a,b,c\in V$ with $a \prec_V b$ and $b \prec_V c$. Hence, $\vert w\vert_a= \vert w\vert_b = \vert w \vert_c $. For any prefix $w'$ of $w$, $\vert w \vert_a \leq \vert w \vert_b \leq \vert w \vert_c$. Let $i\in [\vert w\vert]$ be the the smallest~$i$ with $w[i]=b$. This implies $\vert w[1..i-1] \vert_a \leq 0 =\vert w[1..i-1] \vert_b$ and $\vert w[1..i] \vert_a \leq 0 < 1 =\vert w[1..i] \vert_b \leq \vert w[1..i]\vert_c$. Thus, $a\neq c$ and $a\prec_V c$. Therefore, $\prec_V$ is a strict partial order.

    Let $a,b\in V$. For $\vert w\vert_a\neq  \vert w\vert_b$, $\{a,b\} \notin E$ and $a\not\prec_V b$ as well as $b \not\prec_V a$. Thus, we can assume $\vert w\vert_a =  \vert w\vert_b$.  Since $h_{a,b}$ is a morphism, if $w'$ is a prefix of $w$, then $h_{a,b}(w')$ is a prefix of $h_{a,b}(w)$. Therefore, $\vert w'\vert_a =  \vert h_{a,b}(w')\vert_0$ and $\vert w'\vert_b =  \vert h_{a,b}(w')\vert_1$. 

    By definition,  $a \not \prec_V b$ and $b \not \prec_V a$ \iffl there are prefixes $w',w''$ such that $ \vert h_{a,b}(w')\vert_1 = \vert w'\vert_b < \vert w'\vert_a =  \vert h_{a,b}(w')\vert_0$ and $ \vert h_{a,b}(w'')\vert_0 = \vert w''\vert_a < \vert w''\vert_b  = \vert h_{a,b}(w'')\vert_1$.  Thus, $\{a,b\}\notin E$. 

    Conversely, assume $\{a,b\} \notin E$. Hence, there are prefixes $h',h''$ of $h_{a,b}(w)$  with $\vert h'\vert_1 < \vert h'\vert_0$ and $ \vert h''\vert_0 <  \vert h''\vert_1$. 
    Since $\vert h_{a,b}(v)\vert \leq 1$ for all $v\in V$, we can show by an inductive argument that, for each prefix $h$ of $h_{a,b}(w)$, there exists a prefix $w'$ of $w$ such that $h=h_{a,b}(w')$. Again by definition, this means that $a \not\prec_V b$ and $b \not \prec_V a$.
\end{proof}

With \Cref{prop:compl}, we get a characterization of cocomparability graphs in terms of complements of the symmetric hull of the Dyck language. In typical introductory courses on theoretical computer science, as the simplest variation of a Dyck language, $\{0^n1^n\mid n\in\N\}$  is considered as the simplest non-regular language, so here, we consider $L_1\coloneqq \langle\{0^n1^n\mid n\in\N\}\rangle$, and often, also $L_2\coloneqq \{w\in\{0,1\}^*\mid |w|_0=|w|_1\}$ is looked at as another non-regular language. For defining graph classes, these two languages yield the following results:
\begin{propositionrep} \shortversion{$(*)$}
$\cG_{L_1}$ is the class of graph unions of co-interval graphs.\footnote{Co-interval graphs are not closed under graph union, as, e.g., $2K_2$ is not a co-interval graph.}
$\cG_{L_2}$ is the class of cluster graphs.
\end{propositionrep}

\begin{proof}
Let $G=(V,E)$ be a co-interval graph. This means that we can associate to $V$ a family of intervals $\{I_v\mid v\in V\}$ such that $G$ can be viewed as the comparability graph of the corresponding interval order on~$V$, i.e., $u<v$ \iffl the rightmost endpoint of~$I_u$ is to the left of the leftmost endpoint of~$I_v$. Without loss of generality, we can assume that all endpoints are pairwisely distinct and all are integers from $[2|V|]$. Then, we can associate a word $w\in V^{2|V|}$ to $G$ by setting $w[i]=\ta$ (for $i\in [2|V|]$) \iffl one endpoint of $I_\ta$ is at position~$i$. Now, $G\simeq G(L_1,w)$ is easy to check. If we replace the first occurrence of every $\ta\in V$ by $\ta^k$ in~$w$, we arrive at a word $w_k$ of length $(k+1)|V|$. One can see that  $G\simeq G(L_1,w_k)$ for each positive integer~$k$. 
If we start with a union of $n$ co-interval graphs, yielding a graph~$H$, we can hence find a word $x_i$ of length $(i+1)|V_i|$ to describe component $H_i=(V_i,E_i)$. Define $x\coloneqq x_1\cdots x_n$.
As vertices occurring in $x_i$ do not occur in any other part of~$x$ and as only such vertices have frequentness~$(i+1)$, $H\simeq G(L_1,x)$ follows.

Conversely,  let $w\in \{0,1\}^*$.  Let $V=\alph(w)$. Consider $G_1=G(L_1,w)=(V,E_1)$.  Let $\ta,\tb\in V$ be two vertices and look at $h_{\ta,\tb}(w)$. This binary word belongs to $L_1$ if it is either $0^n1^n$ or $1^n0^n$ for some~$n$. In particular, this means that vertices of different frequentnesses in~$w$ are not connected. 
For each frequentness, however, we can build a separate interval model by defining the interval $I_\ta$ of $\ta$ by the position of the leftmost occurrence of $\ta$ in $w$ as the left endpoint and the position of the rightmost occurrence of $\ta$ as the right endpoint of~$I_\ta$. In this interval model, $\{\ta,\tb\}$ is an edge if and only if the intervals $I_\ta$ and $I_\tb$ are disjoint.
Hence, the whole graph $G_1$ is in fact the graph union of a number of co-interval graphs.

Let $G=(V,E)$ be a cluster graph, i.e., $G=\bigcup_{j=1}^k (V_k,E_k)$, where $G_k=(V_k,E_k)\simeq K_{n_k}$ for some positive integers $k$ and $n_1,\dots,n_k$. Accordingly, let $V_j=\{v_{j,1},\dots,v_{j,n_j}\}$ for $j\in [k]$. Define $w_j\coloneqq v_{j,1}^j\cdots v_{j,n_j}^j$ and $w\coloneqq w_1\cdots w_k$. 
Then, it is clear that $G=G(L_1,w)=G(L_2,w)$.

Conversely, let $w\in \{0,1\}^*$.  Let $V=\alph(w)$. Consider $G_2=(L_2,w)=(V,E_2)$. Let $\ta,\tb\in V$ be two vertices. If $|w|_\ta\neq|w|_\tb$, then $\{a,b\}\notin E_2$. Otherwise, $h_{\ta,\tb}(w)\in L_2$, so  $\{a,b\}\in E_2$.
\end{proof}
By \cref{prop:compl}, we get a characterization of complete multipartite graphs (as complements of cluster graphs) as $\cG_{\overline{L_2}}$.


\longversion{\subsection{Counting Bits}}
\shortversion{\section{Counting Bits}}

The languages that we discussed above, most of which enable us to describe any graph, are not the first ones observed: for instance, $L_{\overline{1^3}}$ is such a language, as described in \cite{JonKPR2015}, see our discussions above.
However, there is one drawback of such an encoding: the word~$w(G)$ describing a given graph $G=(V,E)$ will have length $\cO(|V|!)$, so that we need an exponential number of bits to describe a graph, measured in its order. Even worse is the situation with 1112-representable graphs, see \cite[Thm. 3.12]{GaeJi2020}, where the recursively constructed words can have length  $\cO(2^{|V|^2})$.
With 2-11-representable graphs, a more concise representation of every graph was proven in \cite[Thm. 5.2]{CheKKKP2018}, using only  $\cO(n^3\log(n))$ many bits.
This is still worse than using traditional adjacency lists or adjacency matrices.
In particular, the latter data structure needs $\cO(n^2)$ bits to describe any (un)directed graph of order~$n$.
The good news is that the graph representations that we obtain with the help of palindromes, copy-words or Lyndon words are better. Taking an example of $C_4$, palindromes and copy languages performs much better, as illustrated in \Cref{fig:languages-C4}. Let us phrase this \longversion{as a formal statement}\shortversion{formally}.

\begin{figure}
\begin{minipage}[t]{0.2\textwidth}
 \begin{tikzpicture}
    \tikzset{every node/.style={fill = white,circle,minimum size=0.05cm}}
    
            \node[draw] (x1) at (0,0) {$1$};
            \node[draw] (x2) at (1,0) {$2$};
            \node[draw] (x3) at (1,-1) {$3$};
            \node[draw] (x4) at (0,-1) {$4$};
    
            \path (x1) edge[-] (x2);
            \path (x2) edge[-] (x3);
            \path (x3) edge[-] (x4);
            \path (x4) edge[-] (x1);
    \end{tikzpicture} 
\end{minipage}
\begin{minipage}[t]{0.8\textwidth}
\begin{tabular}[width=.5\textwidth]{| l |  l |}
\hline 
$L$ & $w$ \\
\hline 
 $L_{\overline{1^3}}$ \cite{JonKPR2015} & $2213212431124112341234$ \\

$L_{1112}$ \cite{GaeJi2020} & $221324124311241312341234$ \\
$L_{2\text{-}11}$ \cite{CheKKKP2018}  & $14323412413224314231243143124312$\\
\hline
$L_{\cP}$ \scalebox{.68}{(Thm.~\ref{thm:palindromes-get-all})} & $423121123142$  \\
$L_{\mathcal{C}}$ \scalebox{.68}{(Thm.~\ref{thm:copy-words-get-all})}  & $121324123142$ \\
$L_{\mathcal{L}}$ \scalebox{.68}{(Thm.~\ref{thm:Lyndon-words-get-all})}\negthinspace  & $111222333444123412341124113234234223224343433433444444$ \\
\hline
\end{tabular} 

    \end{minipage}
    \caption{Language $L$ represents $C_4$ with the word $w$. }
    \label{fig:languages-C4}
\end{figure}

\begin{corollary}
With the help of palindromes, copy-words or Lyndon words, each graph of order~$n$
can be represented with $\cO(n^2\log(n))$ many bits. By \Cref{prop:compl}, similar statements are true for the non-palindromes, the non-copy-words and the non-Lyndon words.
\end{corollary}

Keep in mind that the $\log$-factor has to be added, as writing down a single letter of the alphabet~$V$ needs $\cO(\log(|V|))$ many bits.
Notice that this matches other implicit graph representations, e.g., in terms of sum graphs, see~\cite{FerGaj2023}, or also that of adjacency lists\longversion{, a very prominent example of explicit graph representations}.
For sparse graphs, it is also important to take the number of edges into account.
For instance, for adjacency lists, one needs $\cO((n+m)\log n)$ many bits.
Next, we show that a slight modification of our constructions yields a graph representation that matches this bound.

\begin{theorem}\label{thm:sparse-encoding}
There exists a language~$L$ such that any graph~$G=(V,E)$, with $|V|=n$ and $|E|=m$, can be represented as $G(L,w)$ for a word~$w\in V^*$ of length $\cO(n+m)$, so that $w$ can be easily encoded as a binary string of size $\cO((n+m)\log n)$.
\end{theorem}

\begin{proof}
Consider $L=\overline{L_{\mathcal{C}}}$. Describe $G=(V,E)$ (of order~$n$) by the word~$w$ that is used to describe the complement of~$G$ according to the proof of \Cref{thm:copy-words-get-all}. By \Cref{prop:compl}, $G$ is isomorphic to $G(L,w)$. Now, notice that~$w$ contains $2n$ `vertex letters' and twice each of the words $u_i$ that can be described as enumerating all vertices of $N_G(i)\cap [i]$ (due to graph complement). Hence, the length of $u_i$ is bounded by the degree of vertex~$i$. By simple double counting, $|u_1\cdots u_n|\in\cO(m)$, where $m=|E|$, so that the claims follow immediately.   
\end{proof}

This theorem is also interesting from a language-theoretic perspective.
While the copy-language is not context-free, its complement is even a 1-counter language. Hence, there is no need to go very far in the Chomsky hierarchy to obtain such an efficient graph encoding.  Conversely, we neither have any arguments why regular languages have to lead to essentially bigger graph encodings. We leave this as an open question.

Conversely, it might be interesting to have encodings in hand that are suitable for storing dense graphs, as opposed to sparse graphs as in \Cref{thm:sparse-encoding}.
By the proof of this theorem, the language of copy-words $L_{\mathcal{C}}$ seems to be very suitable. For a concrete example, reconsider \Cref{cor:copy-language-clique}: for a vertex set~$V$ or size~$n$, $ww$ encodes $K_n$ if $w$ enumerates $V$ in some order.

\section{Conclusions and Outlook}
\label{sec:conclusions}

We have initiated a systematic study of word-representable graphs from a formal language perspective. So far, we focused on structural properties and discussed information-theoretic arguments that allowed us to describe the possibilities and limitations of certain languages as basis for word-representability.
With the idea of using such implicit representations for graph databases, it would be good to get a better understanding of how succinct concrete classes of graphs can be represented that are expected to occur and that need to be stored. 
For classical word-representability, such line of research was commenced in~\cite{GaeJi2020,SriHar2024}.
Additionally, text compression methods might be useful now for graph compression. Can interesting questions be answered for such compressed graphs without (completely) decompressing their representation? With this motivation, one can also study the  family of all binary languages representing all graphs.

We also started to identify (many) graph classes by means of certain formal languages.\shortversion{ We are currently working on a systematic study of graph classes that can be described by finite languages.} It is well-known that certain graph problems that are \NP-hard on general graphs but polynomial-time solvable or at least fixed-parameter tractable on special graph classes. Can we get a uniform presentation or understanding of these tractability results, based on language representations?
\longversion{Or, to formulate a concrete question: }What is the family of \longversion{binary }languages defining graph classes on which a largest clique of any graph from this class can be found in polynomial time?

One of the next things that we like to do is to exploit the connections between words and graphs in a more algorithmic fashion. The basic algorithmic question related to a graph class~$\cG$ is the \emph{recognition problem}: given a graph~$G$, is $G\in\cG$? Can we make use of the fact that $\cG$ can be given by a language~$L$ that in turn might be given, say, by a finite automaton? Is there a uniform decision procedure capturing the recognition problem of many graph classes at once? We have seen a result in this direction in \Cref{thm:deciding-bounded-treewidth-and-degeneracy}.

\longversion{Let us indicate another direction with an example.
Assume that $w=abcdadc$ represents some graph~$G$. Observe that $a,c,d$ occur twice in~$w$ and that $c,d$ only occur next to each other. Consider now $w'=abeae$ obtained from~$w$ by replacing the factors (blocks) $cd$ and $dc$ by~$e$. Independent of the language $L$ that would define $G=G(L,w)$, vertex~$e$ is neighbor of a vertex $x\in\{a,b,d\}$ in~$G'=G(L,w')$ \iffl $x$ is neighbor of 
$c$ (or of $d$) in~$G$. In other words, $\{c,d\}$ forms a module in~$G$ (see \cite{Gal67,Moh85}), and one could try to form a modular decomposition tree based on the block structure of the word~$w$. This might not be optimal with respect to all graphs in $\cG_w=\{ G(L,w)\mid L\subseteq\{0,1\}^*\}$ but still gives a uniform model for solving algorithmic problems like finding largest independent sets in  $\cG_w$. Observe that going from $G'$ to $G$ in our example means ``adding some twins'' in the sense of \Cref{thm:twins}. Notice that modular decompositions of graphs are useful for solving algorithmic problems, technically speaking parameterized by the width of a modular decomposition. Therefore, this type of questions is related to the discussion on algorithms as mentioned above.} 

\longversion{In order to get a better understanding of the graph classes that we are able to describe, it would also make a lot of sense to continue systematically exploring further finite languages and their graph classes. Initial research indicates that many more graph classes that can be defined by geometric intersection (etc.) models can be also described by appropriately chosen formal languages. }

\bibliography{ab,hen}

\section{Graph Theory and Graph Classes}
\label{sec:graphtheory}

\subsection{Advanced Graph-Theoretic Notions}

Most of these notions can be found in any modern textbook on Graph Theory, e.g., in \cite{Die2000}.

\subsubsection{Tree Deompositions and Treewidth}
  A \emph{tree decomposition} of a graph $G$ is a pair $\mathcal{T}=(T, \{X_t\}_{t\in V(T)})$, where $T$ is a tree whose node $t$ is assigned a vertex subset $X_t\subseteq V(G)$, called a bag, satisfying the following:
    \begin{itemize}
        \item [1)] $\bigcup_{t\in V(T)}X_t=V(G)$;
        \item [2)] for each edge $uv\in E(G)$, there is some node $t$ of $T$ such that $u\in X_t, v\in X_t$;
        \item [3)] for each $u\in V(G)$, the set $T_u=\{t\in V(T)\mid u\in X_t\}$ induces a connected subtree of $T$.
    \end{itemize}

The \emph{width} of tree decomposition $\mathcal{T}$ is given by $\max_{t\in V(T)}|X_t|-1$. The \emph{treewidth} of a graph $G$ is the minimum width over all tree decompositions of $G$, denoted by $\text{tw}(G)$.
\subsubsection{Degeneracy}

A graph is \emph{$d$-degenerate} if every subgraph has a vertex of degree at most~$d$. As explained in the next subsection, every planar graph is $5$-degenerate and every graph embeddable on a torus is $6$-degenerate, to give two concrete examples.

\subsubsection{Genus}
In topology, the genus of an (orientable) surface is discussed. For instance, the surface of a sphere has genus~0, while the surface of a torus has genus~1. This translates to graph theory by requiring that a graph has genus at most~$g$ if it can be drawn without crossings on a surface of genus~$g$, i.e., more technically speaking, that can be embedded on a surface of genus~$g$. Observe that subgraphs of graphs of genus~$g$ have genus (at most)~$g$. Planar graphs have genus~0.
Such an embedding partitions the surface into so-called faces. Now, if the embedded graph has $v$ vertices, $e$ edges and $f$ faces, then Euler's formula ascertains that $v-e+f=2-2g$. As $2e\geq 3f$ by a simple double-counting argument, the inequality $e\leq 3v + 6(g-1)$ follows. If a graph has minimum degree~$\delta$, then $\delta v\leq 2e$. Therefore, $\delta\leq 6 + \frac{6(g-1)}{v}$. Hence, planar graphs always have a vertex of degree at most five, graphs of genus~$1$ always have a vertex of degree  at most six, graphs of genus~2 always have a vertex of degree  at most six (because you can take $v\geq 8$ in the formula, because graphs with at most~7 vertices have maximum degree at most six), etc. 

\subsection{Graph Classes}

The following subsections define the important non-trivial graph classes important for this paper in alphabetical order. For each of these graph classes, we know that it is either representable by some language, or non-representable by any language.

\subsubsection{Bipartite Graphs}

A graph $G=(V,E)$ is \emph{bipartite} if there are two disjoint independent sets $I_1,I_2$ such that $V=I_1\cup I_2$, i.e., $G[I_1]$ and $G[I_2]$ are null graphs.
This means that $G$ is 2-colorable.
In line with this last characterization, we also consider $N_1$ as a bipartite graph.
Bipartite graphs are sometimes also called \emph{bigraphs}.

\subsubsection{Bipartite Chain Graphs}

Let $G=(V,E)$ be a bipartite graph with the partition classes $A,B \subseteq V$. $G$ is called a \emph{bipartite chain graph} if there is a linear ordering $\leq_A$ on $A$ such that for each pair $a_1,a_2 \in A$, $a_1 \leq_A a_2$ implies $N(a_1) \subseteq N(a_2)$. This leads to an linear ordering $\leq_B$ on $B$ such that $b_1\leq_B b_2$ implies $N(b_1)\supseteq N(b_2)$ for all $b_1,b_2\in B$. For more information on this graph class, we refer to~\cite{Yan82}.

\begin{remark}\label{rem:bipartite-chain-convex}
By definition, every bipartite chain graph is convex.
\end{remark}

\begin{remark}
    According to graphclasses.org,  bipartite chain graphs can be characterized as bisplit graphs, i.e., their vertex sets can be partitioned into a biclique and an independent set.
\end{remark}

\begin{remark}
According to graphclasses.org, complements of bipartite chain graphs can be characterized as the  $(3K_1,C_4,C_5)$-free graphs.
\end{remark}

\subsubsection{Chordal Graphs}

A graph is a \emph{chordal graph} (sometimes called triangulated graph, see \cite{Gol2004l}) if each induced cycle in it has length~3.
Also chordal graphs have an interesting intersection model:

\begin{remark}
    \cite{Bun74,Gav74a,Wal78}
A graph is a chordal graph if it has an intersection model by subtrees of a tree.
\end{remark}

\subsubsection{Circle Graphs}

A graph is a \emph{circle graph} if it has an intersection model of chords in a circle. In the context of this paper, it is interesting that alternatively, it can be seen as the so-called alternance graph of a so-called double-occurrence (we would call it 2-uniform) cyclic word. Cyclic words can be seen as equivalence classes of the conjugation operation. So, given a word~$w$ in which each vertex (letter) of $G=A(w)$ occurs exactly twice, this word defines an edge between~$u$ and~$v$ in~$G$ \iffl $u$ and~$v$ alternate in~$w$.
More on this view can be found in \cite{Bou87a,Bou94}.
The class of circle graphs is identical to the class of overlap graphs.

\subsubsection{Cluster Graphs}

A \emph{cluster graph} is a graph union of complete graphs, i.e., each connected component of a cluster graph is a complete graph. In particular, complete graphs and null graphs are cluster graphs.

\subsubsection{Cobipartite Graphs}
A graph $G=(V,E)$ is \emph{cobipartite} if there are two disjoint cliques $C_1,C_2$ such that $V=C_1\cup C_2$, i.e., $G[C_1]$ and $G[C_2]$ are complete graphs.
This means that $G$ is the complement of a bipartite graph.

\subsubsection{Cographs}

The \emph{class of cographs} is the smallest class of graph containing $K_1$ and that is closed under graph union and graph join. For instance, cluster graphs are cographs.

\subsubsection{Cocomparability Graphs}

A graph  $G=(V,E)$ is a \emph{cocomparability graph} if it is the complement of a comparability graph. In other words, there exists a strict order~$\prec$ on~$V$ such that $\{u,v\}\in E$ \iffl $u$ and $v$ are not comparable in~$\prec$.

\subsubsection{Co-interval Graphs}

A graph complement of an interval graph is also known as a \emph{co-interval graph}. They can be alternatively viewed as comparability graphs of so-called interval orders. A partial order $\prec$ on a set $X$ is an \emph{interval order} if we can associate, to each $x\in X$, an interval $I_x$ such that $x\prec y$ \iffl the rightmost endpoint of $I_x$ is to the left of the leftmost endpoint of $I_y$.

\subsubsection{Comparability Graphs}

Recall that a strict order is an asymmetric (and hence irreflexive), transitive relation.
A graph $G=(V,E)$ is a \emph{comparability graph} if the exists a strict order $\prec$ on $V$ such that for all $v,u\in V$,  $\{u,v\}\in E$ \iffl $u,v$ are comparable with respect to~$\prec$, i.e., if $v \prec u$ or $u\prec v$ holds. Hence, $E$ is the symmetric closure of~$\prec$.

\subsubsection{Convex Graphs}

A bipartite graph $G=(V,E)$ with the classes $A,B$ is \emph{convex} if there exists a linear ordering $<_A$ on~$A$ such that for all $b\in B$, $N(b)$ is consecutively ordered with respect to $ <_A$, i.e., there is no $a\in A$ such that there are $a_1,a_2\in N(b)$ with $a_1<_A a <_A a_2$. Also see \cite{Tuc72}.

\subsubsection{Halfline Intersection Graphs}

A halfline (or one-sided unbounded interval) has only one endpoint~$e$ and then stretches infinitely to the left or to the right, i.e., it is either $(-\infty,e]$ or $[e,\infty)$. 
According to \cite{KlaPet87}, a graph $G=(V,E)$ is a \emph{halfline intersection graph} if there is a collection of halflines $\{I_v\}_{v\in V}$ such that $\{u,v\}\in E$ \iffl $I_u\cap I_v\neq \emptyset$.

\begin{remark}\label{rem:halflines} \cite{KlaPet87}
A graph is a halfline intersection graph \iffl it is chordal and cobipartite.
\end{remark}

\subsubsection{Interval Bigraphs}

A bipartite graph $G=(V,E)$ with the partition classes $A$ and $B$ is called an \emph{interval bigraph} if there is family $\mathcal{I} = \{ I_v =[l_v,r_v] \}_{v\in V}$ of closed intervals such that, for every $u \in A$ and $v\in B$, $I_u \cap I_v \neq \emptyset$ \iffl $\{ u,v \} \in E$.  Keep in mind that $I_u \cap I_v \neq \emptyset$ holds \iffl $\max\{l_v,l_u\}\leq \min\{r_v,r_u\}$.
This graph class was introduced in~\cite{HarKabMcM82}.
There are very close connections to interval digraphs~\cite{SenDRW89}, as explained in~\cite{Mul97}.
We also refer to \cite{DasSah2024}.

\subsubsection{Interval Graphs}

A (undirected) graph $G =(V,E)$ is an \emph{interval graph} \iffl $G$ has an interval representation. An interval representation of $G$ is a family $\mathcal{I} = \{ I_v \}_{v\in V}$ of open intervals over a totally ordered set $(X,\leq)$ such that for every $u,v \in V$ with $u \neq v$, $I_u \cap I_v \neq \emptyset$ \iffl $\{ u,v \} \in E$. An open interval~$I$ can be described by its two endpoints $x_\alpha(I)$ and $x_\omega(I)$ such that $I=\{x\in X\mid x_\alpha(I) <x<x_\omega(I)\}$, where $y<z$ means $x\leq z$ and $x\neq z$. We can also denote $I$ as $(x_\alpha(I),x_\omega(I))$. Interval representation of the same graph can be structurally quite different. For instance, on the unit interval $([0,1],\leq)$, a $K_3$ with vertex set $V=\{x,y,z\}$ can be represented as $I_x=I_y=I_z=(0,1)$ or as $I_x=(0,1)$, $I_y=(0.2,0.6)$ and $I_z=(0.4,0.8)$ or as $I_x=(0,0.5)$, $I_y=(0.2,0.6)$ and $I_z=(0.3,1)$.

\begin{remark} \label{rem:interval} \cite{Gol2004c}
A graph is an interval graph \iffl it is a chordal cocomparability graph.
\end{remark}

\subsubsection{Outerplanar Graphs}

A graph~$G$ is \emph{outerplanar} if there is a crossing-free embedding of~$G$ in the plane such that all vertices lie on the outer (unbounded) face. For instance, a $K_{2,2}$ or a $K_3$ is outerplanar, while a $K_{2,3}$ and a $K_4$ are not. Outerplanar graphs are also called $1$-outerplanar.
For $k\in\N_{\geq 2}$, a graph~$G$ is $k$-\emph{outerplanar} if there is a crossing-free embedding of~$G$ in the plane such that, if one removes all vertices~$U$ that  lie on the outer (unbounded) face, then the embedding of the  remaining graph $G'=G[V\setminus U]$ proves that $G'$ is $(k-1)$-outerplanar.

\subsubsection{Overlap Graphs}

A graph $G =(V,E)$ is an \emph{overlap graph} \iffl $G$ has an interval representation $\mathcal{I} = \{ I_v \}_{v\in V}$, such that $\{ u,v \} \in E$ \iffl $I_u$ and $I_v$ partially overlap, that is $I_u \cap I_v \neq \emptyset$ and neither interval contains the other. 

\begin{remark} \label{rem:overlap} \cite{Gol2004f}
The class of overlap graphs is exactly the class of circle graphs. 
\end{remark}

\subsubsection{Permutation Graphs}

The most illustrative way to define permutation graphs is again through an intersection model. $G=([n],E)$ is a \emph{permutation graph} if there exists a permutation (i.e., a bijection) $\pi:[n]\to [n]$ such that, if we draw the vertices $1,\dots,n$ on a line $L_1$ and $\pi(1),\dots,\pi(n)$ on another line $L_2$ that is parallel to $L_1$ and if we connect $i$ with $\pi(i)$ with a straight line called~$\ell_i$ (for $i\in [n]$), then $\{i,j\}\in E$ \iffl $\ell_i$ and $\ell_j$ intersect. Clearly, we also have $\{i,j\}\in E\iff (i-j)(\pi(i)-\pi(j))<0$, which can be taken as the definition for edges without referring to the intersection graph model.

However, from the intersection model definition, one can directly observe that permutation graphs are circle graphs.

\begin{remark} \label{rem:permutation} \cite{Gol2004b}
    A graph is a permutation graph 
\iffl both $G$ and its complement $\overline{G}$ are comparability graphs, or, in other words, if $G$ is both a comparability and a cocomparability graph.
\end{remark}


\subsubsection{Series-Parallel Graphs}

All \emph{series-parallel \underline{multi}-graphs} can be generated from a single loop by the following operations that can be applied repeatedly:
\begin{itemize}
    \item Series operation: Subdivide an edge, i.e., create a new vertex~$x$ and replace an edge $\{u,v\}$ by two edges $\{u,x\}$ and $\{x,v\}$;
    \item Parallel operation: Add an edge parallel to an existent one.
\end{itemize}

The graphs that can be obtained in this process are the \emph{series-parallel graphs}.

\subsubsection{Split Graphs}

A graph $G=(V,E)$ is a \emph{split graph} if its vertex set $V$ can be decomposed into two sets $V=I\cup C$ such that $G[I]$ is a null graph and $G[C]$ is a complete graph.

\begin{remark}\label{rem:split} \cite{Gol2004a} A graph is a split graph \iffl both $G$ and its complement $\overline{G}$ are chordal graphs.
\end{remark}

Also split graphs admit a nice characterization as intersection graphs, similar to chordal graphs or interval graphs.
\begin{remark} \cite{McMShi83}
A graph is a split graph \iffl it is the intersection graph of subtrees of a star graph $K_{1,n}$.
\end{remark}


\subsubsection{Threshold Graphs}
Also this class has many characterizations, see \cite{Gol2004e}. 
To justify its name: $G=(V,E)$ is a threshold graph \iffl there is a function $c:V\to\mathbb{N}$ and a threshold $t\in \mathbb{N}$ so that $uv\in E$ iff $c(u)+c(v)>t$.
\begin{remark}A graph that is both a split and a cograph is a threshold graph.\end{remark}

\end{document}

\newpage
\input{questions}
\input{ideas}

\end{document}